\documentclass[11pt, oneside]{amsart}
\usepackage[marginparwidth=2cm]{geometry}                
\usepackage{graphicx, slashed}
\usepackage{amssymb, amsthm, mathrsfs, bm}
\usepackage{epstopdf}
\usepackage{enumerate}
\usepackage{color}
\usepackage{hyperref}
\newcommand{\be}{\begin{equation*}}
\newcommand{\ee}{\end{equation*}}
\newcommand{\ben}[1]{\begin{equation}\label{#1}}
\newcommand{\een}{\end{equation}}
\newcommand{\bea}{\begin{eqnarray}}
\newcommand{\eea}{\end{eqnarray}}
\newcommand{\bean}{\begin{eqnarray*}}
\newcommand{\eean}{\end{eqnarray*}}

\newcommand{\R}{\mathbb{R}}
\newcommand{\T}{\mathbb{T}}
\newcommand{\J}{\mathbb{J}}

\newcommand{\p}{\partial}

\renewcommand{\r}{{\overline{r}}}

\renewcommand{\O}[1]{\mathcal{O}\left( #1 \right)}

\newcommand{\N}{\mathbb{N}}

\newcommand{\scri}{\mathscr{I}}

\newcommand{\abs}[1]{\left|#1 \right|}

\newcommand{\eq}[1]{({\ref{#1}})}
\newtheorem{Theorem}{Theorem}
\newtheorem{Proposition}{Proposition}
\newtheorem{conj}{Conjecture}
\newtheorem{Lemma}{Lemma}
\newtheorem{Corollary}[Theorem]{Corollary}
\newtheorem{Remark}{Remark}
\newtheorem{Definition}{Definition}
\numberwithin{Theorem}{section}
\numberwithin{Lemma}{section}
\numberwithin{Proposition}{section}

\newcommand{\f}{\frac}
\newcommand{\rd}{\partial}
\def\th{\theta}

\title{Asymptotic Properties of Linear Field Equations in Anti-de Sitter Space}

\author{Gustav Holzegel}
\thanks{\texttt{g.holzegel@imperial.ac.uk} \\
\phantom{1   }\hspace{.05cm} Dept.~of Mathematics, South Kensington Campus, Imperial College London, SW7 2AZ, UK}

\author{Jonathan Luk}
\thanks{\vspace{.1cm} \texttt{jluk@dpmms.cam.ac.uk}\\
\phantom{1   }\hspace{.05cm} Dept.~of Mathematics, Stanford University, 450 Serra Mall, Building 380, Stanford, CA 94305-2125}

\author{Jacques Smulevici}
\thanks{\vspace{.1cm} \texttt{jacques.smulevici@math.u-psud.fr}\\
\phantom{1   }\hspace{.05cm} Laboratoire de Math\'ematiques, Universit\'e Paris-Sud 11, b\^at. 425, 91405 Orsay, France}

\author{Claude Warnick}
\thanks{\vspace{.1cm} \texttt{c.m.warnick@maths.cam.ac.uk}\\
\phantom{1   }\hspace{.05cm} DPMMS and DAMTP, University of Cambridge, Wilberforce~Road,  Cambridge CB3 0WA, UK}

\date{\today \vspace{.1cm}}  

\begin{document}

\begin{abstract}
We study the global dynamics of the wave equation, Maxwell's equation and the linearized Bianchi equations on a fixed anti-de Sitter (AdS) background. Provided dissipative boundary conditions are imposed on the dynamical fields we prove uniform boundedness of the natural energy as well as both degenerate (near the AdS boundary) and non-degenerate integrated decay estimates. Remarkably, the non-degenerate estimates ``lose a derivative". We relate this loss to a trapping phenomenon near the AdS boundary, which itself originates from the properties of (approximately) gliding rays near the boundary. Using the Gaussian beam approximation we prove that non-degenerate energy decay without loss of derivatives does not hold. As a consequence of the non-degenerate integrated decay estimates, we also obtain pointwise-in-time decay estimates for the energy. Our paper provides the key estimates for a proof of the non-linear stability of the anti-de Sitter spacetime under dissipative boundary conditions. Finally, we contrast our results with the case of reflecting boundary conditions. 
\end{abstract}

\maketitle

\section{Introduction}
The non-linear Einstein vacuum equations with cosmological constant $\Lambda$,
\begin{align} \label{eve}
Ric \left[g\right]= \Lambda g \, ,
\end{align}
constitute a complicated coupled quasi-linear hyperbolic system of partial differential equations for a Lorentzian metric $g$. The past few decades have seen fundamental progress in understanding the global dynamics of solutions to (\ref{eve}). In particular, a satisfactory answer -- asymptotic stability -- has been given for the dynamics of (\ref{eve}) with $\Lambda=0$ near  Minkowski space \cite{ChrKla}; the dynamics of (\ref{eve}) with $\Lambda>0$ near de Sitter space \cite{Ringstroem, FriedrichdS}, the maximally symmetric solution of (\ref{eve}) with $\Lambda>0$; as well as for the ($\Lambda >0$) Kerr-dS black hole spacetimes  \cite{Hintz:2016gwb}. Today, the dynamics near black hole solutions of (\ref{eve}) for $\Lambda=0$  is a subject of intense investigation \cite{Mihalisnotes, YakovKerrFull}, with the current state-of-the-art being the results  in \cite{Dafermos:2016uzj} establishing the  linear stability of Schwarzschild, and \cite{Klainerman:2017nrb} establishing nonlinear stability for Schwarzschild within a restricted symmetry class. 

In contrast to the above, the global dynamics of (\ref{eve}) with $\Lambda<0$ near anti-de Sitter space (AdS), the maximally symmetric solution of the vacuum Einstein equations with $\Lambda<0$, is mostly unknown. Part of the problem is that in the case of $\Lambda<0$, the PDE problem associated with (\ref{eve}) takes the form of an initial boundary value problem.
 Therefore, even to construct local in time solutions, one needs to understand what appropriate (well-posed, geometric) boundary conditions are. It also suggests that the global behaviour of solutions starting initially close to the anti-de Sitter geometry may depend crucially on the choice of these boundary conditions \cite{HelmutonAdS, Friedrich}. 

Using his conformal field equations, Friedrich \cite{Friedrich} constructed local in time solutions to (\ref{eve}) with $\Lambda<0$. (See also \cite{Enciso} for a recent proof for Dirichlet conditions using harmonic gauge.) While in general it is quite intricate to isolate the geometric content inherent in the boundary conditions imposed (partly due to the large gauge freedom present in the problem), there is nevertheless a ``conformal" piece of boundary data that does admit a physical interpretation. This goes back to the Bianchi equations,
\begin{align} \label{bianchi}
\left[\bold{\nabla}_g\right]^a W_{abcd} = 0 \, ,
\end{align}
satisfied by the Weyl-tensor of a metric $g$ satisfying (\ref{eve}). As is well known \cite{ChrKla}, the equations (\ref{bianchi}) can be used to estimate the curvature components of a dynamical metric $g$. The boundary data required for well-posed evolution of (\ref{bianchi}) will generally imply a condition on the energy-flux of curvature through the timelike boundary. Two ``extreme" cases seem particularly natural and interesting: The case when this flux vanishes, corresponding to reflecting (Dirichlet or Neumann) conditions, and the case when this flux is ``as large as possible", corresponding to ``optimally dissipative" conditions. While any such boundary data for (\ref{bianchi}) will have to be complemented with other data (essentially the choice of a boundary defining function and various gauge choices) to estimate the full spacetime metric  \cite{HelmutonAdS}, it is nevertheless reasonable, in view of the strong non-linearities appearing in the Einstein equations, to conjecture the following loose statement for the global dynamics of perturbations of AdS under the above ``extreme" cases of boundary conditions:

\begin{conj} \label{conj:ads}
Anti-de Sitter spacetime is non-linearly unstable for reflecting and asymptotically stable for optimally dissipative boundary conditions.
\end{conj}
The instability part of Conjecture \ref{conj:ads}  was first made in \cite{DafHol, Newtontalk} in connection with work on five-dimensional gravitational solitons. See also \cite{Anderson}. By now there exist many refined versions of this part of the conjecture as well as strong heuristic and numerical support in its favour \cite{Bizon, Dias, Buchel}. In a series of papers \cite{Moschidis:2017lcr, Moschidis:2017llu, Moschidis:2018kcf}, Moschidis has considered the stability problem for the AdS spacetime within spherical symmetry, in the presence of null dust or massless Vlasov matter, culminating in a proof of the instability of AdS for the Einstein--Massless Vlasov system in \cite{ Moschidis:2018ruk}.

The present paper is the first of a series of papers establishing the stability part of Conjecture \ref{conj:ads}. Here we contribute the first fundamental ingredient, namely robust decay estimates for the associated linear problem. Our interest in the case of dissipative boundary conditions in part goes back to the original work of Friedrich in \cite{Friedrich}, which establishes that local well posedness does not single out a preferred boundary condition at null infinity. It is then reasonable to ask how questions of global existence depend on the choice of boundary condition. We also observe that to date the only solutions to the vacuum Einstein equations (regardless of boundary conditions) which are future-complete are necessarily stationary. See for example \cite{Chrusciel:2005uy}  for  general constructions of such spacetimes. A positive resolution of the stability part of Conjecture \ref{conj:ads} would necessarily imply the existence of truly dynamical, future-complete, solutions  to \eq{eve}.

\subsection{Linear field equations on AdS}
An important prerequisite for any non-linear stability result is that the associated \emph{linear} problem is robustly controlled \cite{ChKlLinear}. In the present (dissipative) context, this means that the mechanisms and obstructions for the decay of linear waves in the \emph{fixed} AdS geometry should be understood and decay estimates with constants depending on the initial data available. We accomplish this by giving a complete description of the decay properties of three fundamental field equations of mathematical physics:
\begin{enumerate}
\item  [(W)] The conformal wave equation for a scalar function on the AdS manifold\footnote{To simplify the algebra we choose $\Lambda=-3$ throughout the paper.},
\begin{align} \label{introwave}
\Box_{g_{AdS}} u + 2u = 0 \, .
\end{align}
\item [(M)] Maxwell's equations for a two-form $F$ on the AdS manifold,
\begin{align} \label{intromax}
dF = 0 \ \ \  \textrm{and} \ \ \ d \star_{g_{AdS}} F = 0 \, .
\end{align}
\item [(B)] The Bianchi equations for a Weyl field $W$ (see Definition \ref{def:weyl} below) on the AdS manifold
\begin{align} \label{introweyl}
\left[ \nabla_{g_{AdS}}\right]^a W_{abcd} = 0 \, .
\end{align}
\end{enumerate}
Note that since AdS is conformally flat, equation (\ref{introweyl}) is precisely the linearization of the full non-linear Bianchi equations (\ref{bianchi}) with respect to the anti-de Sitter metric. 

The models (W), (M), and (B) will be accompanied by dissipative boundary conditions. In general, to even \emph{state} the latter, one requires a choice of boundary defining function for AdS and a choice of timelike vectorfield (or, alternatively, the choice of an outgoing null-vector) at each point of the AdS boundary. For us, it is easiest to state these conditions in coordinates which suggest a canonical choice for both these vectors. We write the AdS metric in spherical polar coordinates on $\mathbb{R}^4$, where it takes the simple familiar form 
\begin{align} \label{adsmetric}
g_{AdS} = - \left(1+r^2\right) dt^2 + \left(1+r^2\right)^{-1} dr^2 + r^2 d\Omega_{\mathbb S^2}^2 \, ,
\end{align}
with the asymptotic boundary corresponding to the timelike hypersurface ``$r=\infty$".\footnote{The causal nature of this boundary is clearer in the Penrose picture discussed in Section \ref{sec:cof} below.} Note that $1/r$ is a boundary defining function and that
\[
e_0 = \frac{1}{\sqrt{1+r^2}} \partial_t \ \ \ , \ \ \ e_{\r} = \sqrt{1+r^2} \partial_r \ \ \ , \ \ \ e_{A} = \slashed{e}_A 
\]
with $\slashed{e}_A$ ($A=1,2$) an orthonormal frame on the sphere of radius $r$ defines an orthonormal frame for AdS. Finally, the vector $\partial_t$ singles out a preferred timelike direction.

In the case of the wave equation, the optimally dissipative boundary condition can then be stated as
\begin{equation} \label{introwd}
\frac{\partial (r u)}{\partial t} + r^2\frac{\partial (r u)}{\partial r} \to  0, \qquad \textrm{ as }r\to \infty \, .
\end{equation}
Note that $\partial_t + \left(1+r^2\right) \partial_r$ is an outgoing null-vector. See also Section \ref{sec:bc}.

In the case of Maxwell's equations, we define the electric and magnetic field $E_i = F\left(e_0,e_i\right)$ and $H_i = \star_{g_{AdS}} F\left(e_0,e_i\right)$ respectively. Here $\star_{g_{AdS}}$ is the Hodge dual with respect to the AdS metric. The dissipative boundary condition then takes the form
\ben{intromd}
r^2 \left(E_A + \epsilon_A{}^{B} H_B\right) \to 0, \qquad \textrm{ as } r \to \infty \, ,
\een
which means that the Poynting vector points outwards, allowing energy to leave the spacetime. 

Finally, in the case of the Bianchi equations, we define the electric and magnetic part of the Weyl tensor $E_{AB} = W\left(e_0,e_A,e_0,e_B\right)$ and $H_{AB}=\star_{g_{AdS}} W \left(e_0,e_A,e_0,e_B\right)$ respectively.  Introducing the trace-free part of $E_{AB}$ as $\hat{E}_{AB} = E_{AB} - \frac{1}{2}\delta_{AB} E_C{}^C$ and similarly for $H_{AB}$,  the dissipative boundary conditions can then be expressed as:
\ben{introweyld}
r^3 \left( \hat{E}_{AB}+ \epsilon_{(A}{}^{C} \hat{H}_{B) C}\right) \to 0, \qquad \textrm{ as } r \to \infty \, .
\een
Interpreting the Bianchi equations as a linearisation of the full vacuum Einstein equations, we can understand these boundary conditions in terms of the metric perturbations as implying a relation between the Cotton-York tensor of the conformal metric on $\scri$ and the ``stress-energy tensor'' of the boundary\footnote{We use here the nomenclature of the putative AdS/CFT correspondence.}.

That the above boundary conditions are indeed correct, naturally dissipative, boundary conditions leading to a well-posed boundary initial value problem will be a result of the energy identity. While this is almost immediate in the case of the wave equation and Maxwell's equations, we will spend a considerable amount of time on the ``derivation" of (\ref{introweyld}) in the Bianchi case, see Section \ref{sec:mfs}. 

Note that we could consider a more general Klein-Gordon equation than (W):
\begin{align}
\Box_{g_{AdS}} u + a u = 0 \, .
\end{align}
Here the boundary behaviour of $u$ is more complicated \cite{Warnick:2012fi}, but a well posedness theory with dissipative boundary conditions is still available \cite{Gannot:2015twa} for $a$ in a certain range. We restrict attention to the case $a=2$ for two reasons. Firstly, for this choice of $a$ the equation has much in common with the systems (M), (B), owing to their shared conformal invariance. This is the correct `toy model' for the Einstein equations. Secondly, the problem appears to be considerably more challenging for $a \neq2$. In particular, it does not appear that the methods of this paper can be directly applied, even once allowance is made using the formalism of \cite{Warnick:2012fi} for the more complicated boundary behaviour of solutions.

\subsection{The main theorems}
We now turn to the results.
\begin{Theorem}\label{wave full decay}
Let one of the following hold
\begin{enumerate}
\item [(W)] $\Psi$ is a scalar function and a smooth solution of \eq{introwave} subject to dissipative boundary condition \eq{introwd}. We associate with $\Psi$ the energy density
\begin{align} 
\varepsilon \left[\Psi\right] &:= \sqrt{1+r^2} \left( \frac{\left( \p_t \Psi\right)^2 + \Psi^2  }{1+r^2}  + \left[\p_r \left( \sqrt{1+r^2} \Psi\right)\right]^2 + \abs{\slashed{\nabla} \Psi}^2 \right) \nonumber 
\end{align}
\item [(M)] $\Psi$ is a Maxwell-two-form and a smooth solution of \eq{intromax} subject to dissipative boundary conditions \eq{intromd}. We associate with $\Psi$ the energy density
\begin{align}
\varepsilon\left[\Psi \right] = \sqrt{1+r^2} \left( |E|^2 + |H|^2 \right) \nonumber
\end{align}
where $E$ and $H$ denote the electric and magnetic part of $\Psi$ respectively.
\item [(B)] $\Psi$ is a Weyl-field and a smooth solution of \eq{introweyl} subject to dissipative boundary conditions \eq{introweyld}. We associate with $\Psi$ the energy density
\begin{align}
\varepsilon\left[\Psi\right] &=  \left(1+r^2\right)^\frac{3}{2} \left( |E|^2 + |H|^2 \right)  \nonumber
\end{align}
where $E$ and $H$ denote the electric and magnetic part of the Weyl-field respectively.
\end{enumerate}
\vspace{.3cm}
Then we have the following estimates
\begin{enumerate}
\item Uniform Boundedness: For any $0<T<\infty$ we have
\be
\int_{\Sigma_T}  \frac{\varepsilon \left[\Psi\right]}{\sqrt{1+r^2}}r^2 dr d\omega \lesssim \int_{\Sigma_0}  \frac{\varepsilon \left[\Psi\right]}{\sqrt{1+r^2}}r^2 dr d\omega \, ,
\ee
where the implicit constant is independent of $T$.
\item Degenerate (near infinity) integrated decay without derivative loss:
\be
\int_0^\infty dt \int_{\Sigma_t} \frac{\varepsilon[\Psi]}{1+r^2} r^2 dr d\omega \lesssim \int_{\Sigma_0}  \frac{\varepsilon\left[\Psi\right]}{\sqrt{1+r^2}}r^2 dr d\omega \, .
\ee
\item Non-degenerate (near infinity) integrated decay with derivative loss:
\be
\int_0^\infty dt \int_{\Sigma_t} \frac{\varepsilon[\Psi]}{\sqrt{1+r^2}} r^2 dr d\omega \lesssim \int_{\Sigma_0}  \frac{\varepsilon \left[\Psi\right]+\varepsilon \left[\partial_t \Psi\right]}{\sqrt{1+r^2}}r^2 dr d\omega \, .
\ee
\end{enumerate}
\end{Theorem}

\begin{Remark}
Similar statements hold for higher order energies by commuting with $\partial_t$ and doing elliptic estimates. As this is standard we omit the details.
\end{Remark}


\begin{Corollary}[Uniform decay]\label{uniform decay}
Under the assumptions of the previous theorem the following uniform-in-time decay estimate holds for any integer $n\geq 1$
\be
\int_{\Sigma_t}  \frac{\varepsilon \left[\Psi\right]}{\sqrt{1+r^2}}r^2 dr d\omega \lesssim \frac{1}{\left(1+t\right)^n} \int_{\Sigma_0}  \frac{\varepsilon \left[\Psi\right]+\varepsilon \left[\partial_t \Psi\right]+ \ldots + \varepsilon\left[\partial_t^n \Psi\right]}{\sqrt{1+r^2}}r^2 dr d\omega \, .
\ee
\end{Corollary}

What is remarkable about the above theorem is that the derivative loss occurring in (3) allows one to achieve integrated decay of the energy without loss in the asymptotic weight $r$. While it is likely that more refined methods can reduce the loss of a full derivative in (3), we shall however establish that \emph{some} loss is necessary and in fact reflects a fundamental property of the hyberbolic equations on AdS: the presence of trapping at infinity.

\begin{Theorem} \label{theo:gb}
With the assumptions of Theorem \ref{wave full decay}, the term $\varepsilon[\partial_t \Psi]$
on the right hand side of estimate (3) of Theorem \ref{wave full decay} is necessary: The estimate fails (for general solutions) if it is dropped.
\end{Theorem}

We will prove Theorem \ref{theo:gb} only for the case of the wave equation (W), see (\ref{Etud}) and Corollary \ref{cor:ert}. The proof is based on the Gaussian beam approximation for the wave equation.\footnote{See \cite{Sbierski:2013mva}  for a discussion of the Gaussian beam approximation on general Lorentzian manifolds.} In particular, we construct a solution of the conformal wave equation in AdS which contradicts estimate (3) of Theorem \ref{wave full decay} without the $\varepsilon[\partial_t \Psi]$-term on the right hand side. Similar constructions can be given for the Maxwell and the Bianchi case.\footnote{These will be more involved in view of the fact that the equations involve constraints.}

\subsection{Overview of the proof of Theorem \ref{wave full decay} and main difficulties}
The proof of Theorem \ref{wave full decay} is a straightforward application of the vectorfield method once certain difficulties have been overcome. Let us begin by discussing the proof in case of the wave- (W) and Maxwell's equation (M) as it is conceptually easier.  

In the case of (W) and (M), statement (1) follows immediately from integration of the divergence identity $\nabla^a \left(\T_{ab} \left(\partial_t\right)^b\right)=0$ with $\T$ being the energy momentum tensor of the scalar or Maxwell field respectively. In addition, in view of the dissipative condition, this estimate gives control over certain derivatives of $u$ (certain components of the Maxwell field, namely $E_A$ and $H_A$) integrated along the boundary. 

The statement (2) of the main theorem is then obtained by constructing a vectorfield $X$ (see (\ref{Xdef})) which is almost conformally Killing near infinity. The key observations are that firstly, the right hand side of the associated divergence identity $\nabla^a \left(\T_{ab} X^b\right)={}^X\pi \cdot \T$ controls all derivatives of $u$ in the wave equation case (components of $F$ in the Maxwell case). Secondly, when integrating this divergence identity, the terms appearing on the boundary at infinity come either with good signs (angular derivatives in the case of the wave equation, $E_{\r}$ and $H_{\r}$ components in the Maxwell case) or are components already under control from the previous $\T_{ab} \left(\partial_t\right)^b$-estimate. The integrated decay estimate thus obtained comes with a natural degeneration in the $r$-weight, as manifest in the estimate (2) of Theorem \ref{wave full decay}. 

To remove this degeneration we first note that in view of the fact that the vectorfield $\partial_t$ is Killing, the estimates (1) and (2) also hold for the $\partial_t$-commuted equations. In the case of the wave equation, controlling $\partial_t \partial_t \psi$ in $L^2$ on spacelike slices implies an estimate for all spatial derivatives of $\psi$ through an elliptic estimate. The crucial point here is that \emph{weighted} estimates are required. Similarly, one can write Maxwell's equation as a three-dimensional div-curl-system with the time derivatives of $E$ and $H$ on the right hand side. Again the crucial point is that \emph{weighted} elliptic estimates are needed to prove the desired results. 

Once all (spatial) derivatives are controlled in a weighted $L^2$ sense on spacelike slices one can invoke Hardy inequalities to improve the weight in the lower order terms and remove the degeneration in the estimate (2). 

For the case of the spin 2 equations (B), the proof follows a similar structure. However, the divergence identity for the Bel-Robinson tensor (the analogue of the energy momentum tensor in cases (W) and (M)) alone \emph{will not} generate the estimate (1). In fact, the term appearing on the boundary after integration does not have a sign unless one imposes an additional boundary condition! On the other hand, one can show (by proving energy estimates for a reduced system of equations, see Section \ref{sec:reds}) that the boundary conditions (\ref{introweyld}) already uniquely determine the solution.\footnote{The introduction of this system in the context of the full non-linear problem goes back to \cite{Friedrich}.} The resolution is that in the case of the Bianchi equations one needs to prove (1) and (2) at the same time: Contracting the Bel-Robinson tensor with a suitable combination of the vectorfields $\partial_t$ and $X$ ensures that the boundary term on null-infinity does have a favorable sign and so does the spacetime-term in the interior. Once (1) and (2) are established, (3) follows from doing elliptic estimates for the reduced system of Bianchi equations similar to the Maxwell case (M).

\subsection{Remarks on Theorem \ref{wave full decay}}

The estimates of Theorem  \ref{wave full decay} remain true for a class of $C^k$ perturbations of AdS which preserve the general properties of the deformation tensor of the timelike vectorfield $\sqrt{3}\partial_t +X$, cf.~the proof of Proposition \ref{weyl derivative decay}. This includes perturbations which may be dynamical. This fact is of course key for the non-linear problem \cite{globalnlsads}. 

The estimates of Theorem  \ref{wave full decay} are stable towards perturbations of the optimally dissipative boundary conditions. In the cases of (W) and (M) one can in fact establish these estimates for any (however small) uniform dissipation at the boundary, cf.~Section \ref{sec:gen}. Whether this is possible also in case of (B) is an open problem and (if true) will require a refinement of our techniques. In Section \ref{sec:Dirichlet} we discuss the case of Dirichlet conditions for (B)  which fix the conformal class of the induced metric at infinity to linear order. The Dirichlet boundary conditions may be thought of as a limit of dissipative conditions in which the dissipation vanishes. We outline a proof of boundedness for solutions of the Bianchi equations in this setting. We also briefly discuss the relation of our work to the Teukolsky formalism in Section \ref{Teukolsky section}, and argue that the proposed boundary conditions of \cite{Dias:2013sdc} may not lead to a well posed dynamical problem. We state a set of boundary conditions for the Teukolsky formulation that \emph{does} lead to a well posed dynamical problem.

Finally, one may wonder whether and how the derivative loss in Theorem \ref{wave full decay} manifests itself in the non-linear stability problem. In ongoing work \cite{globalnlsads}, we shall see that the degeneration is sufficiently weak in the sense that the degenerate estimate (2)  will be sufficient to deal with the non-linear error-terms.


\subsection{Main ideas for the proof of Theorem \ref{theo:gb}} \label{sec:cof}
In order to better understand both the geometry of AdS and the derivative loss occurring in (3) of Theorem \ref{wave full decay}, it is useful to invoke the conformal properties of AdS and the conformal wave equation (\ref{introwave}).
Setting $r=\tan \psi$ with $\psi \in \left(0, \frac{\pi}{2}\right)$ one finds from (\ref{introwave})
\begin{align}
g_{AdS} = \frac{1}{\cos^2 \psi} \left( -dt^2 + d\psi^2 + \sin^2 \psi d\Omega_{\mathbb S^2}^2\right) =: \frac{1}{\cos^2 \psi} g_{E} \, ,
\end{align}
and hence that AdS is conformal to a part of the Einstein cylinder, namely $\mathbb{R} \times \mathbb S^3_h \subset \mathbb{R} \times \mathbb S^3$, where $\mathbb S^3_h$ denotes a hemisphere of $\mathbb S^3$ with $\psi = \frac{\pi}{2}$ being its equatorial boundary. The wave equation (\ref{introwave}) is called \emph{conformal} or \emph{conformally invariant} because if $u$ is a solution of (\ref{introwave}), then $v := \Omega u$ is a solution of a wave equation with respect to the conformally transformed metric $g = \Omega^2 g_{AdS}$. This suggests that 
understanding the dynamics of (\ref{introwave}) for $g_{AdS}$ is essentially equivalent to that of understanding solutions of
\begin{align} \label{transformedwave}
\Box_{g_{E}} v - v = 0
\end{align}
on one hemisphere of the Einstein cylinder. This latter is a finite problem, which we will refer to as Problem 1 below.
\begin{enumerate}
\item Problem 1: The wave equation \eq{transformedwave} on $\mathbb{R}_t \times \mathbb S_h^3$ (with the natural product metric of the Einstein cylinder) where $\mathbb S_h^3$ is the (say northern) hemisphere of the $3$-sphere $\mathbb S^3$ with boundary at $\psi=\frac{\pi}{2}$, where (say optimally) dissipative boundary conditions are imposed.  We contrast this problem with 
\item Problem 2: the wave equation (\ref{introwave}) on $\mathbb{R}_t \times \mathbb{B}^3$ (with the flat metric) where $\mathbb{B}^3$ is the unit ball with boundary $\mathbb S^2$ where dissipative boundary conditions are imposed. 
\end{enumerate} 
For Problem 2 it is well-known that exponential decay of energy holds without loss of derivatives, see \cite{JLLionsNotes, Chen, Lagnese}. (It is an entertaining exercise to prove this in a more robust fashion using the methods of this paper.) 
For Problem 1, however, there will be a derivative loss present in any decay estimate, as seen in Theorem \ref{wave full decay}.\footnote{It is straightforward to translate the estimates established for (\ref{introwave}) in Theorem \ref{wave full decay} into the conformal picture currently discussed.}

This phenomenon can be explained in the geometric optics approximation for the wave equation. Recall that in this picture, the optimally dissipative boundary condition says that the energy of a ray is fully absorbed if it hits the boundary orthogonally. For rays which graze the boundary, the fraction of the energy that is absorbed upon reflection depends on the glancing angle: the shallower the incident angle, the less energy is lost in the reflection. 

Now let us fix a (large) time interval $\left[0,T\right]$ for both Problem 1 and Problem 2. To construct a solution which decays very slowly, we would like to identify rays which a) hit the boundary as little as possible and b) if they do hit the boundary, they should do this at a very shallow angle (grazing rays). 

What happens in Problem 2, is that the shallower one chooses the angle of the ray, the more reflections will take place in $\left[0,T\right]$. This is easily seen by looking at the projection of the null rays to the surface $t=0$ (see figure). 
\[
\begin{picture}(0,0)%
\includegraphics{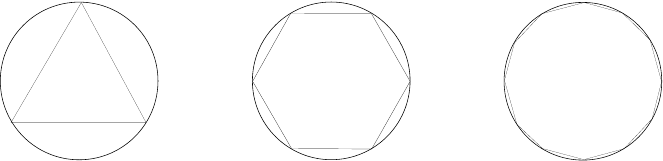}%
\end{picture}%
\setlength{\unitlength}{395sp}%
\begingroup\makeatletter\ifx\SetFigFont\undefined%
\gdef\SetFigFont#1#2#3#4#5{%
  \reset@font\fontsize{#1}{#2pt}%
  \fontfamily{#3}\fontseries{#4}\fontshape{#5}%
  \selectfont}%
\fi\endgroup%
\begin{picture}(52975,12701)(-36737,-12096)
\end{picture}%
\]

In sharp contrast, for Problem 1, the time until a null geodesic will meet the boundary again after reflection \emph{does not depend on the incident angle}! This goes back to the fact that all geodesics emanating from point on the three sphere refocus at the antipodal point. As a consequence, in Problem 1 we can keep the number of reflections in $\left[0,T\right]$ fixed while choosing the incident angle as small as we like. This observation is at the heart of the Gaussian beam approximation invoked to prove Theorem \ref{theo:gb}.

\subsection{Structure of the paper}
We conclude this introduction providing the structure of the paper. In Section \ref{sec:prelim} we define the coordinate systems, frames and basic vectorfields which we are going to employ. Section \ref{sec:fieldeq} introduces the field equations of spin 0, 1 and 2 fields, together with their energy momentum tensors. The well-posedness under dissipative boundary conditions for each of these equations is discussed in Section \ref{sec:wellposed} with particular emphasis on the importance of the reduced system in the spin 2 case. Section \ref{sec:proofmt} is at the heart of the paper proving the global results of Theorem \ref{wave full decay}, first for the spin 0 (Section \ref{sec:wd}), then the spin 1 (Section \ref{sec:md}) and finally the spin 2 case (Section \ref{sec:bd}). Corollary \ref{uniform decay} is proven in Section \ref{sec:proofcor} and Theorem \ref{theo:gb} in Section \ref{sec:proofgb}. We conclude the paper outlining generalizations of our result. Some elementary computations have been relegated to the appendix. 

\subsection{Acknowledgements} 
G.H. is grateful for the support offered by a grant of the European Research Council.  G.H., J.S.\ and C.M.W.\ are grateful to the Newton Institute for support during the programme ``Mathematics and Physics of the Holographic Principle''. J.L. thanks the support of the NSF Postdoctoral Fellowship DMS-1204493.


%
%
%
%
%
\section{Preliminaries} \label{sec:prelim}

\subsection{Coordinates, Frames and Volume forms}
We will consider the so-called ``global'' anti-de Sitter space-time, defined on $\R^4$, with metric (in polar coordinates)
\ben{AdS metric}
g_{AdS} = -\left(1+r^2 \right) dt^2 + \frac{dr^2}{1+r^2} + r^2 d\Omega_{\mathbb S^2}^2
\een
with $d\Omega_{\mathbb S^2}^2$ the standard round metric on the unit sphere. As long as there is no risk of confusion, we will also denote the metric by $g$. It will be convenient to introduce an orthonormal basis $\{e^a\}$:
\ben{on basis}
e^0 = \sqrt{1+r^2} dt, \qquad e^\r = \frac{dr}{\sqrt{1+r^2}}, \qquad e^A =  \slashed{e}^A,
\een
for $A=1,2$, where $\slashed{e}^A$ are an orthonormal basis for the round sphere\footnote{Obviously there's no global orthonormal basis for $\mathbb S^2$. We can either take multiple patches, or else understand the capital Latin indices as abstract indices.} of radius $r$. Throughout this paper, we will use capital Latin letters for indices on the sphere while small Latin letters are reserved as spacetime indices.  The dual basis of vector fields is denoted $\{e_a\}$:
\ben{on basis vf}
e_0 = \frac{1}{\sqrt{1+r^2}} \frac{\partial}{\partial t}, \qquad e_\r = \sqrt{1+r^2} \frac{\partial}{\partial r}, \qquad e_A = \slashed{e}_A,
\een

We introduce the surfaces $\Sigma_T = \{t = T\}$ and $\tilde{\Sigma}^{[T_1, T_2]}_R = \{r =R, T_1\leq t \leq T_2\}$, which have respective unit normals:
\be
n := e_0, \qquad m:= e_\r.
\ee
The surface measures on these surfaces are
\be
dS_{\Sigma_T} = \frac{r^2}{\sqrt{1+r^2}} dr d\omega, \qquad dS_{\Sigma_R} = r^2{\sqrt{1+r^2}} dt d\omega
\ee
with $d\omega$ the volume form of the round unit sphere. We denote by $S_{[T_1, T_2]}:=\{T_1\leq t \leq T_2\}$ the spacetime slab between $\Sigma_{T_1}$ and $\Sigma_{T_2}$. Finally, the spacetime volume form is
\be
d\eta = r^2 dt dr d\omega.
\ee

\subsection{Vectorfields}
The global anti-de Sitter spacetime enjoys the property of being static. The Killing field
\ben{Tdef}
T := \frac{\partial }{\partial t} = \sqrt{1+r^2} e_0
\een
is everywhere timelike, and orthogonal to the surfaces of constant $t$. 

Besides the Killing field $T$, we will exploit the properties of the vectorfield
\ben{Xdef}
X := r \sqrt{1+r^2} \frac{\partial}{\partial r} = r e_\r  
\een
While the vectorfield $X$ is not Killing, it is almost conformally Killing near infinity. More importantly, it generates terms with a definite sign. To see this, note that
\ben{deftX}
{}^X\pi := \frac{1}{2} \mathcal{L}_X g = \frac{(e^0)^2}{\sqrt{1+r^2}} + g \sqrt{1+r^2}.
\een
In particular, contracting ${}^X\pi$ with a symmetric traceless tensor $\mathbb{T}$ will only see the first term and this contraction will in fact have a sign if $\mathbb{T}$ satisfies the dominant energy condition.

\subsection{Differential operators}

Throughout this paper, we will use $\nabla$ to denote the Levi-Civita connection of $g$. Define also $\Box_g$ as the standard Laplace-Beltrami operator:

$$\Box_g u:= \nabla^a\nabla_a u.$$

\subsection{The divergence theorem}
We denote by $\textrm{Div }$ the spacetime divergence associated to the metric $g$ for a vector field. In coordinates, it takes the following form:

If $K = K^0 e_0 + K^\r e_\r + K^A e_A$, then
\ben{vector divergence}
\textrm{Div } K =\frac{1}{\sqrt{1+r^2}} \frac{\p K^0}{\p t}  + \frac{1}{r^2} \frac{\partial}{\partial r} \left(r^2 \sqrt{1+r^2} K^\r \right) +  \slashed{\nabla}^A K_A.
\een
It will be useful to record the following form of the divergence theorem:
\begin{Lemma}[Divergence theorem] \label{Divergence Theorem}
Suppose that $K$ is a suitably regular vector field defined on $\R^4$. Then we have
\begin{align*}
0 &= \int_{\Sigma_{T_2}} K_a n^a dS_{\Sigma_{T_2}}-\int_{\Sigma_{T_1}} K_a n^a dS_{\Sigma_{T_1}} \\
&  \quad  -\lim_{r\to \infty} \int_{\tilde{\Sigma}_{r}^{[T_1, T_2]}} K_a m^a dS_{\tilde{\Sigma}_{r}} +\int_{S_{[T_1, T_2]}} \textrm{Div } K d\eta.
\end{align*}

\end{Lemma}

\section{The field equations} \label{sec:fieldeq}

\subsection{Spin 0: The wave equation}
The conformal wave equation on AdS is
\ben{cwe}
\Box_{g_{AdS}} u + 2u = 0 \, .
\een
A key object to study the dynamics of (\ref{cwe}) is the twisted, or renormalised energy-momentum tensor. See \cite{Holzegel:2012wt,Warnick:2012fi} for further details on the following construction. We first define the twisted covariant derivative\footnote{When acting on a scalar, we will sometimes write $\tilde{\nabla}_\mu f =: \tilde{\p}_\mu f$.} by:
\be
\tilde{\nabla}_a (\cdot) := \frac{1}{\sqrt{1+r^2}} \nabla_a \left( \cdot  \sqrt{1+r^2}\right).
\ee
With this derivative we construct the twisted energy momentum tensor
\ben{reem}
\T_{ab}[u] = \tilde{\nabla}_a u \tilde{\nabla}_b u - \frac{1}{2} g_{ab} \left( \tilde{\nabla}_c u \tilde{\nabla}^c u+ \frac{u^2}{1+r^2}\right),
\een
which satisfies
\ben{twisted em divergence}
\nabla_a \T^{a}{}_{b}[u] = \left(\Box_{g} u + 2 u \right) \tilde{\nabla}_b u -\T^c{}_c \tilde{\nabla}_b 1.
\een
If $K$ is any vector field then we can define the current
\ben{currentdef}
{}^K\J^a[u] = \T^a{}_b[u] K^b \, ,
\een
which is a compatible current in the sense of Christodoulou \cite{orange}. Moreover, if $K$ is Killing and $K(r) = 0$, then $K^a \tilde{\nabla}_a 1 = 0$ and we can see that ${}^K\J[u]$ is a conserved current when $u$ solves the conformal wave equation (\ref{cwe}).

\subsection{Spin 1: Maxwell's equations}
The Maxwell equations in vacuum are given by the following first order differential equations for a $2$-form $F$ on $\mathbb{R}^4$:
\be \label{eq:Maxwell}
dF = 0, \qquad d \left ( \star_g F\right ) = 0.
\ee
We decompose the Maxwell $2-$form as:
\ben{Maxwell decomposition}
F = E_\r\ e^0 \wedge e^\r + E_A\ e^0 \wedge e^A + H_\r\ e^1 \wedge e^2 + H_A\   \epsilon^A{}_B e^{B} \wedge e^\r,
\een
where $\epsilon_{AB}$ is the volume form on the sphere\footnote{$\epsilon_{12} = 1$ in this basis.}. The dual Maxwell $2-$form is:
\be
\star_g F = H_\r\ e^0 \wedge e^\r + H_A\ e^0 \wedge e^A - E_\r\ e^1 \wedge e^2 - E_A\   \epsilon^A{}_B e^{B} \wedge e^\r.
\ee
We will use the notation $E_i$ with $i=\r, 1, 2$, and similarly for $H$. Since $F$ is a smooth $2-$form, and the basis vectors $e_\mu$ are bounded (but not continuous) at the origin, we deduce that the functions  $E_i$, $H_i$ are bounded, but not necessarily continuous, at $r = 0$. With respect to this decomposition, Maxwell's vacuum equations (\ref{eq:Maxwell}) split into six evolution equations:
\begin{align}
\frac{r}{\sqrt{1+r^2}} \partial_t  E_\r&=  - r\epsilon^{AB} \slashed{\nabla}_A H_B, \label{evolEr} \\
\frac{r}{\sqrt{1+r^2}} \partial_t  E_A &= \epsilon_A{}^B \left[ \partial_r \left( r \sqrt{1+r^2} H_B \right) - r\slashed{\nabla}_B H_\r \right],  \label{evolEA}\\
\frac{r}{\sqrt{1+r^2}} \partial_t  H_\r&=   r\epsilon^{AB} \slashed{\nabla}_A E_B, \label{evolHr} \\
\frac{r}{\sqrt{1+r^2}} \partial_t  H_A &= -\epsilon_A{}^B \left[ \partial_r \left( r \sqrt{1+r^2} E_B \right) - r\slashed{\nabla}_B E_\r \right],  \label{evolHA}
\end{align}
and two constraints:
\begin{align}
0 &= \frac{\sqrt{1+r^2}}{r} \partial_r \left( r^2 E_\r \right) + r\slashed{\nabla}^A E_A, \label{consE}  \\
0 &= \frac{\sqrt{1+r^2}}{r} \partial_r \left( r^2 H_\r \right) + r\slashed{\nabla}^A H_A. \label{consH}
\end{align}
Here $r\slashed{\nabla}$ is the covariant derivative on the \emph{unit} sphere, which commutes with $\nabla_{\partial_r}$ (note that our conventions are that $\slashed{e}_A$ are  orthonormal vector fields on the sphere of radius $r$). The evolution equations (\ref{evolEr}--\ref{evolHA}) form a symmetric hyperbolic system. If the evolution equations hold, and assuming sufficient differentiability, it is straightforward to verify that
\be
\frac{\partial \mathscr{E}}{\partial t} = \frac{\partial \mathscr{H}}{\partial t} = 0,
\ee
where
\bean
\mathscr{E} &:=& \frac{\sqrt{1+r^2}}{r} \partial_r \left( r^2 E_\r \right) + r \slashed{\nabla}^A E_A, \\
\mathscr{H} &:=& \frac{\sqrt{1+r^2}}{r} \partial_r \left( r^2 H_\r \right) + r \slashed{\nabla}^A H_A.
\eean
Thus if the constraint equations hold on the initial data surface, then they hold for all times.

\subsubsection{The energy momentum tensor}
The analogue of the energy momentum tensor \eq{reem} for the wave equation is the symmetric tensor
\ben{energy momentum}
\T[F]_{ab} = F_{a c} F_b{}^c - \frac{1}{4} g_{ab} F_{cd} F^{cd},
\een
which satisfies
\be
\nabla^a \T[F]_{ab} = F_{b}{}^{c} \nabla^{d}F_{dc}  +  (\star F)_{b}{}^{c} \nabla^{d}(\star F)_{dc} 
\ee
so that if $F$ satisfies Maxwell's equations, $\T[F]$ is divergence free and traceless. We define in the obvious fashion the current
\ben{Fcurrentdef}
{}^K\J^a[F] = \T^a{}_b[F] K^b
\een
for any vector field $K$.

\subsection{Spin 2: The Bianchi equations} \label{sec:be}

The equations for a spin 2 field, also called the Bianchi equations, can be expressed as first order differential equations for a Weyl tensor, which is an arbitrary $4-$tensor which satisfies the same symmetry properties as the Weyl curvature tensor. More precisely:
\begin{Definition} \label{def:weyl}
We say a $4-$tensor $W$ is a Weyl tensor if it satisfies:
\begin{enumerate}[i)]
\item $W_{abcd} = - W_{bacd} = -W_{abdc}$.
\item $W_{abcd} + W_{acdb} + W_{abdc} = 0$.
\item $W_{abcd} = W_{cdab}$.
\item $W^a{}_{bad}=0$.
\end{enumerate}
The dual of a Weyl tensor is defined by (cf.~\cite{ChrKla})
\be
{}^{\star}W_{abcd} = \frac{1}{2} \epsilon_{abef} W^{ef}{}_{cd} = \frac{1}{2}W_{ab}{}^{ef} \epsilon_{efcd} \, .
\ee
\end{Definition}
The Bianchi equations are
\ben{Bianchi}
\nabla_{[a}W_{bc]de} = 0,
\een
which are equivalent to either of the equations
\begin{align}
\nabla^aW_{abcd} &= 0,  \label{eq:Bianchidiv} \\
\nabla^a{}^{\star}W_{abcd} &= 0,
\end{align}
and for us studying (\ref{eq:Bianchidiv}) will be particularly convenient.

As in the case of the Maxwell equations, it is convenient to decompose the Weyl tensor based on the $3+1$ splitting of space and time. In the Maxwell case, the field strength tensor $F$ decomposes as a pair of vectors tangent to $\Sigma_t$. In the spin 2 case, the Weyl tensor $W$ decomposes as pair of symmetric tensors tangent to $\Sigma_t$. We define:
\begin{align}
E_{ab}&:= W_{0a0b} = W_{cadb} (e_0)^c (e_0)^d, \\
H_{ab}&:= {}^\star W_{0a0b}= {}^\star W_{cadb} (e_0)^c (e_0)^d. 
\end{align}
The symmetries indeed ensure that $E, H$ are symmetric and tangent to $\Sigma_t$. Moreover, both $E$ and $H$ are necessarily trace-free. In fact, one can reconstruct the whole tensor $W$ from the symmetric trace-free fields $E$, $H$, see \cite[\S7.2]{ChrKla}.

We will further decompose $E, H$ along the orthonormal frame defined in \eq{on basis}. To write the equations of motion in terms of $E$, $H$ we consider the equations $0 = \nabla^a W_{a\r\r0}$,  $0 = \nabla^a W_{a(A\r)0}$ and $0 = \nabla^a W_{a(AB)0}$, from which we respectively find the evolution equations for $E$:
\begin{align}
\frac{r }{\sqrt{1+r^2}} \frac{\partial {E}_{\r\r} }{\partial t} &= -r \epsilon^{AB}\slashed{\nabla}_A H_{B\r}, \label{EvolErr} \tag{Evol $E_{\r \r}$}\\
\frac{r  }{\sqrt{1+r^2}} \frac{\partial {E}_{A\r}}{\partial t} &= \frac{1}{2}  \epsilon_{A}{}^{B}\left[ \frac{\partial_r \left(r (1+r^2) H_{B\r} \right)  }{\sqrt{1+r^2}} - r\slashed{\nabla}_B H_{\r\r} \right] - \frac{r}{2} \epsilon^{BC}\slashed{\nabla}_B H_{CA}, \label{EvolEAr} \tag{Evol $E_{A\r}$}\\
\frac{r  }{\sqrt{1+r^2}} \frac{\partial {E}_{AB}}{\partial t}  &= \epsilon_{(A}{}^{C}\left[ \frac{\partial_r \left(r (1+r^2) H_{B)C} \right)  }{\sqrt{1+r^2}}-r \slashed{\nabla}_{|C|} H_{B)\r} \right] \label{EvolEAB}. \tag{Evol $E_{AB}$}
\end{align}
From the equations  $0 = \nabla^a W_{a0A0}$ and $0 = \nabla^a W_{a0r0}$ respectively we find the constraint equations:
\begin{align}
\frac{\sqrt{1+r^2}}{r^2} \frac{\partial}{\partial r}\left(r^3 E_{\r\r} \right) +r\slashed{\nabla}^BE_{B\r} &=: \mathscr{E}_\r = 0 \label{ConsEAr}\tag{Con $E_{\r}$} \\
\frac{\sqrt{1+r^2}}{r^2} \frac{\partial}{\partial r}\left(r^3 E_{A\r} \right) +r\slashed{\nabla}^BE_{AB}  &=: \mathscr{E}_A = 0 \label{ConsEAB}  \tag{Con $E_{A}$}
\end{align}
By considering the equivalent equations for $^{\star}{}W$, we find that the evolution equations for $H$ can be obtained from these by the substitution $(E, H) \to (H, -E)$:
\begin{align}
\frac{r }{\sqrt{1+r^2}} \frac{\partial {H}_{\r\r}}{\partial t} &= r\epsilon^{AB}\slashed{\nabla}_A E_{B\r}, \label{EvolHrr} \tag{Evol $H_{\r\r}$}\\
\frac{r  }{\sqrt{1+r^2}} \frac{\partial {H}_{A\r}}{\partial t} &= -\frac{1}{2}  \epsilon_{A}{}^{B}\left[ \frac{\partial_r \left(r (1+r^2) E_{B\r} \right)  }{\sqrt{1+r^2}} - r\slashed{\nabla}_B E_{\r\r} \right] + \frac{r}{2} \epsilon^{BC}\slashed{\nabla}_B E_{CA}, \label{EvolHAr} \tag{Evol $H_{A\r}$}\\
\frac{r  }{\sqrt{1+r^2}} \frac{\partial {H}_{AB}}{\partial t}&= -\epsilon_{(A}{}^{C}\left[ \frac{\partial_r \left(r (1+r^2) E_{B)C} \right)  }{\sqrt{1+r^2}}- r\slashed{\nabla}_{|C|} E_{B)\r} \right] \label{EvolHAB}, \tag{Evol $H_{AB}$}
\end{align}
and the constraint equations for $H$ are:
\begin{align}
\frac{\sqrt{1+r^2}}{r^2} \frac{\partial}{\partial r}\left(r^3 H_{\r\r} \right) +r\slashed{\nabla}^BH_{B\r} &=:\mathscr{H}_\r = 0\label{ConsHAr}   \tag{Con $H_{\r}$},  \\
\frac{\sqrt{1+r^2}}{r^2} \frac{\partial}{\partial r}\left(r^3 H_{A\r} \right) +r\slashed{\nabla}^BH_{AB} &=: \mathscr{H}_A = 0\label{ConsHAB} \tag{Con $H_{B}$} . 
\end{align}

\subsubsection{The Bel-Robinson tensor}
The spin 2 analogue of the energy momentum tensor for the Maxwell field is the Bel-Robinson tensor \cite[\S7.1.1]{ChrKla}. This is defined to be
\ben{Bel-Robinson}
Q_{abcd} := W_{a e c f} W_{b \phantom{e} d \phantom{f}}^{\phantom{b} e \phantom{d} f} + {}^\star W_{a e c f} {}^\star W_{b \phantom{e} d \phantom{f}}^{\phantom{b} e \phantom{d} f}
\een
It is symmetric, trace-free on all pairs of indices and if $W$ satisfies the Bianchi equations then
\ben{BelRob conservation}
\nabla^a Q_{abcd} = 0.
\een
We shall require the following quantitative version of the dominant energy condition for the Bel-Robinson tensor
\begin{Lemma}\label{BR DEC}
Suppose that $t_1, t_2 \in T_p\mathcal{M}$ are future-directed timelike unit vectors, and that $-g(t_1, t_2) \leq B$ for some $B\geq 1$. Then there exists a constant $C>0$ such that
\be
\frac{1}{C B^4} Q(t_1, t_1, t_1, t_1) \leq Q(t_i, t_j, t_k, t_l) \leq {C B^4} Q(t_1, t_1, t_1, t_1)
\ee
holds for any $i, j, k, l \in \{1, 2\}$. Moreover, $Q(t_1, t_1, t_1, t_1) \geq 0$, with equality at a point $p$ if and only if $W$ vanishes at $p$.
\end{Lemma}
\begin{proof}
We follow the proof of \cite[Prop. 4.2]{ChKlLinear}. We can pick a pair of null vectors $e'_3, e'_4$ with $g(e'_3, e'_4) = -\frac{1}{2}$ such that $t_1 = e'_3 + e'_4$ and $t_2 = b e'_3 + b^{-1} e'_4$ for some $b\geq 1$. The condition $-g(t_1, t_2) \leq B$ implies
\be
\frac{1}{B} \leq b \leq B.
\ee
As a consequence, we can write 
\be 
Q(t_i, t_j, t_k, t_l) = \sum_{a, b, c, d = 3}^4 q^{abcd}_{ijkl} Q(e'_a, e'_b, e'_c, e'_d)
\ee 
where
\be
\frac{c}{B^4} \leq q^{abcd}_{ijkl} \leq c' B^4
\ee
for some combinatorial constants $c, c'$ which are independent of $B$. By \cite[Prop. 4.2]{ChKlLinear}, all of the quantities $Q(e'_a, e'_b, e'_c, e'_d)$ are non-negative, and so we have established the first part. For the second part, we note that, again following \cite[Prop. 4.2]{ChKlLinear}, the quantity $Q(t_1, t_1, t_1, t_1)$ controls all components of $W$.
\end{proof}

\section{Well posedness} \label{sec:wellposed}

\subsection{Dissipative boundary conditions} \label{sec:bc}
\subsubsection{Wave equation} 
For the wave equation, an analysis of solutions near to the conformal boundary leads us to expect that $u$ has the following behaviour:
\be
u = \frac{u_+(t, x^A)}{r} + \frac{u_-(t, x^A)}{r^2} + \O{\frac{1}{r^3}}  \qquad \textrm{ as }r\to \infty.
\ee
The boundary condition we shall impose for (\ref{cwe}) on anti-de Sitter space can be simply written in the form
\ben{wave dissipative}
\frac{\partial (r u)}{\partial t} + r^2\frac{\partial (r u)}{\partial r} \to  0, \qquad \textrm{ as }r\to \infty,
\een
which is equivalent to
\be
\frac{\partial u_+}{\partial t} - u_- = 0, \quad \textrm{ on }\scri.
\ee

\subsubsection{Maxwell's equations}
In the Maxwell case, by (optimally) dissipative boundary conditions we understand the condition
\ben{dcmax}
r^2 \left(E_A + \epsilon_A{}^{B} H_B\right) \to 0, \qquad \textrm{ as } r \to \infty \, .
\een
The above boundary conditions ensure that the asymptotic Poynting vector is outward directed, i.e.\ energy is leaving the spacetime. It effectively means that $\scri$ behaves like an imperfect conductor with a surface resistance, for which the electric currents induced by the electromagnetic field dissipate energy as heat. These boundary conditions are an example of Leontovic boundary conditions \cite[\S 87]{LandauLifshitz8}.

\subsubsection{Bianchi equations}
In the case of the Bianchi equation, by (optimally) dissipative boundary conditions we understand the condition
\be
r^3 \left( \hat{E}_{AB}+ \epsilon_{(A}{}^{C} \hat{H}_{B) C}\right) \to 0, \qquad \textrm{ as } r \to \infty \, .
\ee
which is equivalent to:
\ben{dcbia}
r^3 \left( E_{AB} - \frac{1}{2} \delta_{AB} E_C{}^C + \epsilon_{(A}{}^{C} H_{B) C}\right) \to 0, \qquad \textrm{ as } r \to \infty \, .
\een
Notice that there are only two boundary conditions imposed (since the relevant objects appearing in the boundary term are the \emph{trace-free} parts of $E_{AB}$, $H_{AB}$). We could choose more general dissipative boundary conditions. For simplicity we shall just consider those above, which ought in some sense to represent ``optimal dissipation'' at the boundary, but see \S\ref{sec:gen} for generalisations.


\subsection{The well-posedness statements}
We now state a general well-posedness statement for each of our three models with dissipative boundary conditions. As is well-known, the key to prove these theorems is the existence of a suitable energy estimate under the boundary conditions imposed. In the Wave- and Maxwell case such an estimate is immediate. In the Bianchi case, however, there is a subtlety (discussed and resolved already in \cite{Friedrich}). We will dedicate Section \ref{sec:mfs} to derive a local energy estimate showing that the condition (\ref{dcbia}) indeed leads to a well-posed problem.

\subsubsection{Spin 0: The wave equation}
The following result can be established either directly, or by making use of the conformal invariance of \eq{cwe}:
\begin{Theorem}[Well posedness for the conformal wave equation]\label{well posed wave}
Fix $T>0$. Given smooth functions $u_0, u_1:\Sigma_0 \to \R$ satisfying suitable asymptotic conditions, there exists a unique smooth $u$, such that:
\begin{enumerate}[i)]
\item  $u$ solves the conformal wave equation \eq{cwe} in $S_{[0, T]}$.
\item We have the estimate:
\be
\sup_{S_{[0, T]} } \sum_{k, l, m \leq K} \abs{\left( r^{2} \partial_r\right )^k \left( \partial_t \right)^l  \left(r  \slashed{\nabla} \right)^m \left( r u \right)}  \leq C_{T, u_0, u_1, K},
\ee
for $k, l, m = 0, 1, \ldots$ and a constant $C_{T, u_0, u_1, K}$ depending on $K,  T$ and the initial data. This in particular implies a particular asymptotic behaviour for the fields.
\item The initial conditions hold:
\be
\left. u\right|_{t=0} = u_0, \qquad \left. u_t\right|_{t=0} = u_1
\ee
\item Dissipative boundary conditions \eq{wave dissipative} hold.
\end{enumerate}
\end{Theorem}

\subsubsection{Spin 1: Maxwell's equations}
Working either directly with (\ref{evolEr}--\ref{consH}), or else by making use of the conformal invariance of Maxwell's equations, it can be shown that:
\begin{Theorem}[Well posedness for Maxwell's equations]\label{well posed}
Fix $T>0$. Given smooth vector fields $E^0_i, H^0_i$ satisfying the constraint equations, together with suitable asymptotic conditions, there exists a unique set of smooth vector fields: $E_i(t), H_i(t)$, such that:
\begin{enumerate}[i)]
\item  $E_i(t), H_i(t)$ solve (\ref{evolEr}--\ref{consH}) in $S_{[0, T]}$ and the corresponding Maxwell tensor $F$ is a smooth $2-$form on $S_{[0, T]}$.
\item We have the estimate:
\be
\sup_{S_{[0, T]} }\sum_{\substack{k, l, m \\ k', l', m'} \leq K} \abs{\left( r^{2} \partial_r\right )^k \left( \partial_t \right)^l  \left( r \slashed{\nabla} \right)^m \left( r^2 E_i \right)} + \abs{\left( r^{2} \partial_r\right )^{k'} \left( \partial_t \right)^{l'}  \left(  r \slashed{\nabla}\right)^{m'} \left( r^2 H_i \right)} \leq C_{T, E^0, H^0, K},
\ee
for any $K\geq 0$, where $C_{T, E^0, H^0, K}$ depends on $K, T$ and the initial data. This in particular implies a particular asymptotic behaviour for the fields.
\item The initial conditions hold:
\be
E_i(0) = E_i^0, \qquad H_i(0) = H_i^0
\ee
\item Dissipative boundary conditions \eq{dcmax} hold.
\end{enumerate}
\end{Theorem}

The asymptotic conditions on the initial data are corner conditions that come from ensuring that the initial data are compatible with the boundary conditions. It is certainly possible to construct initial data satisfying the constraints and the corner conditions to any order. We could work at finite regularity, and our results will in fact be valid with much weaker assumptions on the solutions, but for convenience it is simpler to assume the solutions are smooth.

\subsubsection{Spin2 : Bianchi equations}\label{spin2 well posedness section}
In the case of the Bianchi equation we can prove:
\begin{Theorem}[Well posedness for the Bianchi system]\label{spin2 well posedness}
Fix $T>0$. Given smooth traceless symmetric $2$-tensors $E^0_{ab}, H^0_{ab}$ on $\Sigma_0$ satisfying the constraint equations, together with suitable asymptotic conditions as $r \rightarrow \infty$, there exists a unique set of traceless symmetric $2$-tensors: $E_{ab}(t), H_{ab}(t)$ such that:
\begin{enumerate}[i)]
\item  $E_{ab}(t), H_{ab}(t)$ are traceless and the corresponding Weyl tensor $W$ is a smooth $4-$tensor satisfying the Bianchi equations on $S_{[0, T]}$.
\item We have the asymptotic behaviour:
\be
\sup_{S_{[0, T]} } \sum_{\substack{k, l, m \\ k', l', m'} \leq K} \abs{\left( r^{2} \partial_r\right )^k \left( \partial_t \right)^l  \left(  r \slashed{\nabla} \right)^m \left( r^3 E_{ab} \right)} + \abs{\left( r^{2} \partial_r\right )^{k'} \left( \partial_t \right)^{l'}  \left(  r \slashed{\nabla} \right)^{m'} \left( r^3 H_{ab} \right)} \leq C_{T, E^0, H^0, K}
\ee
as $r \to \infty$, for any $K \geq 0$.
\item The initial conditions hold:
\be
E_{ab}(0) = E_{ab}^0, \qquad H_{ab}(0) = H_{ab}^0
\ee
\item Dissipative boundary conditions \eq{dcbia} hold.
\end{enumerate}
\end{Theorem}
We will spend the remainder of this section to derive the key-energy estimate that is behind the proof of Theorem \ref{spin2 well posedness}.\footnote{Theorem \ref{wave full decay} of course establishes a stronger (global) estimate. The key point here is to explain why the naive approach using only the Bel-Robinson tensor and the Killing fields fails and also to derive the (reduced) system of equations which will be needed in the second part of the proof of Theorem \ref{wave full decay}.}

\subsection{The modified system of Bianchi equations} \label{sec:mfs}
In order to establish a well posedness theorem for the initial-boundary value problem associated to the spin 2 equations, the natural thing to do is to consider the evolution equations (Evol) as a symmetric hyperbolic\footnote{Or at least Friedrichs \emph{symmetrizable}, but the distinction is not important for our purposes. The key issue is the existence of a good energy estimate. Providing this exists, putting the system into symmetric hyperbolic form is straightforward.} system. Having established existence and uniqueness for this system (Step 1), one can then attempt to show that the constraints (Con) are propagated from the initial data (Step 2).

If one attempts this strategy with the equations in the form given in Section \ref{sec:be}, one finds that Step 1 causes no problems: one can easily derive an energy estimate for the system (Evol). In fact, it is a matter of simple calculation to check that taking:
\be
r (1+r^2)^{\frac{3}{2}} \left[ E_{\r\r} \times \textrm{\eq{EvolErr}} + 2E_{A\r}  \times \textrm{{\eq{EvolEAr}}} + E_{AB}  \times \textrm{{\eq{EvolEAB}}} + E\leftrightarrow H \right]
\ee
and integrating over $\Sigma_t$ with measure $dr d\omega$ we arrive at\footnote{Here and elsewhere, in expressions like $\abs{E_{A\r}}^2 +\abs{E_{AB}}^2$, contraction over the indices $A, B$ etc.\ is implied, i.e.\ $\abs{E_{A\r}}^2 +\abs{E_{AB}}^2 = E_{A\r}E^A{}_\r + E_{AB}E^{AB}$}
\begin{align} \nonumber
&\frac{d}{dt} \frac{1}{2} \int_{\Sigma_t} \left(\abs{E_{\r\r}}^2 + 2 \abs{E_{A\r}}^2 +\abs{E_{AB}}^2 + E\leftrightarrow H \right) r^2(1+r^2) dr d\omega \\& \qquad = \lim_{r\to \infty} \int_{\tilde{\Sigma}_r\cap \Sigma_t} \left(\frac{1}{2} \epsilon^{AB}E_{A\r}H_{B\r} + \epsilon^{AB}E_{A C}H_{B}{}^C \right) r^6d\omega \label{naive energy} \, .
\end{align}
Once one has an energy estimate of this kind, establishing the existence of solutions to the evolution equations under the assumption that the boundary term has a good sign is essentially straightforward. Looking at the boundary term we obtain, this suggests that boundary conditions should be imposed on both the $Ar$ and the $AB$ components of $E$ (or alternatively $H$) at infinity. We shouldn't be so hasty, however. Before declaring victory, we must return to look at the constraints (Step 2). 

Firstly, it is simple to verify that if (Evol) hold with sufficient regularity then
\be
0=\frac{\partial E_a{}^a}{\partial t} = \frac{\partial H_a{}^a}{\partial t} \, ,
\ee
so that the trace constraints on $E$ and $H$ are propagated by the evolution equations. Next, we turn to the differential constraints. We find that if (Evol) and the trace constraints hold, then the functions $\mathscr{E}_a, \mathscr{H}_a$ defined in Section \ref{sec:be} satisfy the system of equations:
\begin{align}
\frac{r}{\sqrt{1+r^2}} \frac{\partial \mathscr{E}_\r}{\partial t} &= - \frac{r}{2 } \epsilon^{AB}\slashed{\nabla}_A \mathscr{H}_B \label{EvolCEr} \tag{Evol $\mathscr{E}_{\r}$} \, , \\
\frac{r}{\sqrt{1+r^2}} \frac{\partial \mathscr{E}_A}{\partial t} &= \frac{\epsilon_A{}^B}{2} \left[\frac{\partial_r\left( r^2 \sqrt{1+r^2} \mathscr{H}_B\right)}{\sqrt{1+r^2}} - r\slashed{\nabla}_B\mathscr{H}_\r - \frac{\mathscr{H}_B}{\sqrt{1+r^2}} \right] \label{EvolCEA} \tag{Evol $\mathscr{E}_{A}$} \, , \\
\frac{r}{\sqrt{1+r^2}} \frac{\partial \mathscr{H}_\r}{\partial t} &= \frac{r}{2 } \epsilon^{AB}\slashed{\nabla}_A \mathscr{E}_B \label{EvolCHr} \tag{Evol $\mathscr{H}_{\r}$} \, ,  \\
\frac{r}{\sqrt{1+r^2}} \frac{\partial \mathscr{H}_A}{\partial t} &= -\frac{\epsilon_A{}^B}{2} \left[\frac{\partial_r\left( r^2 \sqrt{1+r^2} \mathscr{E}_B\right)}{\sqrt{1+r^2}} - r\slashed{\nabla}_B\mathscr{E}_\r - \frac{\mathscr{E}_B}{\sqrt{1+r^2}} \right] \label{EvolCHA} \tag{Evol $\mathscr{H}_{A}$} \, .
\end{align}
Now, things appear to be working in our favour. This system is symmetric hyperbolic and  we can check that by taking:
\be
r^2 (1+r^2) \left[ \mathscr{E}_{\r} \times \textrm{\eq{EvolCEr}} + \mathscr{E}_{A} \times \textrm{\eq{EvolCEA}}+  \mathscr{H}_{\r} \times \textrm{\eq{EvolCHr}} + \mathscr{H}_{A} \times \textrm{\eq{EvolCHA}} \right] \, 
\ee
and integrating over $\Sigma_t$ with the measure $dr d\omega$ we can derive:
\begin{align*}
&\frac{d}{dt} \frac{1}{2} \int_{\Sigma_t} \left[ \abs{\mathscr{E}_{\r}}^2 +  \abs{\mathscr{E}_{A}}^2+ \abs{\mathscr{H}_{\r}}^2+ \abs{\mathscr{H}_{A}}^2 \right] r^3 \sqrt{1+r^2} dr d\omega \\&\quad = \lim_{r\to \infty} \int_{\tilde{\Sigma}_r\cap \Sigma_t} \left(\frac{1}{2} \epsilon^{AB}\mathscr{E}_{A }\mathscr{H}_{B}  \right) r^6d\omega -  \int_{\Sigma_t} \left[ \epsilon^{AB}\mathscr{E}_{A }\mathscr{H}_{B}\right] r^2 \sqrt{1+r^2} dr d\omega \, .
\end{align*}
Here we see the problem. We have no reason a priori to expect that the boundary term on $\tilde{\Sigma}_r\cap \Sigma_t$ vanishes. If it did, we could infer by Gronwall's lemma that the constraints are propagated. We conclude that the form of the propagation equations of Section \ref{sec:be} does not in general propagate the constraints at the boundary if boundary conditions are imposed on both $E_{AB}$ and $E_{A\r}$ (or $H_{AB}$ and $H_{A\r}$).

\subsubsection{The modified equations} \label{sec:reds}

In order to resolve this issue, we have to modify the propagation equations \emph{before} attempting to solve them as a symmetric hyperbolic system. In the previous calculation, the problematic boundary terms arise due to the radial derivatives appearing on the right hand side of \eq{EvolEAr}, \eq{EvolHAr}. We can remove these radial derivatives at the expense of introducing angular derivatives by using the constraint equations \eq{ConsEAr}. It is also convenient to eliminate $E_{\r\r}$ and $H_{\r\r}$ from our equations using the trace constraints. Doing this, we arrive at the modified set of propagation equations:
\begin{align}
\frac{r}{\sqrt{1+r^2}} \frac{\partial E_{A\r}}{\partial t} &=-r\epsilon^{BC}\slashed{\nabla}_B H_{CA}-\frac{\epsilon_A{}^B H_{B\r}}{\sqrt{1+r^2} }, \label{EvolEAr2} \tag{Evol' $E_{A\r}$}\\
\frac{r}{\sqrt{1+r^2}} \frac{\partial E_{AB}}{\partial t} &= \epsilon_{(A}{}^{C}\left[ \frac{1}{\sqrt{1+r^2}}\frac{\partial}{\partial r}\left(r (1+r^2) H_{B)C} \right) - r\slashed{\nabla}_{|C|} H_{B)\r} \right] \label{EvolEAB2}. \tag{Evol' $E_{AB}$} \\
\frac{r}{\sqrt{1+r^2}} \frac{\partial H_{A\r}}{\partial t} &= r\epsilon^{BC}\slashed{\nabla}_B E_{CA}+\frac{\epsilon_A{}^B E_{B \r}}{\sqrt{1+r^2} }, \label{EvolHAr2} \tag{Evol' $H_{A\r}$}\\
\frac{r}{\sqrt{1+r^2}} \frac{\partial H_{AB}}{\partial t} &= -\epsilon_{(A}{}^{C}\left[ \frac{1}{\sqrt{1+r^2}}\frac{\partial}{\partial r}\left(r (1+r^2) E_{B)C} \right) - r\slashed{\nabla}_{|C|} E_{B)\r} \right] \label{EvolHAB2}. \tag{Evol' $H_{AB}$}
\end{align}
This again forms a symmetric hyperbolic system. Taking
\be
r (1+r^2)^{\frac{3}{2}} \left[ E_{A\r}  \times \textrm{\eq{EvolEAr}} + E_{AB}  \times \textrm{\eq{EvolEAB}} + E\leftrightarrow H \right]
\ee
and integrating over $\Sigma_t$ with measure $dr d\omega$ we arrive at
\begin{align}\nonumber
&\frac{d}{dt} \frac{1}{2} \int_{\Sigma_t} \left( \abs{E_{A\r}}^2 +\abs{E_{AB}}^2 +  \abs{H_{A\r}}^2 +\abs{H_{AB}}^2 \right) r^2(1+r^2) dr d\omega \\& \qquad = \lim_{r\to \infty} \int_{\tilde{\Sigma}_r\cap \Sigma_t} \epsilon^{AB}E_{A C}H_{B}{}^C r^6d\omega +2 \int_{\Sigma_t}  \epsilon^{AB}E_{A\r}H_{B\r}\ r (1+r^2) dr d\omega \, , \label{good energy}
\end{align}
which is certainly sufficient to establish a well posedness result for the equations (Evol'), provided we choose boundary conditions such that the term on $\scri$ has a sign. Notice that this will involve imposing conditions only on the $E_{AB}$ (or $H_{AB}$) components, so this formulation of the propagation equations is clearly different to the previous one!\footnote{For instance, specifying $E_{AB}=0$ on $\scri$ will clearly lead to a unique solution of the modified system. On the other hand, the same boundary condition will not completely fix the boundary term occurring in (\ref{naive energy}) in the unmodified formulation, illustrating the severe drawback of the latter.}

 Notice also that, unlike the estimate (\ref{naive energy}) for the (Evol) equations, we now have a bulk term in the energy estimate (\ref{good energy}). For well-posedness this is no significant obstacle, but it will make establishing global decay estimates more difficult. In particular, it is no longer immediate that solutions with Dirichlet boundary conditions remain uniformly bounded globally.

Now let us consider the propagation of the constraints. We of course have to interpret the $E_{\r\r}$ term in $\mathscr{E}_\r$ as $(-E_A{}^A)$ and similarly for $H_{\r\r}$. The evolution equations for the constraints take a simpler form:
\begin{align}
\frac{r}{\sqrt{1+r^2}} \frac{\partial \mathscr{E}_\r}{\partial t} &= -r\epsilon^{AB}\slashed{\nabla}_A \mathscr{H}_B \label{EvolCEr2} \tag{Evol' $\mathscr{E}_{\r}$} \, ,  \\
\frac{r}{\sqrt{1+r^2}} \frac{\partial \mathscr{E}_A}{\partial t} &= -\frac{\epsilon_A{}^B}{2} \left[ r\slashed{\nabla}_B\mathscr{H}_\r + \frac{\mathscr{H}_B}{\sqrt{1+r^2}} \right] \label{EvolCEA2} \tag{Evol' $\mathscr{E}_{A}$} \, , \\
\frac{r}{\sqrt{1+r^2}} \frac{\partial \mathscr{H}_\r}{\partial t} &= r\epsilon^{AB}\slashed{\nabla}_A \mathscr{E}_B \label{EvolCHr2} \tag{Evol' $\mathscr{H}_{\r}$} \, , \\
\frac{r}{\sqrt{1+r^2}} \frac{\partial \mathscr{H}_A}{\partial t} &= \frac{\epsilon_A{}^B}{2} \left[r\slashed{\nabla}_B\mathscr{E}_\r + \frac{\mathscr{E}_B}{\sqrt{1+r^2}} \right] \label{EvolCHA2} \tag{Evol' $\mathscr{H}_{A}$} \, .
\end{align}
Now, once again this is a symmetric hyperbolic system. Taking
\be
r^2 (1+r^2) \left[ \mathscr{E}_{\r} \times \textrm{\eq{EvolCEr2}} + 2 \mathscr{E}_{A} \times \textrm{\eq{EvolCEA2}}+  \mathscr{H}_{\r} \times \textrm{\eq{EvolCHr2}} + 2 \mathscr{H}_{A} \times \textrm{\eq{EvolCHA2}} \right]
\ee
and integrating over $\Sigma_t$ with the measure $dr d\omega$ we can derive:
\begin{align*}
&\frac{d}{dt} \frac{1}{2} \int_{\Sigma_t} \left[ \abs{\mathscr{E}_{\r}}^2 +  2 \abs{\mathscr{E}_{A}}^2+ \abs{\mathscr{H}_{\r}}^2+ 2 \abs{\mathscr{H}_{A}}^2 \right] r^3 \sqrt{1+r^2} dr d\omega \\&\quad =-  2 \int_{\Sigma_t} \left[ \epsilon^{AB}\mathscr{E}_{A }\mathscr{H}_{B}\right] r^2 \sqrt{1+r^2} dr d\omega
\end{align*}
Immediately, with Gronwall's Lemma, we deduce that the constraints, if initially satisfied, will be satisfied for all time.

To our knowledge, identifying the above modified system as the correct formulation to prove well-posedness goes back to Friedrich's work \cite{Friedrich}. In particular, Theorem \ref{spin2 well posedness} above could be inferred from this paper.

\section{Proof of  the main theorems} \label{sec:proofmt}

\subsection{Proof of Theorem \ref{wave full decay} for Spin 0} \label{sec:wd}
The Killing field $T$ immediately gives us a boundedness estimate:
\begin{Proposition}[Boundedness of energy]\label{wave conservation}
Let $u$ be a solution of \eq{cwe} subject to dissipative boundary conditions \eq{wave dissipative} as in Theorem \ref{well posed wave}. Define the energy to be:
\begin{equation} \label{Etud}
E_t[u] := \frac{1}{2} \int_{\Sigma_{t}} \left( \frac{\left( \p_t u\right)^2 + u^2  }{1+r^2}  + \left(1+r^2 \right)\left( \tilde{\p}_r u\right)^2 + \abs{\slashed{\nabla} u}^2 \right) r^2 dr d\omega.
\end{equation}
Then we have for any $T_1<T_2$:
\be
E_{T_2}[u] + \frac{1}{2} \int_{\tilde{\Sigma}_\infty^{[T_1, T_2]}} \left(r^2  (\p_t u)^2 + r^6 (\tilde{\p}_r u)^2\right) d\omega dt = E_{T_1}[u].
\ee
\end{Proposition}
\begin{proof}
We have that 
\be
\textrm{Div} \left({}^T\J \right)= 0.
\ee
Integrating this over $S_{[T_1, T_2]}$ we pick up terms from $\Sigma_{T_1}, \Sigma_{T_2}$ and $\tilde{\Sigma}_\infty^{[T_1, T_2]}$. A straightforward calculation shows
\be
\int_{\Sigma_t} {}^T\J_a n^a dS_{\Sigma_t} = E_t[u]
\ee
We also find
\begin{align*}
\int_{\tilde{\Sigma}_r^{[T_1, T_2]}} {}^T\J_a m^a dS_{\tilde{\Sigma}_r} &= \int_{\mathbb S^2} r^2(1+r^2) \left( \p_t u\right ) (\tilde{\p}_r u) d\omega dt \\
&= -\frac{1}{2} \int_{\mathbb S^2} \Bigg[ \left(r^2  (\p_t u)^2 + r^2(1+r^2)^2 (\tilde{\p}_r u)^2\right) \\
&\qquad - \left\{ \p_t( r u) + r (1+r^2) (\tilde{\p}_r u) \right\}^2 \Bigg] d\omega dt.
\end{align*}
As $r \to \infty$, the term in braces vanishes by the boundary condition and we find
\be
\lim_{r \to \infty} \int_{\tilde{\Sigma}_r^{[T_1, T_2]}} {}^T\J_a m^a dS_{\tilde{\Sigma}_r} =-\frac{1}{2} \int_{\tilde{\Sigma}_\infty^{[T_1, T_2]}} \left(r^2  (\p_t u)^2 + r^6 (\tilde{\p}_r u)^2\right) d\omega dt.
\ee
Applying Lemma \ref{Divergence Theorem}, we are done.
\end{proof}
We next show an integrated decay estimate with a loss in the weight at infinity:
\begin{Proposition}[Integrate decay estimate with loss]\label{Morawetz wave}
Let $T_2>T_1$. Suppose $u$ is a solution of \eq{cwe} subject to \eq{wave dissipative} as in Theorem \ref{well posed wave}. Then the estimate
\begin{align*}
 \int_{S_{[T_1,T_2]}} \left( \frac{\left( \p_t u\right)^2 + u^2  }{1+r^2}  + \left(1+r^2 \right)\left( \tilde{\p}_r u\right)^2 + \abs{\slashed{\nabla} u}^2 \right) \frac{r^2}{\sqrt{1+r^2}} dr d\omega dt & \\ +  \int_{\tilde{\Sigma}_\infty^{[T_1, T_2]}} \left(r^2  (\p_t u)^2 + \abs{r^2 \slashed{\nabla} u}^2 + r^6 (\tilde{\p}_r u)^2\right) d\omega dt  & \leq C E_{T_1}[u]
\end{align*}
holds for some constant $C>0$, independent of $T_1$ and $T_2$.
\end{Proposition}

\begin{proof}
We integrate a current constructed from the (renormalized) energy momentum tensor (\ref{reem}) and a radial vector field. The current is
\be
J_{X, w_1, w_2}^a = {}^X\J^a +\frac{w_1}{\sqrt{1+r^2}}  u  \tilde{\nabla}^a u +  w_2 u^2X^a \, ,
\ee
where $X$ is the radial vector field defined in (\ref{Xdef}). The proof of the theorem is a straightforward corollary of the following two Lemmas. 

\begin{Lemma}
We have
\begin{align*}
\textrm{Div } J_{X, w_1, w_2} &=  \left ({}^T\J\cdot e_0\right ) - \left( X^b \tilde{\nabla}_b 1 - \sqrt{1+r^2} + \frac{w_1}{\sqrt{1+r^2}} \right)  \T^c{}_c \\
&\qquad + \left( \nabla_a \left( \frac{w_1}{\sqrt{1+r^2}} \right) + 2 w_2  X_a\right)u \tilde{\nabla}^a u \\
&\qquad + \left( (1+r^2) \textrm{Div}\left (  \frac{w_2}{1+r^2} X\right ) - \frac{ w_1}{(1+r^2)^{\frac{3}{2}}}  \right) u^2.
\end{align*}
In particular, if we take $w_1 = 1$, $w_2 = \frac{1}{2}(1+r^2)^{-1}$ then we have
\be
\textrm{Div } J_{X, w_1, w_2} =  \left({}^T\J\cdot e_0\right) + \frac{u^2}{2(1+r^2)^{\frac{3}{2}}}.
\ee
\end{Lemma}

\begin{proof}
The vector field $X$ has the 
deformation tensor:
\be
{}^X\pi = \frac{(e^0)^2}{\sqrt{1+r^2}}  +  g \sqrt{1+r^2},
\ee
so that 
\be
\textrm{Div } {}^X\J = ({}^T\J\cdot e_0)  - \left( X^b \tilde{\nabla}_b 1 - \sqrt{1+r^2}\right)  \T^c{}_c.
\ee
We also require the observation that
\be
\sqrt{1+r^2}\nabla^a \left(\frac{\tilde{\nabla}_a u}{\sqrt{1+r^2}} \right) = (\Box_g u + 2 u) + \frac{u}{1+r^2}.
\ee
Thus when $u$ solves the conformal wave equation, we have:
\begin{align*}
\textrm{Div } \left(  \frac{w_1}{\sqrt{1+r^2}}  u  \tilde{\nabla}^a u \right) =  - \frac{w_1}{\sqrt{1+r^2}}  \T^c{}_c + \nabla_a \left( \frac{w_1}{\sqrt{1+r^2}} \right) u \tilde{\nabla}^a u - \frac{u^2 w_1}{(1+r^2)^{\frac{3}{2}}}.
\end{align*}
Finally, we have
\be
\textrm{Div }(  w_2 u^2 X)   = 2 w_2 u X^a  \tilde{\nabla}_a u + (1+r^2) \textrm{Div }\left (  \frac{w_2}{1+r^2} X\right ) u^2.
\ee
Combining these identities we have the first part of the result. We can arrange that the term proportional to $T^c{}_c$ vanishes by taking $w_1 = 1$. The term proportional to $u\tilde{\nabla}_a u$ vanishes if $w_2 = \frac{1}{2} (1+r^2)^{-1}$, and the final part of the result follows from a brief calculation.
\end{proof}

\begin{Lemma}
With $w_1$, $w_2$ chosen as in the second part of the previous Lemma, we have:
\be
\abs{\int_{\Sigma_{T_2}} {}{J}^a_{X, w_1, w_2} n_a dS_{\Sigma_{T_2}}} \leq  C E_{T_2}[u] \leq CE_{T_1}[u] 
\ee
and
\be
\lim_{r \to \infty} \int_{\tilde{\Sigma}_r^{[T_1, T_2]}} J^a_{X, w_1, w_2} m_a dS_{\tilde{\Sigma}_r} +\frac{1}{2} \int_{\tilde{\Sigma}_\infty^{[T_1, T_2]}}\abs{ r^2\slashed{\nabla} u}^2 d\omega dt   \leq CE_{T_1}[u]. 
\ee
\end{Lemma}
\begin{proof}
We calculate
\be
\int_{\Sigma_t} {J}^a_{X, w_1, w_2} n_a dS_{\Sigma_t} = -\int_{\Sigma_t} \frac{\p_t u}{(1+r^2)^{\frac{3}{2}}} \left( r (1+r^2) \tilde{\p}_r u + u \right) r^2 dr d\omega
\ee
which, after applying Cauchy-Schwarz, can certainly be controlled by the energy $E_t[u] $, which in turn is controlled by $E_0[u] $ using Theorem \ref{wave conservation}.

For the other surface terms, we calculate
\begin{align*}
\int_{\tilde{\Sigma}_r^{[T_1, T_2]}} J^a_{X, w_1, w_2} m_a dS_{\tilde{\Sigma}_r} &= \frac{1}{2} \int_{\tilde{\Sigma}_r^{[T_1, T_2]}}  \Bigg[ \frac{r^2(\p_t u)^2 }{1+r^2} + r^2 (1+r^2) (\tilde{\p}_r u)^2+\\ &\qquad \quad + 2 r u (\tilde{\p}_r u)- \abs{ r\slashed{\nabla} u}^2  \Bigg] r \sqrt{1+r^2} dt d\omega
\end{align*}
so that
\be
\lim_{r \to \infty} \int_{\tilde{\Sigma}_r^{[T_1, T_2]}} J^a_{X, w_1, w_2} m_a dS_{\tilde{\Sigma}_r} =\frac{1}{2} \int_{\tilde{\Sigma}_\infty^{[T_1, T_2]}} \left(r^2  (\p_t u)^2  -\abs{r^2 \slashed{\nabla} u}^2 + r^6 (\tilde{\p}_r u)^2\right)  d\omega dt 
\ee
and the result follows since we already control the time and radial derivatives of $u$ on the boundary by Theorem \ref{wave conservation}.
\end{proof}
This concludes the proof of Proposition \ref{Morawetz wave}, and establishes the claimed degenerate integrated decay without derivative loss result.
\end{proof}
We next improve the radial weight of the spacetime term in the integrated decay estimate, at the expense of losing a derivative.\footnote{We shall take the frugal approach of commuting with the timelike Killing field. If one is happy to exploit the angular symmetries of AdS, our approach can be simplified.}

\begin{Proposition} [Higher order estimates]\label{prop:elliw}
Let $u$ be a solution of \eq{cwe} subject to \eq{wave dissipative} as in Theorem \ref{well posed wave}. Then the estimate
\begin{align*}
& \int_{S_{[T_1,T_2]}} (1+r^2)^{\frac{3}{2}} \left[ \left[ \tilde{\p}_r \left(r^2 \tilde{\p}_r u \right)\right]^2 + \frac{r^2}{1+r^2}\left( \abs{\tilde{\p}_r\left[ r\slashed{\nabla} u\right]}^2 +  \ \abs{\tilde{\p}_r\left[ \p_t u\right]}^2\right)\right]dr d\omega dt \\
&\qquad+\int_{S_{[T_1,T_2]}} \frac{r^4}{\sqrt{1+r^2}}|\slashed{\nabla}^2 u|^2 dr d\omega dt  \leq C (E_{T_1}[u]+ E_{T_1}[u_t])
\end{align*}
holds for some constant $C>0$, independent of $T_1$ and $T_2$.
\end{Proposition}


\begin{proof}
Let us define, for a solution of the conformal wave equation:
\begin{align*}
Lu &:= \frac{1}{1+r^2 }\left(u_{tt} + u \right) =  \frac{\sqrt{1+r^2}}{r^2} \p_r \left(r^2 \p_r \left(u\sqrt{1+r^2} \right) \right)+ \slashed{\Delta} u \, .
\end{align*}
By commuting with $T$ and applying Proposition \ref{Morawetz wave}, we have
\ben{L estimate}
\int_{S_{[T_1,T_2]} }\left\{ \left[ Lu\right]^2  \frac{r^2}{\sqrt{1+r^2}} \right\} d\eta \leq C \left(E_{T_1}[u] + E_{T_1}[u_t] \right) \, .
\een
We can expand the integrand to give
\begin{align} \label{crosst}
 \left[ Lu\right]^2  \frac{r^2}{\sqrt{1+r^2}} = \frac{r^2}{\sqrt{1+r^2}} \Bigg \{&\left[\frac{\sqrt{1+r^2}}{r^2} \p_r \left(r^2 \p_r \left(u\sqrt{1+r^2} \right) \right)\right]^2 + \left[\slashed{\Delta} u\right]^2 \nonumber \\
 & + 2 \left[\frac{\sqrt{1+r^2}}{r^2} \p_r \left(r^2 \p_r \left(u\sqrt{1+r^2} \right) \right)\right] \left[\slashed{\nabla}_A \slashed{\nabla}^A u\right] \Bigg \}
\end{align}
We clearly have two terms with a good sign and a cross term. To deal with the cross term, we integrate by parts twice, so that we obtain a term (with a good sign) that looks like $\abs{\partial_r \slashed{\nabla} u}^2$ and some lower order terms. More explicitly, we have
\begin{Lemma}\label{elliptic wave}
Let $K$ be the vector field given by
\begin{align*}
K &=  \p_r \left(r^2 \p_r \left(u\sqrt{1+r^2} \right) \right) (\slashed{\nabla}^Au)\,  e_A -  (r \slashed{\nabla}_A u) \tilde{\p}_r \left(r \slashed{\nabla}^Au\right)\, e_\r  - \frac{r }{2(1+r^2)} \abs{r \slashed{\nabla} u}^2 e_\r .
\end{align*}
Then we have
\be
\p_r \left(r^2 \p_r \left(u\sqrt{1+r^2} \right) \right)\left[\slashed{\nabla}_A \slashed{\nabla}^A u\right] = \sqrt{1+r^2} \abs{\tilde{\p}_r\left( r\slashed{\nabla} u\right)}^2 + \frac{3 \abs{ r\slashed{\nabla} u}^2}{2 (1+r^2)^{\frac{3}{2}}} + \textrm{Div }K.
\ee
\end{Lemma}
\begin{proof}
See Appendix \ref{elliptic wave proof}.
\end{proof}

To prove Proposition \ref{prop:elliw}, we simply insert (\ref{crosst}) into \eq{L estimate} and handle the cross-term using Lemma \ref{elliptic wave}. The boundary terms coming from $\textrm{Div }K$ are
\be
\int_{S_{[T_1,T_2]}} \textrm{Div }K d\eta = -\frac{1}{2} \int_{\tilde{\Sigma}_\infty^{[T_1,T_2]}} \abs{r^2 \slashed{\nabla} u}^2 d\omega dt
\ee
which we control by the estimate of Proposition \ref{Morawetz wave}. 

This in particular controls the term $\left[\slashed{\Delta} u\right]^2$. Since $\mathbb S^2$ has constant positive Gauss curvature, we have the following elliptic estimates (see for instance Corollary 2.2.2.1 in \cite{ChrKla}):
$$\int_{\mathbb S^2} |\slashed{\nabla}^2 u|^2\leq C\int_{\mathbb S^2}\left[\slashed{\Delta} u\right]^2$$
from which we obtain the desired bounds for $|\slashed{\nabla}^2 u|^2$.

Finally, the $(\tilde{\p}_r\p_t u)^2$ term appearing in Proposition \ref{prop:elliw} is directly controlled by the $T$-commuted version of the estimates in Proposition \ref{Morawetz wave}.

\end{proof}

We finally improve the weight in the spacetime term of Proposition \ref{Morawetz wave} making use of the fact that by Proposition \ref{prop:elliw} we now control radial derivatives of $\p_t u, \slashed{\nabla}{u}$ and $\tilde{\p}_r u$ which lead to improved zeroth order terms through a Hardy inequality. This is a standard result, but for convenience we include here a proof.
\begin{Lemma}[Hardy's inequality]\label{wave Hardy}
Fix $a \neq 0$. Suppose that $f:[1, \infty) \to \R$ is smooth, $f(1)=0$ and $\abs{f(r)}^2 r^{-a}\to 0$ as $r\to \infty$. Then we have the estimate
\be
\int_1^\infty \abs{f}^2 r^{-1-a} dr \leq \frac{4}{a^2} \int_1^\infty \abs{\p_r f}^2 r^{1-a} dr,
\ee
provided the right hand side is finite.
\end{Lemma}
\begin{proof}
We write
\begin{align*}
\int_1^\infty \abs{f}^2 r^{-1-a}  dr &= \int_1^\infty f^2 \frac{d}{dr}\left(-\frac{r^{-a}}{a} \right)dr \\
&= -\frac{1}{a} \left[ \abs{f}^2 r^{-a} \right]_1^\infty + \frac{2}{a} \int_1^\infty f \p_r f  r^{-a} dr \\
&=   \frac{2}{a} \int_1^\infty f \p_r f  r^{-a} dr 
\end{align*} 
Here we have used $f(1)=0$ and the fact that $\lim_{r\to \infty} r^{-a}|f(r)|^2=0$ to discard the boundary terms. Now applying Cauchy-Schwarz, we deduce
\be
\int_1^\infty \abs{f}^2 r^{-1-a}  dr  \leq    \left( \frac{4}{a^2} \int_1^\infty \abs{f}^2 r^{-1-a} dr \int_1^\infty \abs{\p_r f}^2 r^{1-a} dr\right)^{\frac{1}{2}},
\ee
whence the result follows.
\end{proof}
 
From Lemma \ref{wave Hardy} we establish:
\begin{Theorem}[Full integrated decay]\label{elliptic+hardy wave}
Let $u$ be a smooth function such that  $\abs{r u}$ is bounded. Then the estimate
\begin{align*}
&\int_{S_{[T_1,T_2]}} \left( \frac{\left( \p_t u\right)^2 + u^2  }{1+r^2}  + \left(1+r^2 \right)\left( \tilde{\p}_r u\right)^2 + \abs{\slashed{\nabla} u}^2 \right) r^2 dr d\omega dt\\
&\leq C\Bigg[ \int_{S_{[T_1,T_2]}} (1+r^2)^{\frac{3}{2}} \left[ \left[ \tilde{\p}_r \left(r^2 \tilde{\p}_r u \right)\right]^2 + \frac{r^2}{1+r^2}\left( \abs{\tilde{\p}_r\left[ r\slashed{\nabla} u\right]}^2 +  \ \abs{\tilde{\p}_r\left[ \p_t u\right]}^2\right)\right]dr d\omega dt \\
&\qquad +\int_{S_{[T_1,T_2]}} \left( \frac{\left( \p_t u\right)^2 + u^2  }{1+r^2}  + \left(1+r^2 \right)\left( \tilde{\p}_r u\right)^2 + \abs{\slashed{\nabla} u}^2 \right) \frac{r^2}{\sqrt{1+r^2}} dr d\omega dt\Bigg].
\end{align*}
holds for some constant $C>0$ independent of $T_1$ and $T_2$. If, moreover, $u$ solves  \eq{cwe} subject to \eq{wave dissipative}, then the right-hand side may be bounded by $C' (E_{T_1}[u]+E_{T_1}[\partial_t u])$ for some $C'>0$ independent of $T_1$ and $T_2$.
\end{Theorem}
\begin{proof}
By introducing a cut-off we can quickly reduce to showing that the estimate holds for $u$ supported either on $r\leq 2$ or on $r\geq 1$. For $u$ supported on $r\leq 2$, the estimate follows immediately, since the first order terms on the right hand side are comparable to those on the left hand side on any finite region. For $u$ supported on $r\geq 1$, we first apply Lemma  \ref{wave Hardy} with $f = r^2\sqrt{1+r^2} \tilde{\p}_r u$ and $a=1$ to deduce
\begin{align*}
\int_{S_{[T_1,T_2]}}  \abs{\tilde{\p}_r u}^2 r^2(1+r^2) dr d\omega dt&\leq4 \int_{S_{[T_1,T_2]}}  \left[\p_r ( r^2\sqrt{1+r^2} \tilde{\p}_r u) \right]^2 dr d\omega dt \\
&\leq  4 \int_{S_{[T_1,T_2]}} (1+r^2)\left[ \tilde{\p}_r \left(r^2 \tilde{\p}_r u \right)\right]^2 dr d\omega dt \\
&\leq 4 \int_{S_{[T_1,T_2]}} (1+r^2)^{\frac{3}{2}} \left[ \tilde{\p}_r \left(r^2 \tilde{\p}_r u \right)\right]^2 dr d\omega dt.
\end{align*}
Similarly, applying Lemma  \ref{wave Hardy} to $f = \sqrt{1+r^2} \p_t u$ with $a=1$ we deduce
\begin{align*}
\int_{S_{[T_1,T_2]}}  \abs{\p_t u}^2 \frac{r^2}{1+r^2} dr d\omega dt
&\leq \int_{S_{[T_1,T_2]}}  \abs{\p_t u}^2 r^{-2}(1+r^2) dr d\omega \\&\leq 4 \int_{S_{[T_1,T_2]}} \abs{\tilde{\p}_r [\p_t u]}^2 (1+r^2) dr d\omega dt \\
&\leq 8 \int_{S_{[T_1,T_2]}} \abs{\tilde{\p}_r [\p_t u]}^2 r^2 \sqrt{1+r^2} dr d\omega dt,
\end{align*}
where we've used that $u$ is supported on $r\geq 1$. A similar calculation gives the estimate for $\slashed{\nabla} u$. Finally, making use of Propositions \ref{Morawetz wave}, \ref{prop:elliw}, we see that if $u$ satisfies the equation, then we can bound the right hand side in terms of  $E_{T_1}[u]+E_{T_1}[\partial_t u]$.
\end{proof}

Combining Proposition \ref{prop:elliw} with Theorem \ref{elliptic+hardy wave} we have established the claimed non-degenerate integrated decay with derivative loss result for the wave equation.

\subsection{Proof of Theorem \ref{wave full decay} for Spin 1} \label{sec:md}

The proof of the theorem for the Maxwell field follows a similar pattern to that of the conformal scalar field. There is a simplification owing to the fact that the energy-momentum tensor is trace-free, and the elliptic estimate takes a slightly different form.

\begin{Proposition}[Boundedness of energy]\label{maxee}
Let $T_2 >T_1$. Suppose that $F$ is a solution of Maxwell's equations (\ref{evolEr}-\ref{consH}), subject to the dissipative boundary conditions \eq{dcmax} as in Theorem \ref{well posed}. Then we have:
\bean
&& \int_{\Sigma_{T_2}} \left( \abs{E}^2  + \abs{H}^2  \right) r^2 dr d\omega +  \int_{\tilde{\Sigma}^{[T_1, T_2]}_{\infty}} \left( \abs{r^2 E_A}^2+ \abs{r^2 H_A}^2\right)  dt d\omega \\ &&\qquad =  \int_{\Sigma_{T_1}} \left( \abs{E}^2  + \abs{H}^2  \right) r^2 dr d\omega.
\eean
\end{Proposition}
\begin{proof}
We apply Lemma \ref{Divergence Theorem} to the vector field $\J^T[F]_a = \T[F]_{ab} T^b$. Consider the term on $\scri$. We have
\begin{align*}
 \lim_{r\to \infty} \int_{\tilde{\Sigma}^{[T_1, T_2]}_{r}} {}^T\J_a  m^a dS_{\tilde{\Sigma}_r} &=  \lim_{r\to \infty} \int_{\tilde{\Sigma}^{[T_1, T_2]}_{r}}  r^2 (1+r^2) \epsilon^{AB}E_A H_B dt d\omega \\ &= \lim_{r\to \infty} \int_{\tilde{\Sigma}^{[T_1, T_2]}_{r}}  r^4 \epsilon^{AB}(-\epsilon_A{}^C H_C) H_B dt d\omega\\
 &= -\lim_{r\to \infty} \int_{\tilde{\Sigma}^{[T_1, T_2]}_{r}}  \abs{r^2 H_A}^2 dt d\omega
 \\
 &= -\frac{1}{2} \int_{\tilde{\Sigma}^{[T_1, T_2]}_{\infty}}  \left( \abs{r^2 E_A}^2 + \abs{r^2 H_A}^2\right) dt d\omega.
\end{align*}
Here we have used the dissipative boundary conditions \eq{dcmax}. Since $\nabla_a {}^T\J^a=0$, there is no bulk term, and a simple calculation gives
\be
\int_{\Sigma_{t}} {}^T\J_a n^a dS_{\Sigma_{T_2}} = \frac{1}{2} \int_{\Sigma_{t}} \left( \abs{E}^2  + \abs{H}^2  \right) r^2 dr d\omega,
\ee
which completes the proof.
\end{proof}
We next show an integrated decay estimate with a loss in the weight at infinity:
\begin{Proposition}[Integrated decay estimate with loss]\label{maxwell ILED}
Let $T_2 >T_1$. Suppose that $F$ is a solution of Maxwell's equations (\ref{evolEr}-\ref{consH}), subject to the dissipative boundary conditions \eq{dcmax} as in Theorem \ref{well posed}. Then we have:
\bean
&& \int_{S_{[T_1, T_2]} } \left( \abs{E}^2  + \abs{H}^2  \right) \frac{r^2}{\sqrt{1+r^2}} dr d\omega dt + \int_{\tilde{\Sigma}^{[T_1, T_2]}_{\infty}} \left( \abs{r^2 E}^2  + \abs{r^2 H}^2  \right)  dt d\omega \\ &&\qquad \leq  3 \int_{\Sigma_{T_1}} \left( \abs{E}^2  + \abs{H}^2  \right) r^2 dr d\omega.
\eean
\end{Proposition}
\begin{proof}
We apply the divergence theorem to integrate the current ${}^X\J^a = \T^a{}_b X^b$ over $S_{[T_1, T_2]}$. Recalling the expression \eq{deftX} for ${}^X\pi$, and noting the fact that the energy-momentum tensor is traceless, we have
\begin{align*}
\textrm{Div }{}^X\J &= {}^X\pi_{ab} \T^{ab} \\&= \frac{\T_{00}}{\sqrt{1+r^2}} + \T_a{}^a \sqrt{1+r^2} \\
&= \frac{ \abs{E}^2+\abs{H}^2  }{\sqrt{1+r^2}}.
\end{align*}
Now consider the flux through a spacelike surface $\Sigma_t$. We have
\begin{align*}
\int_{\Sigma_t} {}^X\J_a n^a dS_{\Sigma_t} &= \int_{\Sigma_t} \T_{0 \r} \frac{r^3}{\sqrt{1+r^2}} dr d\omega \\
&= \int_{\Sigma_t} \epsilon^{AB}E_A H_B \frac{r^3}{\sqrt{1+r^2}} dr d\omega
\end{align*}
so that if $t\in[T_1, T_2]$ we have by Proposition \ref{maxee}:
\be
\abs{\int_{\Sigma_{t}} {}^X\J_a n^a dS_{\Sigma_{t}}} \leq \frac{1}{2} \int_{\Sigma_{t}} \left( \abs{E}^2 + \abs{H}^2 \right) r^2 dr d\omega \leq \frac{1}{2} \int_{\Sigma_{T_1}} \left( \abs{E}^2 + \abs{H}^2 \right) r^2 dr d\omega.
\ee
Next consider the flux through a surface of constant $r$. We have:
\begin{align*}
\int_{\tilde{\Sigma}^{[T_1, T_2]}_r} {}^X\J_a m^a dS_{\tilde{\Sigma}_r} &= \int_{\tilde{\Sigma}^{[T_1, T_2]}_r}  \T_{\r \r}  {r^3}{\sqrt{1+r^2}} dt d\omega \\
&= \int_{\tilde{\Sigma}^{[T_1, T_2]}_r}  \left(-\abs{E_\r}^2 -\abs{H_\r}^2 + \abs{E_A}^2 + \abs{H_A}^2 \right) {r^3}{\sqrt{1+r^2}} dt d\omega 
\end{align*}
so that
\begin{align*}
\lim_{r\to \infty} \int_{\tilde{\Sigma}^{[T_1, T_2]}_r} {}^X\J_a m^a dS_{\tilde{\Sigma}_r} &\leq - \int_{\tilde{\Sigma}^{[T_1, T_2]}_{\infty}} \left( \abs{r^2 E}^2+ \abs{r^2 H}^2\right)  dt d\omega \\& \qquad + 2  \int_{\Sigma_{T_1}} \left( \abs{E}^2 + \abs{H}^2 \right) r^2 dr d\omega,
\end{align*}
where we use Proposition \ref{maxee} to control the angular components of $E, H$ at infinity. Integrating $\textrm{Div }{}^X\J$ over $S_{[T_1, T_2]}$, applying the divergence theorem and using the estimates above for the fluxes completes the proof of the Proposition.
\end{proof}

We next seek to improve the weight at infinity at the expense of a derivative loss. The first stage in doing this is an elliptic estimate. First note that by commuting the equations with the Killing field $T$ we have:
\bea\label{max T ILED}
&& \int_{S_{[T_1, T_2]} } \left( \abs{\p_t E}^2  + \abs{\p_t H}^2  \right) \frac{r^2}{\sqrt{1+r^2}} dr d\omega dt + \int_{\tilde{\Sigma}^{[T_1, T_2]}_{\infty}} \left( \abs{r^2 \p_t E}^2  + \abs{r^2 \p_t H}^2  \right)  dt d\omega \\ &&\qquad \leq  3 \int_{\Sigma_{T_1}} \left( \abs{\p_t E}^2  + \abs{\p_t H}^2  \right) r^2 dr d\omega.
\eea
We can use the evolution equations to control the right hand side in terms of spatial derivatives in the slice $\Sigma_{T_1}$ if we choose. We will require the following Lemma:
\begin{Lemma}\label{maxwell cross}
Let $K$ be the vector field
\be
K := \frac{2r}{\sqrt{1+r^2}} H_\r \left[ \left(r\slashed{\nabla}^AH_A\right) e_\r -\frac{\partial }{\partial r}\left ( r \sqrt{1+r^2} H^A \right) e_A \right] - \frac{r^3}{1+r^2} H_\r^2 e_\r.
\ee
Then if $H$ satisfies the constraint equation \eq{consH} we have the identity
\begin{align*}
-\frac{2}{\sqrt{1+r^2}}  \partial_r \left( r \sqrt{1+r^2} H_A \right)& r\slashed{\nabla}^A H_\r \\&=  2\frac{\sqrt{1+r^2}}{r^2} \abs{\frac{\partial }{\partial r}\left ( r^2H_\r \right)}^2 + \frac{r^2}{\left(1+r^2\right)^{\frac{3}{2}}} \abs{H_\r}^2 + \textrm{Div } K
\end{align*}
\end{Lemma}
\begin{proof}
See Appendix \ref{maxwell cross proof}.
\end{proof}

%

Now consider the Maxwell equation \eq{evolEA}:
\be
\frac{r}{\sqrt{1+r^2}} \partial_t  E_A = \epsilon_A{}^B \left[ \partial_r \left( r \sqrt{1+r^2} H_B \right) - r\slashed{\nabla}_B H_\r \right],
\ee
squaring this and multiplying by $(1+r^2)^{-\frac{1}{2}}$, we have
\bea \nonumber
\frac{r^2 \abs{\p_t E_A}^2}{\left({1+r^2}\right)^{\frac{3}{2}}} &=& \frac{1}{\sqrt{1+r^2}} \left[ \abs{r \slashed{\nabla}_A H_\r}^2 + \abs{\frac{\partial }{\partial r}\left ( r \sqrt{1+r^2} H_A \right)}^2   \right] \\&& -\frac{2}{\sqrt{1+r^2}}  \partial_r \left( r \sqrt{1+r^2} H_A \right) r\slashed{\nabla}^A H_\r  \label{elliptic}.
\eea

From here we readily conclude:
\begin{Theorem}[Higher order estimates]\label{derivative decay}
Suppose that $F$ is a solution of Maxwell's equations, as in Theorem \ref{well posed}. Then there exists a constant $C>0$, independent of $T_1$ and $T_2$ such that we have
\begin{align*}
\int_{S_{[T_1,T_2]}}&\Bigg \{  \frac{1}{\sqrt{1+r^2}} \left[ \abs{r \slashed{\nabla}_A H_\r}^2 + \abs{\frac{\partial }{\partial r}\left ( r \sqrt{1+r^2} H_A \right)}^2  + \abs{r \slashed{\nabla}_A H_B}^2 \right] \\& + 2\frac{\sqrt{1+r^2}}{r^2} \abs{\frac{\partial }{\partial r}\left ( r^2H_\r \right)}^2  + H \leftrightarrow E  \Bigg \} r^2 dt dr d\omega \\
& \leq C  \int_{\Sigma_{T_1}} \left(  \abs{E}^2 + \abs{H}^2 + \abs{\p_t E}^2  + \abs{\p_t H}^2 \right) r^2 dr d\omega.
\end{align*}
\end{Theorem}
\begin{proof}
We integrate the identity \eq{elliptic} over $S_{[T_1,T_2]}$. We control the left hand side by the time commuted integrated decay estimate \eqref{max T ILED} since $r^2(1+r^2)^{-1}\leq 1$. The right hand side, after making use of Lemma \ref{maxwell cross} and applying the divergence theorem will give us good derivative terms, a good zero'th order term which we can ignore and a surface term. The surface term at infinity gives\footnote{Notice that there is no surface term contribution at infinity arising from $(\f{2r^2}{\sqrt{1+r^2}} H_{\bar{r}}\slashed{\nabla}^A H_A)e_{\bar r}$ due to the decay of the angular derivative of $H$ in $r$ by the local well-posedness theorem (Theorem \ref{well posed}).} a term proportional to $\abs{r^2 H_\r}^2$, integrated over the cylinder, which we control with the estimate in Theorem \ref{maxwell ILED}. This immediately gives the result for all of the terms except the term $\abs{r\slashed{\nabla}_A H_B}^2$, which we obtain from a standard elliptic estimate on the sphere (see for instance Proposition 2.2.1 in \cite{ChrKla}) :
$$\int_{\mathbb S^2} \abs{\slashed{\nabla}_A H_B}^2\leq C \int_{\mathbb S^2} |\epsilon^{AB} \slashed{\nabla}_A H_B|^2+|\slashed{\nabla}^A H_A|^2,$$
after noticing that we already control $\epsilon^{AB} \slashed{\nabla}_A H_B$ and $\slashed{\nabla}^A H_A$ with a suitable weight from (\ref{evolEr}) and (\ref{consH}). Finally, we note that the estimate can be derived in an identical manner for $E$.
\end{proof}

\begin{Theorem}[Full integrated decay] \label{full integrated decay}
Suppose that $F$ is a solution of Maxwell's equations, as in Theorem \ref{well posed}. Then we have
\bean
&& \int_{S_{[T_1,T_2]} } \left[ \abs{E_\r}^2 + \abs{H_\r}^2+ \abs{E_A}^2+  \abs{H_A}^2  \right]r^2   dt dr d\omega \\ &&\qquad \leq  C \int_{\Sigma_{T_1}} \left( \abs{E_\r}^2 + \abs{E_A}^2+\abs{H_\r}^2 + \abs{H_A}^2 + \abs{\dot E_\r}^2 + \abs{\dot E_A}^2+\abs{\dot H_\r}^2 + \abs{\dot H_A}^2  \right) r^2 dr d\omega
\eean
for some $C>0$ independent of $T_1$ and $T_2$.
\end{Theorem}
\begin{proof}
We combine the result of Theorem \ref{derivative decay} with the Hardy estimates of Lemma \ref{wave Hardy} with $a=1$ to improve the weights near infinity in the integrated decay estimates, making use of a cut-off in much the same way as for the spin 0 problem.
\end{proof}

\subsection{Proof of Theorem \ref{wave full decay} for Spin 2} \label{sec:bd}
As a useful piece of notation we first introduce the trace-free part of $E_{AB}$ by
\be
\hat{E}_{AB} = E_{AB} - \frac{1}{2} \delta_{AB} E_C{}^C = E_{AB} + \frac{1}{2} \delta_{AB} E_{\r\r},
\ee
and similarly for $\hat{H}_{AB}$. The boundary conditions of Theorem \ref{spin2 well posedness} may then be conveniently expressed as
\be
r^3 \left( \hat{E}_{AB} + \epsilon_{(A}{}^{C} \hat{H}_{B) C}\right) \to 0, \qquad \textrm{ as } r \to \infty \, ,
\ee
and we also have $|E_{AB}|^2 = |\hat{E}_{AB}|^2 + \frac{1}{2} |E_{\bar{r} \bar{r}}|^2$.
We will first prove the boundedness and the degenerate integrated decay statement of the main theorem.\footnote{As we will see in the proof, the  key difference to the case of the wave- and Maxwell's equations is that here we will have to establish both these estimates at the same time!}
\begin{Proposition}[Boundedness of energy and integrated decay estimate with loss]\label{weyl derivative decay}
Let $T_2>T_1$. Suppose $E_{ab}$ and $H_{ab}$ are solutions to the spin $2$ equations subject to \eqref{dcbia} as in Theorem \ref{spin2 well posedness}, there exists a constant $C>0$, independent of $T_1$ and $T_2$ such that we have
\begin{align}
\int_{S_{[T_1, T_2]}}&\left\{ E_{ab}E^{ab} +H_{ab} H^{ab} \right\} r^2\sqrt{1+r^2} dt dr d\omega + \int_{\Sigma_{T_2}}\left\{ E_{ab}E^{ab} +H_{ab} H^{ab} \right\} r^2(1+r^2)  dr d\omega\nonumber \\ \label{spin2 Morawetz}
&\quad   + \int_{\tilde{\Sigma}_\infty^{[T_1, T_2]}}\left\{ E_{ab} E^{ab} + H_{ab}  H^{ab} \right\} r^6 dt d\omega
\\& \qquad \leq C  \int_{\Sigma_{T_1}}\left\{ E_{ab}E^{ab} +H_{ab} H^{ab} \right\} r^2(1+r^2)  dr d\omega. \nonumber
\end{align}
\end{Proposition}

\begin{proof}
Let us now introduce the vector field
\be
Y = T + \frac{1}{\sqrt{3}} X.
\ee
We certainly have that $Y$ is timelike, since
\be
g(Y, Y) = -(1+r^2) + \frac{1}{3} r^2 = -1 -\frac{2}{3}r^2 <0.
\ee
Moreover, we have
\be
{}^Y \pi = \frac{1}{\sqrt{3}} \left( \frac{(e^0)^2}{\sqrt{1+r^2} }+ g \sqrt{1+r^2} \right)
\ee
The current that we shall integrate over $S_{[T_1, T_2]}$ is
\be
J^a = Q^a_{bcd} Y^b Y^c Y^d. 
\ee
Clearly, we have by Lemma \ref{BR DEC} that the flux of $J$ through a spacelike surface with respect to the future directed normal will be positive. Moreover, since $Q$ is trace-free and divergence free, $\textrm{Div } J$ will also be positive. To establish a combined energy and integrated decay estimate, we simply have to verify that the surface term on $\scri$ has a definite sign (and check the weights appearing in the various integrals). We shall require some components of $Q$, which are summarised in the following Lemma:

\begin{Lemma}[Components of $Q_{abcd}$]
With respect to the orthonormal basis in which we work, we have
\begin{align*}
Q_{0000} &=  \abs{E_{AB}}^2 + \abs{H_{AB}}^2 + 2 \abs{E_{A\r}}^2 + 2 \abs{H_{A\r}}^2 +\abs{E_{\r\r}}^2 + \abs{H_{\r\r}}^2  \\
Q_{000\r}&= 2\left( E_{AC}\epsilon^{AB}H_{B}{}^C+ E_{A\r}\epsilon^{AB}H_{B\r} \right) \\
Q_{00\r\r}&=  \abs{E_{AB}}^2 + \abs{H_{AB}}^2 - \abs{E_{\r\r}}^2 - \abs{H_{\r\r}}^2 \\
Q_{0\r\r\r} &= 2\left( E_{AC}\epsilon^{AB}H_{B}{}^C- E_{A\r}\epsilon^{AB}H_{B\r} \right) \\
Q_{\r\r\r\r} &= \abs{E_{AB}}^2 + \abs{H_{AB}}^2 - 2 \abs{E_{A\r}}^2 - 2 \abs{H_{A\r}}^2 +\abs{E_{\r\r}}^2 + \abs{H_{\r\r}}^2
\end{align*}
\end{Lemma}

With our definition of the current $J$ above, we apply the divergence theorem, Lemma \ref{Divergence Theorem}. We now verify that all the terms have a definite sign.
\begin{enumerate}[a)]
\item\textbf{Fluxes through $\Sigma_t$} We compute
\be
\int_{\Sigma_t} J_a n^a dS_{\Sigma_t} = \int_{\Sigma_t} Q(e_0, Y, Y, Y)  \frac{r^2}{\sqrt{1+r^2}} dr d\omega .
\ee
Now, defining $\hat{Y} = (1+ \frac{2}{3}r^2)^{-\frac{1}{2}}Y$, we have that 
\be
-g(e_0, \hat{Y}) =  \sqrt{\frac{1+r^2}{1+\frac{2}{3}r^2}} \leq \sqrt{\frac{3}{2}}.
\ee
so by Lemma \ref{BR DEC} we deduce\footnote{We write $f \sim g$ to mean that there exists a numerical constant $C>0$ such that $C^{-1} f \leq g \leq C f$}:
\be
 Q(e_0, \hat{Y}, \hat{Y}, \hat{Y}) \sim  Q_{0000},
\ee
so that
\be
Q(e_0, Y, {Y}, {Y}) \sim \left(1+r^2\right)^{\frac{3}{2}} Q_{0000},
\ee
and hence
\be
\int_{\Sigma_t} J_a n^a dS_{\Sigma_t} \sim \int_{\Sigma_t}  \left(\abs{E_{ab}}^2 + \abs{H_{ab}}^2 \right) r^2 (1+r^2) dr.
\ee
\item \textbf{Bulk term} We have
\be
\textrm{Div} J = 3Q_{abcd}({}^Y\pi)^{ab} Y^c Y^d = \frac{\sqrt{3}}{\sqrt{1+r^2}} Q(e_0, e_0, Y, Y).
\ee
Again applying Lemma \ref{BR DEC} to $\hat{Y}$ and rescaling, we have
\be
\textrm{Div} J \sim  \left(1+r^2\right)^{\frac{1}{2}} Q_{0000}.
\ee
As a result, we have
\be
\int_{S_{[T_1, T_2]}} \textrm{Div }J d\eta \sim \int_{S_{[T_1, T_2]}}  \left(\abs{E_{ab}}^2 + \abs{H_{ab}}^2 \right) r^2 \sqrt{1+r^2} dt dr d\omega.
\ee
\item \textbf{Boundary term at $\scri$} Finally, we consider the flux through surfaces $\tilde{\Sigma}_r$. We have
\be
\int_{\tilde{\Sigma}_r} J_a m^a dS_{\tilde{\Sigma}_r} = \int_{\tilde{\Sigma}_r} Q(e_{\r}, Y, Y, Y) r^2 \sqrt{1+r^2} dt d\omega,
\ee
Now, inserting our expression for $Y$ and expanding, we have
\begin{equation*}
\begin{split}
&Q(e_{\r}, Y, Y, Y) \\
=&\left(1+r^2\right)^{\frac{3}{2}}\left( Q_{\r000} +  \sqrt{3} \frac{r}{\sqrt{1+r^2}} Q_{\r\r00} + \frac{r^2}{1+r^2} Q_{\r\r\r 0} + \frac{1}{3 \sqrt{3}}  \frac{r^3}{\left(1+r^2\right)^{\frac{3}{2}}} Q_{\r\r\r\r} \right)
\end{split}
\end{equation*}
so that
\begin{align*}
\lim_{r \to \infty} Q(e_{\r}, Y, Y, Y) r^2 \sqrt{1+r^2}  = \lim_{r \to \infty} r^6 \left( Q_{\r000} +  \sqrt{3} Q_{\r\r00} +  Q_{\r\r\r 0} + \frac{1}{3 \sqrt{3}} Q_{\r\r\r\r} \right).
\end{align*}
Let us consider the first and third terms on the right hand side\footnote{The factor $\frac{1}{\sqrt{3}}$ appearing in the definition of $Y$ was chosen to arrange a cancellation between these terms.}. We have:
\begin{align*}
Q_{\r000} + Q_{\r\r\r 0} &= 4  E_{AC}\epsilon^{AB}H_{B}{}^C \\
&= 2 \left(E_{AB} -\frac{1}{2} \delta_{AB} E_C{}^C + \epsilon_{(A}{}^CH_{B)C} \right) \left(\hat{E}^{AB} -\frac{1}{2} \delta^{AB} E_D{}^D + \epsilon^{(A}{}_D H^{B)D} \right) \\
&\qquad - 2 \abs{E_{AB}}^2 - 2 \abs{H_{AB}}^2
\end{align*}
so that
\be
\lim_{r \to \infty}  r^6 \left( Q_{\r000}  +  Q_{\r\r\r 0} \right) = -2 \lim_{r \to \infty}  r^6 \left(\abs{E_{AB}}^2 + \abs{H_{AB}}^2  \right)
\ee
where we make use of the boundary condition. Taking this together with the expressions for $Q_{\r\r00}$, $Q_{\r\r\r\r}$ we have:
\begin{align*}
\lim_{r \to \infty} Q(e_{\r}, Y, Y, Y) r^2 \sqrt{1+r^2} = -\lim_{r \to \infty}  r^6 \Bigg [& \left(2 - \frac{10}{3\sqrt{3}}\right)\left(\abs{E_{AB}}^2 + \abs{H_{AB}}^2  \right) + \frac{8}{3\sqrt{3}} \left(\abs{E_{\r\r}}^2 + \abs{H_{\r\r}}^2  \right) \\& + \frac{2}{3\sqrt{3}} \left(\abs{E_{A\r}}^2 + \abs{H_{A\r}}^2\right) \Bigg],
\end{align*}
so that
\be
\lim_{r \to \infty} Q(e_{\r}, Y, Y, Y) r^2 \sqrt{1+r^2} \sim -\lim_{r \to \infty}  r^6\left( \abs{E_{ab}}^2 + \abs{H_{ab}}^2\right)
\ee
and
\be
\int_{\tilde{\Sigma}_\infty^{[T_1, T_2]}} J_a m^a dS_{\tilde{\Sigma}_r} \sim - \int_{\tilde{\Sigma}_\infty^{[T_1, T_2]}} \lim_{r \to \infty}  r^6\left( \abs{E_{ab}}^2 + \abs{H_{ab}}^2\right) dt d\omega.
\ee
\end{enumerate}
Taking all of this together, we arrive at the result.
\end{proof}

As we did in the Maxwell case, we shall now use the structure of the equations to allow us to establish (weighted) integrated decay estimates for all derivatives of the fields $E, H$. To control time derivatives we can simply commute with the Killing field $T$ and apply Proposition  \ref{weyl derivative decay}. To control spatial derivatives we replace the time derivatives by the equations of motion and integrate the resulting cross terms by parts making use also of the constraints equations. The remarkable fact is that in the process we only see spacetime terms with good signs and lower order surface terms that we already control by the estimate before commutation.

We will note the following useful result, which allows the cross term to be integrated by parts:
\begin{Lemma}\label{bianchi cross}
Let $K$ be the vector field
\begin{equation} \label{eq:Yid}
K:=\left(2 r^3 \slashed{\nabla}^C H_{BC} H^B{}_\r - \frac{r^4}{\sqrt{1+r^2}} \abs{H_{B\r}}^2 \right)e_\r - \frac{2 r^2 H_{B\r}}{\sqrt{1+r^2}} \frac{\partial}{\partial r}\left[ r(1+r^2) H^{BC}\right] e_C. 
\end{equation}
If $H$ satisfies the constraint equation \eq{ConsHAB}, we have the identity
\be
\textrm{Div }K = -\frac{2 r^2}{\sqrt{1+r^2}} \slashed{\nabla}^C H_{B\r}  \frac{\partial}{\partial r}\left[ r(1+r^2) H^B{}_C\right] - 2\frac{1+r^2}{r^3} \abs{\partial_r \left(r^3 H_{B\r}\right)}^2 \, .
\ee
\end{Lemma} 
\begin{proof}
See Appendix \ref{bianchi cross proof}.
\end{proof}

\begin{Proposition} [Higher order estimates]\label{prop:elliptic2}
Let $T_2>T_1$. If $E_{ab}, H_{ab}$ solve the spin $2$ equations subject to the dissipative boundary condition \eqref{dcbia} as in Theorem \ref{spin2 well posedness}, then we have
\begin{align*}
& \int_{S_{[T_1, T_2]}} \Bigg\{  \frac{r^3}{1+r^2} \abs{\partial_r(r (1+r^2) H_{BC}) }^2 + \frac{1+r^2}{r} \abs{\partial_r \left(r^3 H_{B\r}\right)}^2 + (H \leftrightarrow E) \Bigg \} dr dt d\omega \\
& +\int_{S_{[T_1, T_2]}} \Bigg\{r^5 |\slashed{\nabla}_B H_{C\r}|^2 +r^5 |\slashed{\nabla}_A H_{BC}|^2  + (H \leftrightarrow E) \Bigg \} dr dt d\omega \\
& \leq C  \int_{\Sigma_{T_1}}\left\{ E_{ab}E^{ab} +H_{ab} H^{ab} + \dot{E}_{ab}\dot{E}^{ab} +\dot{H}_{ab} \dot{H}^{ab} \right\} r^2(1+r^2)  dr d\omega,
\end{align*}
for some $C>0$ independent of $T_1$ and $T_2$.
\end{Proposition}

\begin{proof}
Recall now \eq{EvolEAB2}:
\be
\frac{r}{\sqrt{1+r^2}} \frac{\partial E_{AB}}{\partial t} = \epsilon_{(A}{}^{C}\left[ \frac{1}{\sqrt{1+r^2}}\frac{\partial}{\partial r}\left(r (1+r^2) H_{B)C} \right) - r\slashed{\nabla}_{|C|} H_{B)\r} \right]
\ee
Consider
\begin{align*}
&\epsilon^{AB} \epsilon_{A}{}^{C}\left[ \frac{1}{\sqrt{1+r^2}}\frac{\partial}{\partial r}\left(r (1+r^2) H_{BC} \right) - r\slashed{\nabla}_{C} H_{B\r} \right] \\
&= \frac{1}{\sqrt{1+r^2}}\frac{\partial}{\partial r}\left(r (1+r^2) H{}^B{}_{B} \right) - r\slashed{\nabla}^B H_{B\r} \\
&= -\frac{1}{\sqrt{1+r^2}}\frac{\partial}{\partial r}\left(r^3 \frac{(1+r^2)}{r^2} H_{\r\r} \right) - r\slashed{\nabla}^B H_{B\r} \\
&= - \frac{\sqrt{1+r^2}}{r^2} \partial_r(r^3 H_{\r\r}) -  r\slashed{\nabla}^B H_{B\r} + \frac{2 H_{\r\r}}{\sqrt{1+r^2}} \\
&=  \frac{2 H_{\r\r}}{\sqrt{1+r^2}} \, .
\end{align*}
Now, since for any $2-$tensor $Z$ on $\mathbb S^2$ we have $Z_{[AB]} = \frac{1}{2} \epsilon_{AB}(\epsilon^{CD}Z_{CD})$, we deduce that if the constraints hold then \eq{EvolEAB2} may be re-written:
\be
\frac{r}{\sqrt{1+r^2}} \frac{\partial E_{AB}}{\partial t} = \epsilon_{A}{}^{C}\left[ \frac{1}{\sqrt{1+r^2}}\frac{\partial}{\partial r}\left(r (1+r^2) H_{BC} \right) - r\slashed{\nabla}_{C} H_{B\r} \right] - \frac{\epsilon_{AB}}{\sqrt{1+r^2}} H_{\r\r},
\ee
whence we deduce
\be
X_{AB}:= \frac{1}{\sqrt{1+r^2}}\frac{\partial}{\partial r}\left(r (1+r^2) H_{BC} \right) - r\slashed{\nabla}_{C} H_{B\r} =  \epsilon^A{}_C\frac{r}{\sqrt{1+r^2}} \frac{\partial E_{AB}}{\partial t} + \frac{\delta_{BC}}{\sqrt{1+r^2}} H_{\r\r} .
\ee
Now, using the second equality and applying the estimates in Proposition \ref{weyl derivative decay} for $H$ and the commuted quantity $\dot{E}$, we can verify that
\begin{align*}
\int_{S_{[T_1, T_2]}}& X_{AB}X^{AB}r^3 dt dr d\omega \\& \leq C \int_{S_{[T_1, T_2]}}\left\{ \abs{\dot{E}_{AB}}^2  + \abs{H_{\r\r}}^2 \right\} r^2 \sqrt{1+r^2} dt dr d\omega
 \\& \leq C  \int_{\Sigma_{T_1}}\left\{ E_{ab}E^{ab} +H_{ab} H^{ab} + \dot{E}_{ab}\dot{E}^{ab} +\dot{H}_{ab} \dot{H}^{ab} \right\} r^2(1+r^2)  dr d\omega. \nonumber
\end{align*}
We also have, however,
\begin{align*}
\int_{S_{[T_1, T_2]}}& X_{AB}X^{AB}r^3 dt dr d\omega \\&= \int_{S_{[T_1, T_2]}} \Bigg\{  \frac{r}{1+r^2} \abs{\partial_r(r (1+r^2) H_{BC}) }^2 + r^3 \abs{\slashed{\nabla}_{B} H_{C\r}}^2 \\
& \qquad - \frac{2 r^2}{\sqrt{1+r^2}} \slashed{\nabla}^C H_{B\r}  \frac{\partial}{\partial r}\left[ r(1+r^2) H^B{}_C\right]  \Bigg \} d\eta \\
&= \int_{S_{[T_1, T_2]}} \Bigg\{  \frac{r}{1+r^2} \abs{\partial_r(r (1+r^2) H_{BC}) }^2 + r^3 \abs{\slashed{\nabla}_{B} H_{C\r}}^2 \\
& \qquad +2 \frac{1+r^2}{r^3} \abs{\partial_r \left(r^3 H_{B\r}\right)}^2 \Bigg \} d\eta - \int_{\tilde{\Sigma}_\infty^{[T_1, T_2]}} \abs{r^3 H_{B\r}}^2 dt d\omega,
\end{align*}
where in the last step, we have used the result of Lemma \ref{bianchi cross} to replace the cross term with a good derivative term and a surface term. 

It remains to control the term $|\slashed{\nabla}_A H_{BC}|^2$. Notice that by \eqref{EvolEAr2} and \eqref{ConsHAB}, we have
$$r^2|\epsilon^{BC}\slashed{\nabla}_B H_{CA}|^2+r^2|\slashed{\nabla}^BH_{AB}|^2\leq C\left(\frac{r^2}{1+r^2} |\dot{E}_{A\r}|^2+\frac{1+r^2}{r^4} \left|\frac{\partial}{\partial r}\left(r^3 H_{A\r} \right)\right|^2+\frac{|H_{A\r}|^2}{1+r^2 }\right).$$

On the other hand, since
$$\epsilon^{A}{ }_D\epsilon^{BC}\slashed{\nabla}_BH_{CA}=\slashed{\nabla}^A H_{AD}-\slashed{\nabla}_DH_A{ }^A, $$
we can apply the standard elliptic estimate (see for instance Lemma 2.2.2 in \cite{ChrKla})
$$\int_{\mathbb S^2} |\slashed{\nabla}_AH_{BC}|^2 \leq C\int_{\mathbb S^2} \left(|\epsilon^{BC}\slashed{\nabla}_B H_{CA}|^2+|\slashed{\nabla}^BH_{AB}|^2+|\slashed{\nabla}_A H_B{ }^B|^2\right) $$
to obtain the desired bounds for $|\slashed{\nabla}_A H_{BC}|^2$.

This gives all the desired estimates for the derivatives of $H$. As in the Maxwell case, similar bounds for the derivatives of $E$ can be derived in an identical manner.
\end{proof}

%
%
Finally, much as in the Maxwell case, we apply the Hardy inequalities to establish integrated decay of the non-degenerate energy with the loss of a derivative:
\begin{Theorem}[Full integrated decay]\label{full decay}
Suppose that $W$ is Weyl tensor, satisfying the Bianchi equations with dissipative boundary conditions, as in Theorem \ref{spin2 well posedness}. Then there exists a constant $C>0$, independent of $T$ such that we have
\begin{align}
\int_{S_{[T_1, T_2]}}&\left\{ E_{ab}E^{ab} +H_{ab} H^{ab} \right\} r^2(1+r^2) dt dr d\omega \nonumber \\ 
& \qquad \leq C  \int_{\Sigma_{T_1}}\left\{ E_{ab}E^{ab} +H_{ab} H^{ab}+ \dot{E}_{ab}\dot{E}^{ab} +\dot{H}_{ab} \dot{H}^{ab} \right\} r^2(1+r^2)  dr d\omega. \nonumber
\end{align}
\end{Theorem}
\begin{proof}
Using the result of Proposition \ref{weyl derivative decay}, \ref{prop:elliptic2} with the Hardy estimates of Lemma \ref{wave Hardy} to improve the weights near infinity in the integrated decay estimates, making use of a cut-off in much the same way as for the spin 0 and spin 1 problems, we obtain the desired estimates for $|E_{AB}|$, $|H_{AB}|$, $|E_{A\bar r}|$ and $|H_{A \bar r}|$. Finally, the bounds for $|E_{\bar r\bar r}|$ and $|H_{\bar r\bar r}|$ are obtained trivially using the trace-free condition for $E$ and $H$.
\end{proof}

\subsection{Proof of Corollary \ref{uniform decay} (uniform decay)} \label{sec:proofcor}

In this subsection, we show a uniform decay rate for the solutions to the confomal wave, Maxwell and Bianchi equations, hence proving Corollary \ref{uniform decay}. We will in fact prove the uniform decay estimates for all of the equations at once by showing that this is a consequence of the bounds that we have obtained previously. The result below is a combination of relatively standard ideas (for example, see \cite{Dafermos:2009uq} and Prop 3.1 (a) of \cite{Batkai}), but for completeness we include a direct proof. 

\begin{Lemma}\label{semigroup decay}
Let $\Psi$ be a solution of either the conformal wave, Maxwell or Bianchi equations. Suppose that we are given some positive quantity $\mathcal{E}[\Psi](t)$ depending smoothly on $\Psi$ and its derivatives at some time $t$ which satisfies:
\begin{enumerate}[1.]
\item $\mathcal{E}[\Psi](t)$ is a non-increasing $C^1$ function of $t$,
\item For every $0\leq T_1\leq T_2$, $\mathcal{E}[\Psi](t)$ satisfies the integrated decay estimate:
\be
\int_{T_1}^{T_2} \mathcal{E}[\Psi](t) dt \leq C \left\{ \mathcal{E}[\Psi](T_1) + \mathcal{E}[\p_t \Psi](T_1) \right\},
\ee
for some $C>0$ independent of $T_1$ and $T_2$. 
\end{enumerate}
Then we have the estimate
\be
 \mathcal{E}[\Psi](t) \leq \frac{C_n}{(1+t)^n} \sum_{k=0}^n  \mathcal{E}\left [\left(\p_t\right)^k \Psi \right](0)
\ee
for some constants $C_n>0$ depending only on $n$ and $C$.

\end{Lemma}
\begin{proof}
Let us set
\be
\mathcal{E}^{(k)}(t) =\mathcal{E}\left [\left(\p_t\right)^k \Psi \right](t).
\ee
We calculate
\begin{align*}
(1+t-T_1)\mathcal{E}^{(0)}(t) &=  \mathcal{E}^{(0)}(t) +\int_{T_1}^t \frac{d}{ds}\left((s-T_1) \mathcal{E}^{(0)}(s)\right) ds \\&= \mathcal{E}^{(0)}(t) +\int_{T_1}^t \mathcal{E}^{(0)}(s) + (s-T_1)\dot{\mathcal{E}}^{(0)}(s) ds \\& \leq \mathcal{E}^{(0)}(t) +\int_{T_1}^t \mathcal{E}^{(0)}(s) ds,
\end{align*}
where we have used the monotonicity of $\mathcal E^{(0)}$ to obtain the last inequality.
Now, it follows from the assumptions of the Lemma that
\be
\mathcal{E}^{(0)}(t) + \int_{T_1}^t \mathcal{E}^{(0)}(s)ds   \leq  C_1 \left(\mathcal{E}^{(0)}(T_1)+ \mathcal{E}^{(1)}(T_1) \right),
\ee
which together with the preceding estimate immediately imply
\begin{equation}\label{basecase.T1}
\mathcal{E}^{(0)}(t)  \leq  \frac{C_1}{1+t-T_1} \left(\mathcal{E}^{(0)}(T_1)+ \mathcal{E}^{(1)}(T_1) \right).
\end{equation}
Taking $T_1=0$, this in particular implies the conclusion of the Lemma in the case $n=1$.

To proceed, we induct on $n$. The $n=1$ case has just been established. Suppose now that the statement holds for some $n$. Noticing that the equation commutes with $\p_t$, we use the induction hypothesis for both $\Psi$ and $\p_t\Psi$ to obtain
\be
\mathcal{E}^{(0)}(t) + \mathcal{E}^{(1)}(t) \leq \frac{C_n}{(1+t)^n}\sum_{k=0}^{n+1} \mathcal{E}^{(k)}(0).
\ee
Now, we apply \eqref{basecase.T1} with $T_1=\frac{t}{2}$ to deduce
\begin{align*}
\mathcal{E}^{(0)}(t)& \leq \frac{C_1}{1+\frac{t}{2}}\left( \mathcal{E}^{(0)}\left (\frac{t}{2} \right) + \mathcal{E}^{(1)}\left (\frac{t}{2} \right) \right) \\
& \leq \frac{C_1C_n}{(1+\frac{t}{2})^{n+1}}\sum_{k=0}^{n+1} \mathcal{E}^{(k)}(0) \\
&\leq \frac{C_1 C_n 2^n}{(1+t)^{n+1}}\sum_{k=0}^{n+1} \mathcal{E}^{(k)}(0),
\end{align*}
whence the result follows.
\end{proof}

\begin{proof}[Proof of Corollary \ref{uniform decay}]
For the conformal wave and Maxwell equations, Corollary \ref{uniform decay} follows immediately by applying Lemma \ref{semigroup decay} to the quantity
\be
\mathcal{E}[\Psi](t) = \int_{\Sigma_T}  \frac{\varepsilon \left[\Psi\right]}{\sqrt{1+r^2}}r^2 dr d\omega
\ee
which is monotone decreasing and satisfies an integrated decay statement with loss of one derivative. For the Bianchi equations, this quantity is not monotone decreasing (merely bounded by a constant times its initial value). We can circumvent this by applying Lemma \ref{semigroup decay} to the quantity:
\be
\mathcal{E}[W](t) =  \int_{\Sigma_t} Q(e_0, Y, Y, Y)  \frac{r^2}{\sqrt{1+r^2}} dr d\omega,
\ee
which is monotone decreasing and satisfies an integrated decay statement with loss of one derivative. Noting that
\be
\mathcal{E}[W](t) \sim \int_{\Sigma_T}  \frac{\varepsilon \left[W\right]}{\sqrt{1+r^2}}r^2 dr d\omega,
\ee
we are done.
\end{proof}

\subsection{Proof of Theorem \ref{theo:gb}: Gaussian beams} \label{sec:proofgb}
It is noteworthy that in the first instance, for all of the integrated decay estimates we obtained above the $r$-weight near infinity is weaker than that for the energy estimate. In particular, in order to show a uniform-in-time decay estimate, we needed to lose a derivative. In this section, we show that without any loss, there cannot be any uniform decay statements for the conformal wave equation. Moreover, an integrated decay estimate with no degeneration in the $r$-weight does not hold.

In order to show this, we will construct approximate solutions to the conformally coupled wave equation for a time interval $[0,T]$ with an arbitrarily small loss in energy. We will in fact first construct a Gaussian beam solution on the Einstein cylinder and make use of the fact that (one half of) the Einstein cylinder is conformally equivalent to the AdS spacetime to obtain an approximate solution to the conformally coupled wave equation on AdS.

In the following, we will first study the null geodesics on the Einstein cylinder. We then construct Gaussian beam approximate solutions to the wave equation on the Einstein cylinder. Such construction is standard and in particular we follow closely Sbierski's geometric approach \cite{Sbierski:2013mva} (see also \cite{Arnaud, Ralston}). After that we return to the AdS case and build solutions that have an arbitrarily small loss in energy.


\subsubsection{Geodesics in the Einstein cylinder}

We consider the spacetime $(\mathcal M_E, g_E)$, where $\mathcal M_E$ is diffeomorphic to $\mathbb R\times \mathbb S^3$ and the metric $g_E$ is given by
\begin{equation}\label{gE.def}
g_E=-dt^2+ d\psi^2+\sin^2\psi(d\th^2+\sin^2\th d\phi^2).
\end{equation}

We will slightly abuse notation and denote the subsets of $\mathcal M_E$ with notations similar to that for the AdS spacetime. More precisely, we will take
$$\Sigma_T:=\{(t,\psi,\th,\phi): t=T\},$$
$$S_{[T_1,T_2]}:=\{(t,\psi,\th,\phi): T_1\leq t\leq T_2\}.$$

Take null geodesics $\gamma:(-\infty,\infty)\to \mathcal M_E$ in the equatorial plane $\{\th = \f\pi 2\}$. In coordinates, we express $\gamma$ as
$$(t,\psi,\th,\phi)=(T(s),\Psi(s),\f\pi 2,\Phi(s)).$$
Since $\f\rd{\rd t}$ and $\f\rd{\rd \phi}$ are Killing vector fields, $E$ and $L$ defined as
$$\dot{T}=E,\quad \sin^2\Psi \dot{\Phi}=L$$
are both conserved quantities. Here, and below, we use the convention that $\dot{ }$ denotes a derivative in $s$. We require from now on that $0\leq |L|<E$. The geodesic equation therefore reduces to the ODE
\begin{equation}\label{geod.ODE}
\dot{\Psi}=\sqrt{E^2-\f{L^2}{\sin^2\Psi}}.
\end{equation}
Solving \eqref{geod.ODE} with the condition that $\Psi$ achieves its minimum at $s=0$ and $T(0)=\Phi(0)=0$, we have
$$T(s)=Es,$$
$$\sin\Psi(s)=\f{\sqrt{(E^2-L^2)\sin^2(Es)+L^2}}{E},$$
and
$$\Phi(s)=\int_0^s\f{L E^2}{(E^2-L^2)\sin^2(Es')+L^2} ds'.$$
In order to invert the sine function to recover $\Psi$, we will use the convention that for $Es\in [2k\pi-\f{\pi}{2},2k\pi+\f{\pi}{2})$, and $k\in \mathbb Z$, we require $\Psi(s) \in [0,\f\pi 2]$; while for $E s\in [(2k+1)\pi-\f{\pi}{2},(2k+1)\pi+\f{\pi}{2})$, and $k\in \mathbb Z$, we require $\Psi(s) \in [\f\pi 2, \pi]$. Notice that this choice of the inverse of the sine function gives rise to a smooth null geodesic.

Moreover, direct computations show that
$$\dot{T}(s)=E,\quad\dot{\Psi}(s)=\f{E\sqrt{E^2-L^2}\sin(Es)}{\sqrt{(E^2-L^2)\sin^2(Es)+L^2}},\quad\dot{\Phi}(s)=\f{L E^2}{(E^2-L^2)\sin^2(Es)+L^2}.$$

\subsubsection{Constructing the Gaussian beam}\label{Construct.GB}

Given a null geodesic $\gamma$ on $(\mathcal M_E, g_E)$ as above, we follow the construction in Sbierski \cite{Sbierski:2013mva} to obtain an approximate solution to the wave equation on $(\mathcal M_E, g_E)$ which is localised near $\gamma$ and has energy close to that of $\gamma$. We first define the phase function $\varphi$ and its first and second partial derivatives on $\gamma$ and then construct the function $\varphi$ in a neighbourhood of $\gamma$. We also define the amplitude $a$ on $\gamma$. More precisely, on $\gamma$, we require the following conditions:

\begin{enumerate}
\item $\varphi(\gamma(s))=0 $
\item $d\varphi(\gamma(s))=\dot\gamma(s)_{\flat} $
\item The matrix $M_{\mu\nu}:= \rd_\mu\rd_\nu\varphi(\gamma(s))$ is a symmetric matrix satisfying the ODE
$$\f d{ds} M=-A-BM-MB^T-MCM,$$
where $A$, $B$, $C$ are matrices given by
$$A_{\kappa\rho}=\f 12(\rd_\kappa\rd_\rho g^{\mu\nu})\rd_\mu\varphi\rd_\nu\varphi,$$
$$B_{\kappa\rho}=\rd_\kappa g^{\rho\mu}\rd_\mu\varphi,$$
$$C_{\kappa\rho}=g^{\kappa\rho},$$
and obeying the initial conditions
\begin{enumerate}
\item $M(0)$ is symmetric;
\item $M(0)_{\mu\nu}\dot{\gamma}^\nu=(\dot{\rd_\mu\varphi})(0)$;
\item $\Im (M(0)_{\mu\nu})dx^\mu\left.\right|_{\gamma(0)}\otimes dx^\nu\left.\right|_{\gamma(0)} $ is positive definite on a three dimensional subspace of $T_{\gamma(0)}M$ that is transversal to $\dot\gamma$.
\end{enumerate}
\item $a(\gamma(0))\neq 0$ and $a$ satisfies the ODE 
\be
2 \mbox{grad} \varphi(a)+\Box\varphi\cdot a=0
\ee 
along $\gamma$.
\end{enumerate}

The results in \cite{Sbierski:2013mva} ensure that $\left. \varphi \right|_{\gamma}, \left. \p_\mu \varphi \right|_{\gamma}, \left. \p_\mu \p_\nu\varphi \right|_{\gamma}, \left. a \right|_\gamma$ can be constructed satisfying these conditions.  We then let $\varphi$ to be a smooth extension of $\varphi$ away from $\gamma$ compatible with these derivatives. Likewise $a_{\mathcal N}$ is defined to be an extension of $\left. a\right|_{\gamma}$ as constructed above. Moreover, we require $a_{\mathcal N}$ to be compactly supported in a (small) tubular neighborhood $\mathcal N$ of the null geodesic $\gamma$.

We define the energy for a function on $\mathcal M_E$ by
$$\hat{E}_t({w}):=\frac 12\int_{\Sigma_t} \big((\rd_t w)^2+(\rd_\psi w)^2+\f{(\rd_\th {w})^2}{\sin^2\psi}+\f{(\rd_\phi {w})^2}{\sin^2\psi\sin^2\th}\big) \sin^2\psi d\psi d\omega.$$
Here, as elsewhere, $d\omega= \sin\th d\th d\phi$ is the volume form of the round unit sphere. We also associate the geodesic $\gamma$ with a conserved energy
$$E(\gamma)=-g_E(\dot\gamma,\rd_t).$$
Notice that this agrees with the convention $E=\dot{T}$ used in the previous subsection.

The main result\footnote{Translated into our notation, the result in \cite{Sbierski:2013mva} requires the following bounds on the geometry of $(\mathcal M_E, g_E)$: 
$$-C\leq g(\rd_t,\rd_t)\leq c<0,\,|\nabla \rd_t(\rd_t,\rd_t)|+|\nabla \rd_t(\rd_t,e_i)|+|\nabla \rd_t(e_i,e_j)|\leq C,$$
where $e_i$ is an orthonormal frame on the $\mathbb S^3$ slice. These estimates are obviously satisfied in our setting.} of Sbierski  regarding the approximate solution constructed above is the following theorem\footnote{Regarding point (4) in the theorem below, the original work of Sbierski gives a more general characterization of the energy of Gaussian beams in terms of the the energy of geodesics on general Lorentzian manifold. Since in our special setting, the energy of a geodesic is conserved, we will not record the most general result but will refer the readers to \cite{Sbierski:2013mva} for details.} (see Theorems 2.1 and 2.36 in \cite{Sbierski:2013mva}):

\begin{Theorem}\label{sbierski.thm}
Given a geodesic $\gamma$ parametrized by $E$ and $L$ as above, let 
$${w}_{E,L,\lambda,\mathcal N}=a_{\mathcal N} e^{i\lambda\varphi},$$
where $a_{\mathcal N}$ and $\varphi$ are defined as above. Then ${w}_{\lambda,\mathcal N}$ obeys the following conditions:
\begin{enumerate}
\item $\|\Box w_{E,L,\lambda,\mathcal N}\|_{L^2(\mathcal S_{[0,T]})}\leq C(T)$;
\item $\hat{E}_0(w_{E,L,\lambda,\mathcal N})\to \infty \mbox{ as }\lambda\to \infty$;
\item $w_{E,L,\lambda,\mathcal N}$ is supported in $\mathcal N$, a tubular neighborhood of $\gamma$;
\item Fix $\mu>0$ and normalize the initial energy of $w_{E,L,\lambda,\mathcal N}$ by
$$ \tilde{w}_{E,L,\lambda,\mathcal N}:=\f{w_{E,L,\lambda,\mathcal N}}{\sqrt{\hat{E}_0(w_{E,L,\lambda,\mathcal N})}}\cdot E(\gamma).$$
Then for $\mathcal N$ a sufficiently small neighborhood of $\gamma$ and $\lambda$ sufficiently large, the following bound holds:
$$\sup_{t\in [0,T]}\abs{\hat{E}_t(\tilde{w}_{E,L,\lambda,\mathcal N})-E(\gamma)}<\mu.$$
\end{enumerate}
\end{Theorem}
We also need another fact regarding the second derivatives of $\varphi$ which is a consequence of the construction in \cite{Sbierski:2013mva} (see (2.14) in \cite{Sbierski:2013mva}):
\begin{Lemma}\label{d2phi}
$\Im(\varphi\left.\right|_\gamma)=\Im(\nabla\varphi\left.\right|_\gamma)=0.$ Moreover, $\Im(\nabla\nabla\varphi\left.\right|_\gamma)$ is positive definite on a $3$-dimensional subspace transversal to $\dot\gamma$.
\end{Lemma}
The fact that $a_{\mathcal N}$ is independent of $\lambda$ together with Lemma \ref{d2phi} imply that the Gaussian beam approximate solution constructed above has bounded $L^2$ norm independent of $\lambda$. We record this bound in the following lemma:
\begin{Lemma}\label{L2.bd}
Let $w_{E,L,\lambda,\mathcal N}$ be as in Theorem \ref{sbierski.thm}. The following bound holds:
$$\|w_{E,L,\lambda,\mathcal N}\|_{L^2(S_{[0,T]})}\leq C(T).$$
\end{Lemma}

\subsubsection{The conformal transformation}

Once we have constructed the Gaussian beam for the wave equation on the Einstein cylinder, it is rather straightforward to construct the necessary sequence of functions using the conformal invariance of the operator 
$$L=\Box_g-\f 16 R(g),$$
where $R(g)$ is the scalar curvature of the metric $g$. 

We first set up some notations. We will be considering $\mathcal M_E$ restricted to $\psi\leq \f\pi 2$ as a manifold with boundary diffeomorphic to $\mathbb R\times \mathbb S^3_h$. The interior of this manifold will also be identified with $\mathcal M_{AdS}$ via identifying the coordinate functions $(t,\th,\phi)$, as well as
$$\tan\psi=r.$$
It is easy to see that $g_E$ and $g_{AdS}$ are conformal. More precisely, $g_E$ as before can be written as
\begin{equation*}
g_E=-dt^2+ d\psi^2+\sin^2\psi(d\th^2+\sin^2\th d\phi^2);
\end{equation*}
while in the $(t,\psi,\th,\phi)$ coordinate system, $g_{AdS}$ takes the following form
\begin{equation*}
g_{AdS}=\frac{1}{\cos^2 \psi} \big(-dt^2+ d\psi^2+\sin^2\psi(d\th^2+\sin^2\th d\phi^2)\big).
\end{equation*}

We now proceed to the construction of the approximate solution to the conformally coupled wave equation on AdS using the conformal invariance of $L$. More precisely, we have
\begin{Lemma}\label{conf.trans}
Let $w$ be a function on $\mathcal M_E$ restricted to $\psi\leq \f\pi 2$. Then we have
$$\f{1}{(1+r^2)^{\f32}}( \Box_{g_E} w- w) = \Box_{g_{AdS}} \f w{\sqrt{1+r^2}}+2\f w{\sqrt{1+r^2}}.$$
\end{Lemma}
\begin{proof}
This can be verified by an explicit computation.
\end{proof}

\begin{figure}[t]
\begin{minipage}[b]{0.45\linewidth}
\centering
\includegraphics[width=7cm]{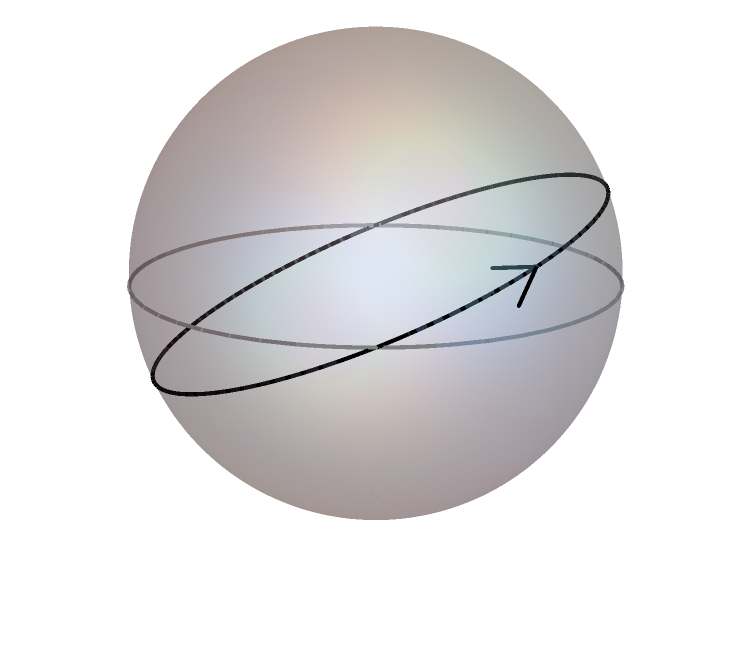}
\end{minipage}
\hspace{0.8cm}
\begin{minipage}[b]{0.45\linewidth}
\centering
\includegraphics[width=7cm]{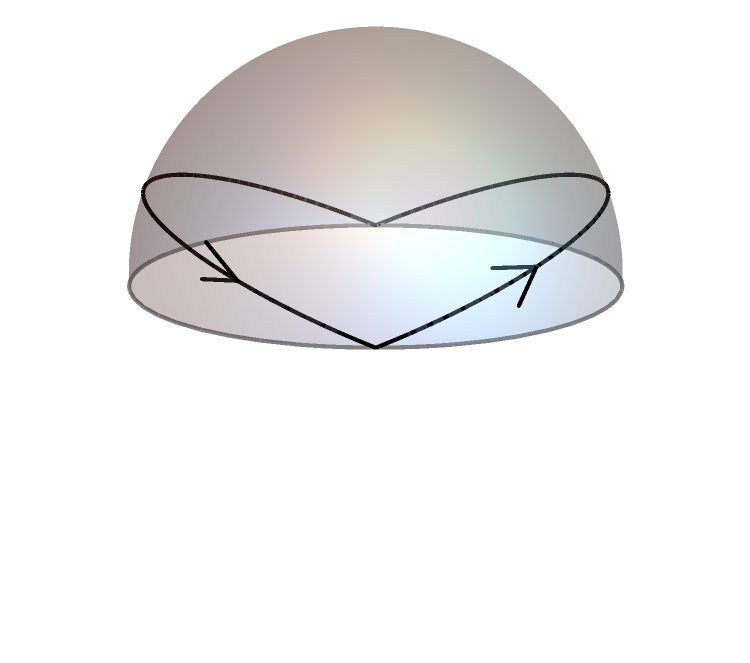}
\end{minipage}
\caption{An typical null geodesic $\gamma$ shown in black projected to a surface of constant $\theta$ in the optical geometry of the Einstein cylinder (left). The corresponding curve for the anti-de Sitter spacetime, after reflections at $\scri$ (right).} \label{fig1}
\end{figure}

We are now ready to construct the approximate solution. To heuristically explain the construction, consider Figure \ref{fig1}. On the left we sketch the curve $\gamma$ along which our approximate solution is concentrated in the Einstein cylinder. To allow this to be plotted, we project onto a surface of constant $t$, $\theta$ which carries a spherical geometry and can be visualised via its embedding in $\R^3$. We actually wish to construct an approximate solution on anti-de Sitter, which is conformal to one half of the Einstein cylinder. To do so, we restrict our attention to one hemisphere and arrange that whenever the curve $\gamma$ strikes the equator of the Einstein cylinder, it is reflected. Each time this occurs, we shall arrange that the Gaussian beam is attenuated by a factor which depends on the angle of incidence (which in turn depends on $E$ and $L$). The more shallow the reflection, the weaker the attenuation. Crucially, the time, $t$, between reflections is independent of the angle that we choose. As a result, by taking the angle of incidence to be sufficiently small, we can arrange that an arbitrarily small fraction of the initial energy is lost at the boundary over any given time interval. 

To be more precise, let us fix $T\geq 0$. We will construct an approximate solution for $t\in [0,T]$. Take $N\in \N$ be the smallest integer such that $T<\f{(4N+3)\pi}{2}$. We then construct an approximate solution for $t\in (-\f {\pi} 2, \f{(4N+3)\pi}{2})$ starting from the function $w_{E,L,\lambda,\mathcal N}$ constructed previously. This will correspond to a Gaussian beam which strikes the boundary $\sim 2N$ times between $t=0$ and $t=T$. Notice that the geodesic $\gamma$ has the property that it lies in the hemisphere $\{\psi<\f\pi 2\}$ for $t\in (2k\pi-\f\pi 2,2k\pi+\f\pi 2)$ (for $k\in \mathbb Z$) and that it lies in the other hemisphere, i.e., $\{\psi>\f \pi 2\}$, for $t\in ((2k+1)\pi-\f\pi 2,(2k+1)\pi+\f\pi 2)$ (for $k\in \mathbb Z$). Therefore, we will assume without loss of generality that the neighborhood $\mathcal N$ has been taken sufficiently small such that it lies entirely in $\{\psi<\f\pi 2\}$ for $t\in (2k\pi-\f\pi 4,2k\pi+\f\pi 4)$ (and entirely in $\{\psi>\f\pi 2\}$ for $t\in ((2k+1)\pi-\f\pi 4,(2k+1)\pi+\f\pi 4)$), where $k\in \mathbb Z$.

Now for $(t,\psi,\th,\phi)$ in $\mathcal M_E$ restricted to $\psi\leq \f\pi 2$, we define for $t\in [0,T]$
\begin{equation}\label{u.tot.def}
\begin{split}
u_{E,L,\lambda}(t,\psi,\th,\phi)
=&\sum_{k=0}^{N}R^{2k}\chi_{t\in (2k\pi-\f{3\pi}{4},2k\pi+\f{3\pi}{4}]}\f{w_{E,L,\lambda,\mathcal N}(t,\psi,\th,\phi)\left.\right|_{\{\psi\leq \f\pi 2\}}}{\sqrt{1+r^2}}\\
+&\sum_{k=0}^N R^{2k+1}\chi_{t\in ((2k+1)\pi-\f{3\pi} 4,(2k+1)\pi+\f {3\pi} 4]}\f{w_{E,L,\lambda,\mathcal N}(t,\pi-\psi,\th,\phi)\left.\right|_{\{\psi\leq \f\pi 2\}}}{\sqrt{1+r^2}},
\end{split}
\end{equation}
where $R$ is taken to be $R=-\f{E-\sqrt{E^2-L^2}}{E+\sqrt{E^2-L^2}}$ and $\chi$ is the indicator function. Notice in particular that when the time cutoff function is $0$ in the first term (resp. in the second term), the support of $w_{E,L,\lambda,\mathcal N}$ is entirely in $\{\psi >\f\pi 2\}$ (resp. $\{\psi <\f\pi 2\}$). We also depict this in Figure \ref{cutoff}, where we denote 
$$\chi_1(t)=\sum_{k=0}^{N}\chi_{t\in (2k\pi-\f{3\pi}{4},2k\pi+\f{3\pi}{4}]}, \quad \chi_2(t)=\sum_{k=0}^N\chi_{t\in ((2k+1)\pi-\f{3\pi} 4,(2k+1)\pi+\f {3\pi} 4]}.$$

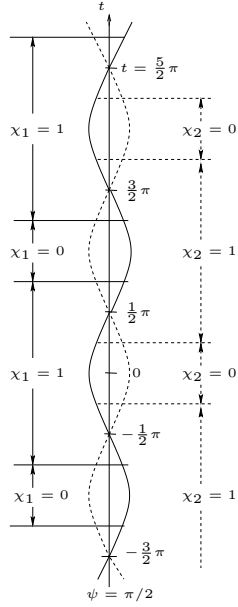
\begin{figure}
\begin{picture}(0,0)%
\includegraphics{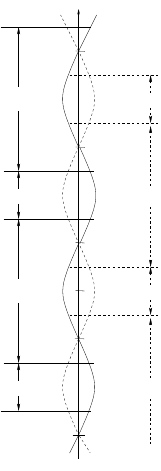}%
\end{picture}%
\setlength{\unitlength}{1263sp}%
\begingroup\makeatletter\ifx\SetFigFont\undefined%
\gdef\SetFigFont#1#2#3#4#5{%
  \reset@font\fontsize{#1}{#2pt}%
  \fontfamily{#3}\fontseries{#4}\fontshape{#5}%
  \selectfont}%
\fi\endgroup%
\begin{picture}(4002,11770)(3736,-9959)
\put(5926,464){\makebox(0,0)[lb]{\smash{{\SetFigFont{5}{6.0}{\rmdefault}{\mddefault}{\updefault}{\color[rgb]{0,0,0}$t=\frac{5}{2}\pi$}%
}}}}
\put(3751,-736){\makebox(0,0)[lb]{\smash{{\SetFigFont{5}{6.0}{\rmdefault}{\mddefault}{\updefault}{\color[rgb]{0,0,0}$\chi_1=1$}%
}}}}
\put(3751,-3136){\makebox(0,0)[lb]{\smash{{\SetFigFont{5}{6.0}{\rmdefault}{\mddefault}{\updefault}{\color[rgb]{0,0,0}$\chi_1=0$}%
}}}}
\put(3751,-5536){\makebox(0,0)[lb]{\smash{{\SetFigFont{5}{6.0}{\rmdefault}{\mddefault}{\updefault}{\color[rgb]{0,0,0}$\chi_1=1$}%
}}}}
\put(5251,-9886){\makebox(0,0)[lb]{\smash{{\SetFigFont{5}{6.0}{\rmdefault}{\mddefault}{\updefault}{\color[rgb]{0,0,0}$\psi=\pi/2$}%
}}}}
\put(3826,-7936){\makebox(0,0)[lb]{\smash{{\SetFigFont{5}{6.0}{\familydefault}{\mddefault}{\updefault}{\color[rgb]{0,0,0}$\chi_1=0$}%
}}}}
\put(6001,-1936){\makebox(0,0)[lb]{\smash{{\SetFigFont{5}{6.0}{\rmdefault}{\mddefault}{\updefault}{\color[rgb]{0,0,0}$\frac{3}{2}\pi$}%
}}}}
\put(6001,-4336){\makebox(0,0)[lb]{\smash{{\SetFigFont{5}{6.0}{\rmdefault}{\mddefault}{\updefault}{\color[rgb]{0,0,0}$\frac{1}{2}\pi$}%
}}}}
\put(6151,-5536){\makebox(0,0)[lb]{\smash{{\SetFigFont{5}{6.0}{\rmdefault}{\mddefault}{\updefault}{\color[rgb]{0,0,0}$0$}%
}}}}
\put(5926,-6736){\makebox(0,0)[lb]{\smash{{\SetFigFont{5}{6.0}{\rmdefault}{\mddefault}{\updefault}{\color[rgb]{0,0,0}$-\frac{1}{2}\pi$}%
}}}}
\put(6001,-9061){\makebox(0,0)[lb]{\smash{{\SetFigFont{5}{6.0}{\rmdefault}{\mddefault}{\updefault}{\color[rgb]{0,0,0}$-\frac{3}{2}\pi$}%
}}}}
\put(7126,-3136){\makebox(0,0)[lb]{\smash{{\SetFigFont{5}{6.0}{\rmdefault}{\mddefault}{\updefault}{\color[rgb]{0,0,0}$\chi_2=1$}%
}}}}
\put(7126,-5536){\makebox(0,0)[lb]{\smash{{\SetFigFont{5}{6.0}{\rmdefault}{\mddefault}{\updefault}{\color[rgb]{0,0,0}$\chi_2=0$}%
}}}}
\put(7126,-7936){\makebox(0,0)[lb]{\smash{{\SetFigFont{5}{6.0}{\rmdefault}{\mddefault}{\updefault}{\color[rgb]{0,0,0}$\chi_2=1$}%
}}}}
\put(7126,-736){\makebox(0,0)[lb]{\smash{{\SetFigFont{5}{6.0}{\rmdefault}{\mddefault}{\updefault}{\color[rgb]{0,0,0}$\chi_2=0$}%
}}}}
\put(5476,1664){\makebox(0,0)[lb]{\smash{{\SetFigFont{5}{6.0}{\rmdefault}{\mddefault}{\updefault}{\color[rgb]{0,0,0}$t$}%
}}}}
\end{picture}%
\caption{The cutoff functions.}
\label{cutoff}
\end{figure}

The definition above is such that in the time interval $t\in (-\f{3\pi} 4,\f {3\pi} 4]$, we take $w_{E,L,\lambda,\mathcal N}$ constructed previously, restrict it to $\psi\leq \f\pi 2$ and rescale it by $\f 1{\sqrt{1+r^2}}$. Then on the time interval $t\in (\pi-\f{3\pi} 4,\pi+\f {3\pi} 4]$, we take the part of $w_{E,L,\lambda,\mathcal N}$ that is supported in $\psi\geq \f\pi 2$, reflect it across the $\psi=\f\pi 2$ hypersurface, rescale by a factor $\f 1{\sqrt{1+r^2}}$ and then multiply by the factor $R$. As we will see later, the factor $R$ is chosen so that the boundary conditions are approximately satisfied. We then continue this successively, taking parts of the solutions in $\psi\leq \f\pi 2$ and $\psi \geq \f\pi 2$, reflecting when appropriate, and multiplying by factors of $R$'s.

\begin{Lemma}\label{lemma.u}
The function $u_{E,L,\lambda}$ defined as in \eq{u.tot.def} has the following properties:
\begin{align}
&E_{0}[u_{E,L,\lambda}] \to \infty \qquad \textrm{ as } \lambda \to \infty; \\
&\int_{S_{[0,T]}} (\Box_{g_{AdS}} u_{E,L,\lambda}+2u_{E,L,\lambda})^2 r^2(1+r^2) \,dr\,d\omega\,dt \leq C(T); \\
&\int_{\bar{\Sigma}_\infty^{[0, T]}} (\rd_t(ru_{E,L,\lambda}) +r^2\rd_r(ru_{E,L,\lambda}))^2  \,d\omega\, dt \leq C(T).
\end{align}
\end{Lemma}
\begin{proof}
First note that $\hat{E}_0[w_{E,L,\lambda}]\leq CE_{0}[u_{E,L,\lambda}]$, so the first claim follows from Theorem \ref{sbierski.thm}. Next recall that we have assumed that the neighbourhood $\mathcal N$ has been taken sufficiently small so that for each of the summands, the indicator function $\chi$ is constant on the support of $w_{E,L,\lambda,\mathcal N}$ when restricted to a hemisphere. It therefore suffices to consider only the contributions from $w_{E,L,\lambda,\mathcal N}$ (as no derivatives fall on the indicator functions). We will write 
$$u_{E,L,\lambda}=u_{E,L,\lambda,1}+u_{E,L,\lambda,2},$$
where $u_{E,L,\lambda,1}$ is the first sum in \eqref{u.tot.def} and $u_{E,L,\lambda,2}$ is the second sum in \eqref{u.tot.def}. For $u_{E,L,\lambda,1}$, we apply Lemma \ref{conf.trans} to get
\begin{equation*}
\begin{split}
&\int_{S_{[0,T]}} (\Box_{g_{AdS}} u_{E,L,\lambda,1}+2u_{E,L,\lambda,1})^2 r^2(1+r^2) \,dr\,d\omega\,dt\\
\leq &\int_{S_{[0,T]}} \f{1}{(1+r^2)^3}( \Box_{g_E} w_{E,L,\lambda,\mathcal N}- w_{E,L,\lambda,\mathcal N})^2 r^2(1+r^2) \,dr\,d\omega\,dt\\
=&\int_{S_{[0,T]}} ( \Box_{g_E} w_{E,L,\lambda,\mathcal N}- w_{E,L,\lambda,\mathcal N})^2 \sin^2\psi \,d\psi\,d\omega\,dt\\
\leq &C(T).
\end{split}
\end{equation*}
In the last line, we have used the estimates in Theorem \ref{sbierski.thm} and Lemma \ref{L2.bd}. By a straightforward identification of the two hemispheres in $\mathcal M_E$, we can prove a similar bound for $u_{E,L,\lambda,2}$:
\begin{equation*}
\begin{split}
&\int_{S_{[0,T]}} (\Box_{g_{AdS}} u_{E,L,\lambda,2}+2u_{E,L,\lambda,2})^2 r^2(1+r^2) \,dr\,d\omega\,dt\\
\leq &\int_{S_{[0,T]}} ( \Box_{g_E} w_{E,L,\lambda,\mathcal N}- w_{E,L,\lambda,\mathcal N})^2 \sin^2\psi \,d\psi\,d\omega\,dt
\leq C(T).
\end{split}
\end{equation*}
This concludes the proof of the second claim.

We now show that the boundary terms are appropriately bounded in $L^2$.  By construction, there are only contributions to the boundary terms in a neighborhood of $\f 1\pi(t-\f \pi 2)\in \mathbb N$. Since there are at most $O(N)=O(T)$ such contributions, it suffices to show that one of them is bounded. We will look at the boundary contribution near $t=\f{\pi}{2}$, which takes the form
\begin{equation}\label{main.bdry.term}
\begin{split}
&\lim_{r\to \infty}(r\rd_t u_{E,L,\lambda}+r^2\rd_r(r u_{E,L,\lambda}))(t,r,\th,\phi)\\
= &\lim_{\psi\to \f\pi 2^-} (\rd_t w_{E,L,\lambda,\mathcal N}+\rd_\psi w_{E,L,\lambda,\mathcal N})(t,\psi,\th,\phi)\\
&+\lim_{\psi\to \f\pi 2^-}R(\rd_tw_{E,L,\lambda,\mathcal N}+\rd_\psi w_{E,L,\lambda,\mathcal N})(t,\pi-\psi,\th,\phi)\\
=&i\lambda a_{\mathcal N}((\rd_t\varphi+\rd_\psi\varphi)+R(\rd_t\varphi-\rd_\psi\varphi)) e^{i\lambda \varphi(t,\psi=\f\pi 2,\th,\phi)}\\
&+((\rd_t a_{\mathcal N}+\rd_\psi a_{\mathcal N})+R(\rd_t a_{\mathcal N}-\rd_\psi a_{\mathcal N})) e^{i\lambda \varphi(t,\psi=\f\pi 2,\th,\phi)}.
\end{split}
\end{equation}
The latter term is clearly bounded pointwise independent of $\lambda$. The first term has a factor of $\lambda$ and we will show that it is nevertheless bounded in $L^2$ since by the choice of $R$, $((\rd_t\varphi+\rd_\psi\varphi)+R(\rd_t\varphi-\rd_\psi\varphi))$ vanishes on $\gamma$. More precisely, by points (1), (2) in Section \ref{Construct.GB} and Lemma \ref{d2phi}, we have
$$\Im\varphi(t,\psi=\f\pi 2,\th,\phi)\geq \alpha ((t-\f{\pi}{2})^2+(\th-\f \pi 2)^2+(\phi-\int_0^{\f{E\pi}{2}}\f{L E^2}{(E^2-L^2)\sin^2(Es')+L^2} ds')^2)$$
for some $\alpha>0$. We further claim that 
$$((\rd_t\varphi+\rd_\psi\varphi)+R(\rd_t\varphi-\rd_\psi\varphi))\left.\right|_{(t=\f\pi 2,\,\psi=\f{\pi}{2},\, \th=\f\pi 2,\, \phi= \int_0^{\f{E\pi}{2}}\f{L E^2}{(E^2-L^2)\sin^2(Es')+L^2})}=0.$$
This follows from the fact that $d\varphi=\dot\gamma_\flat$ on $\gamma$ and the choice of $R$. More precisely, we have
\begin{equation*}
\begin{split}
&((\rd_t\varphi+\rd_\psi\varphi)+R(\rd_t\varphi-\rd_\psi\varphi))\left.\right|_{(t=\f\pi 2,\,\psi=\f{\pi}{2},\, \th=\f\pi 2,\, \phi= \int_0^{\f{E\pi}{2}}\f{L E^2}{(E^2-L^2)\sin^2(Es')+L^2})}\\
=& (-E+\sqrt{E^2-L^2}+R(-E-\sqrt{E^2-L^2}))\\
=&0.
\end{split}
\end{equation*}
The desired $L^2$ bound on the boundary then follows from Lemma \ref{gaussian lemma}.
\end{proof}

\begin{Lemma}\label{gaussian lemma}
Suppose that $f$ is a function defined on $\{\psi=\f\pi 2\}$ which vanishes to order $0$ at $(t=\f\pi 2,\, \th=\f\pi 2,\, \phi= \int_0^{\f{E\pi}{2}}\f{L E^2}{(E^2-L^2)\sin^2(Es')+L^2})$. Then
\be
\int d\omega\, dt |f e^{i \lambda \varphi}|^2 \leq C_{f, \varphi} \lambda^{- \frac{5}{2}}.
\ee
\end{Lemma}
\begin{proof}
In order to simplify notation, we define
$$\phi_0:=\int_0^{\f{E\pi}{2}}\f{L E^2}{(E^2-L^2)\sin^2(Es')+L^2} ds'.$$
The statement that $f$ vanishes to order $0$ is equivalent to
\be
\abs{f} \leq C \left( (t-\f{\pi}{2})^2+(\th-\f \pi 2)^2+(\phi-\phi_0)^2 \right)^{\frac{1}{2}}.
\ee
Thus
\begin{align*}
&\int d\omega\, dt |f e^{i \lambda \varphi}|^2 
= \int d\omega\, dt \abs{f}^2 e^{-2\lambda  \Im \varphi} \\
 \leq & C \int_0^\pi \sin \th d\th \int_0^{2\pi} d\phi \int_0^\pi dt \left( (t-\f{\pi}{2})^2+(\th-\f \pi 2)^2+(\phi-\phi_0)^2 \right) e^{-2 \lambda \alpha  ((t-\f{\pi}{2})^2+(\th-\f \pi 2)^2+(\phi-\phi_0)^2)} \\
\leq &C \int_{\R^3} \abs{\bm{x}}^2 e^{-2 \alpha \lambda \abs{\bm{x}}^2} dx
\leq C \lambda^{ -\frac{5}{2}},
\end{align*}
where in the last line we simply scale $\lambda$ out of the integral.
\end{proof}

\subsubsection{Building a true solution}

Given the approximate solution constructed above, we are now ready to build a true solution to the homogeneous conformally coupled wave equation with dissipative boundary condition. To this end, we need a strengthening of Theorem \ref{wave conservation} which includes inhomogeneous terms:
\begin{Theorem}\label{inhom wave conservation}
Let $u$ be a solution of the inhomogeneous conformally coupled wave equation
\begin{equation}\label{wave.eqn.inho}
\Box_{g_{AdS}} u+2u = f \qquad \textrm{ in AdS}
\end{equation}
with finite (renormalized) energy initial data
\be
E_{T_1}(u)<\infty
\ee
and subject to inhomogeneous dissipative boundary conditions
\be
\frac{\partial (r u)}{\partial t} + r^2\frac{\partial (r u)}{\partial r} \to  g, \qquad \textrm{ as }r\to \infty.
\ee
Then we have for any $T_1<t<T_2$:
\be
E_{t}[u] \leq C_{T_1, T_2}\left(  E_{T_1}[u] + \int_{\tilde{\Sigma}_\infty^{[T_1, T_2]}} g^2 d\omega dt + \int_{S_{[T_1, T_2}]} f^2 r^2 (1+r^2) dr d\omega dt\right).
\ee
\end{Theorem}
\begin{proof}
We have by \eq{twisted em divergence} that
\be
\textrm{Div} \left({}^T\J \right)= \left(\Box_{AdS} u+2u  \right) \p_t u = f \p_t u.
\ee
Integrating this over $S_{[T_1, T_2]}$ and applying the divergence theorem we pick up terms on the left hand side from $\Sigma_{T_1}, \Sigma_{T_2}$ and $\tilde{\Sigma}_\infty^{[T_1, T_2]}$. A straightforward calculation shows
\be
\int_{\Sigma_t} {}^T\J_a n^a dS_{\Sigma_t} = E_t[u]
\ee
We also find
\begin{align*}
\int_{\tilde{\Sigma}_r^{[T_1, T_2]}} {}^T\J_a m^a dS_{\tilde{\Sigma}_r} &= \int_{\mathbb S^2} r^2(1+r^2) \left( \p_t u\right ) (\tilde{\p}_r u) d\omega dt \\
&= -\frac{1}{2} \int_{\mathbb S^2} \Bigg[ \left(r^2  (\p_t u)^2 + r^2(1+r^2)^2 (\tilde{\p}_r u)^2\right) \\
&\qquad - \left\{ \p_t( r u) + r (1+r^2) (\tilde{\p}_r u) \right\}^2 \Bigg] d\omega dt.
\end{align*}
As $r \to \infty$, the term in braces may be replaced with $g$, so we have
\be
\lim_{r \to \infty} \int_{\tilde{\Sigma}_r^{[T_1, T_2]}} {}^T\J_a m^a dS_{\tilde{\Sigma}_r} =-\frac{1}{2} \int_{\tilde{\Sigma}_\infty^{[T_1, T_2]}} \left(r^2  (\p_t u)^2 + r^6 (\tilde{\p}_r u)^2 - g^2\right) d\omega dt.
\ee
Applying Lemma \ref{Divergence Theorem}, we have:
\be
E_{t}[u] - E_{T_1}[u] = \frac{1}{2} \int_{\tilde{\Sigma}_\infty^{[T_1, t]}} \left(-r^2  (\p_t u)^2 - r^6 (\tilde{\p}_r u)^2 + g^2\right) d\omega dt + \int_{S_{[T_1, t]}} f \p_t u r^2 dr dt d\omega
\ee
so that
\begin{align} \nonumber
\sup_{t\in[T_1, T_2]} E_{t}[u] &\leq E_{T_1}[u]  +  \frac{1}{2} \int_{\tilde{\Sigma}_\infty^{[T_1, T_2]}}  g^2 d\omega dt + \frac{1}{4(T_2 - T_1)} \int_{S_{[T_1, T_2]}} \frac{(\p_t u)^2}{1+r^2} r^2 dr dt d\omega  \\ &\quad + (T_2 - T_1)  \int_{S_{[T_1, T_2}]} f^2 r^2 (1+r^2) dr d\omega dt. \label{inhom estimate}
\end{align}
Now
\be
 \frac{1}{4(T_2 - T_1)}  \int_{S_{[T_1, T_2]}} \frac{(\p_t u)^2}{1+r^2} r^2 dr dt d\omega  \leq \frac{1}{4} \sup_{t \in [T_1, T_2]} \int_{\Sigma_t} \frac{(\p_t u)^2}{1+r^2} r^2 dr d\omega \leq \frac{1}{2} \sup_{t\in[T_1, T_2]} E_{t}[u].
\ee
Applying this estimate to \eq{inhom estimate} and absorbing the energy term on the left hand side, we are done.
\end{proof}

After taking $\lambda$ to be large, we now construct true solutions of the wave equation with dissipative boundary conditions which only have a small loss of energy. To do this, we define $\tilde{u}_{E,L,\lambda}$ to be $u_{E,L\lambda}$ (defined by \eqref{u.tot.def} in the previous subsection) multiplied by a constant factor in such a way that
\be
E_0(\tilde{u}_{E,L,\lambda}) = 1.
\ee
Let $U_{E,L,\lambda}$ solve the homogeneous initial-boundary value problem:
\be
\Box_{g_{AdS}} U_{E,L,\lambda}+2U_{E,L,\lambda} = 0
\ee
subject to dissipative boundary conditions
\be
\frac{\partial (r U_{E,L,\lambda})}{\partial t} + r^2 \frac{\partial (rU_{E,L,\lambda})}{\partial r} \to 0, \qquad \textrm{ as }r\to \infty,
\ee
and initial conditions
\be
\left. U_{E,L,\lambda} \right|_{t=0} = \left. \tilde{u}_{E,L,\lambda}\right|_{t=0}, \quad \left. \frac{\partial U_{E,L,\lambda}}{\partial t} \right|_{t=0} = \left. \frac{\partial \tilde{u}_{E,L,\lambda}}{\partial t}\right|_{t=0}.
\ee
We obtain the following theorem:
\begin{Theorem}
Fix $E>0$ and $0\leq |L|<E$. Fix also $T, \epsilon>0$. There exists a solution $U_{E,L,\lambda}$ of the homogeneous conformally coupled wave equation with dissipative boundary conditions such that $U_{E,L,\lambda}$ has energy $1$ at time $0$, and
\be
\inf_{t\in[0,T]}E_t(U_{E,L,\lambda}) \geq \left( \frac{E-\sqrt{E^2-L^2}}{E+\sqrt{E^2-L^2}}\right)^{CT}-\epsilon,
\ee
where $C>0$ is some universal constant.
\end{Theorem}
\begin{proof}
Consider $U_{E,L,\lambda}$ as defined above, which is a solution to the homogeneous conformal wave equation. By Lemma \ref{lemma.u}, as $\lambda\to \infty$, we have
$$\int_{S_{[0,T]}} (\Box_{g_{AdS}} \tilde u_{E,L,\lambda}+2\tilde u_{E,L,\lambda})^2 r^2(1+r^2) \,dr\,d\omega\,dt \to 0 $$
and
$$\int_{\bar{\Sigma}_\infty^{[0, T]}} (\rd_t(r\tilde u_{E,L,\lambda}) +r^2\rd_r(r\tilde u_{E,L,\lambda}))^2  \,d\omega\, dt \to 0.$$
Therefore, after taking $\lambda$ to be sufficiently large and applying Theorem \ref{inhom wave conservation} to $U_{E,L,\lambda}-\tilde{u}_{E,L,\lambda}$,
we can assume
$$\sup_{t\in[0,T]} E_t[U_{E,L,\lambda}-\tilde{u}_{E,L,\lambda}]<\epsilon.$$
On the other hand, by the construction of $u_{E,L,\lambda}$, we have
$$\inf_{t\in[0,T]} E_t[\tilde{u}_{E,L,\lambda}]\geq \left( \frac{E-\sqrt{E^2-L^2}}{E+\sqrt{E^2-L^2}}\right)^{2N+1}\geq \left( \frac{E-\sqrt{E^2-L^2}}{E+\sqrt{E^2-L^2}}\right)^{CT}$$
for some $C>0$. The results follow straightforwardly.
\end{proof}

In particular, by taking $|L|$ sufficiently close to $E$ and $\epsilon$ sufficiently small, this implies that on the time interval $[0,T]$, the loss of energy can be arbitrarily small. This also implies that any uniform integrated decay estimates without loss do not hold:
\begin{Corollary} \label{cor:ert}
There exists no constant $C>0$, such that
\be
\int_0^\infty E_t[u] dt \leq C E_0[u]
\ee
holds for every solution $u$ to the conformal wave equation with finite initial energy subject to dissipative boundary conditions.

Similarly, there exists no continuous positive function $f: \R_+ \to \R_+$, such that $f(t) \to 0$ as $t \to \infty$ and
\be
E_t[u] \leq f(t) E_0[u],
\ee
holds for every solution $u$ to the conformal wave equation with finite initial energy subject to dissipative boundary conditions.
\end{Corollary}

\begin{Remark}
We are grateful to an anonymous referee, who points out that our construction above can in fact be adapted to establish the stronger statement that the degeneracy in $r$ for the integrated decay rate established above is in fact optimal. That is to say for any $\delta>0$, there can exist no constant $C$ such that Proposition \ref{Morawetz wave} holds with the weight $\frac{r^2}{\sqrt{1+r^2}}$ replaced by $\frac{r^{2+\delta}}{\sqrt{1+r^2}}$. This in particular suggests that the trapping phenomenon present here is different to the normally hyperbolic trapping observed, for example, at the photon sphere of the Schwarzschild black hole.
\end{Remark}

\section{Generalizations} \label{sec:gen}
\subsection{Alternative boundary conditions} \label{boundary conditions section}

\subsubsection{Conformal Wave}
We assumed that our solution $u$ satisfied the boundary condition
\be
\frac{\partial (r u)}{\partial t} + r^2\frac{\partial (r u)}{\partial r} \to  0, \qquad \textrm{ as }r\to \infty \, .
\ee
We can also consider the boundary conditions
\be
\frac{\partial (r u)}{\partial t} + \beta(\omega)  r^2\frac{\partial (r u)}{\partial r} \to  0, \qquad \textrm{ as }r\to \infty \, .
\ee
where $\beta :\mathbb S^2 \to \R$ is a smooth function satisfying the uniform positivity condition:
\be
\beta(\omega) \geq \kappa^2
\ee
for some $\kappa>0$ for all $\omega \in \mathbb S^2$. This can be treated exactly as above, with the same results although the constants in the various estimates will now depend on $\beta$. One could also imagine adding some small multiples of tangential derivatives of $u$ to the boundary condition. This can also be handled by the methods above, but this will require combining the energy and integrated decay estimates.

\subsubsection{Maxwell}
The boundary conditions that we assumed, 
\be
r^2 \left(E_A + \epsilon_A{}^{B} H_B\right) \to 0, \qquad \textrm{ as } r \to \infty,
\ee
can also be generalised without materially affecting the results. In particular, we could choose as boundary conditions
\be
r^2 \left(E_A + \beta_{AB}(\omega) \epsilon^{BC} H_C\right) \to 0, \qquad \textrm{ as } r \to \infty
\ee
where the symmetric matrix valued function $\beta:\mathbb S^2 \to M(2\times 2)$ satisfies a uniform positivity bound:
\begin{equation} \label{leo}
\beta_{AB}(\omega) \xi^A \xi^B \geq \kappa^2 \abs{\xi}^2,
\end{equation}
for some $\kappa>0$ and for all $\xi \in \R^2$, $\omega \in \mathbb S^2$. In particular, our results hold for any Leontovic boundary condition \cite[\S 87]{LandauLifshitz8} satisfying (\ref{leo}). Again, one could also permit other components of $E$, $H$ to appear in the boundary condition with small coefficients and this can be handled by combining the energy and integrated decay estimates.

\subsubsection{Bianchi}
Recall that we are considering boundary conditions of the form:
\begin{equation} \label{optd}
r^3 \left( \hat{E}_{AB} + \epsilon_{(A}{}^{C} \hat{H}_{B) C}\right) \to 0, \qquad \textrm{ as } r \to \infty.
\end{equation}
To generalise these, let us introduce a $4-$tensor on $\mathbb S^2$, $\beta_{ABCD}(\omega)$, symmetric on its first and last pairs of indices and also under interchange of the first and last pair of indices:
\be
\beta_{ABCD}= \beta_{(AB)CD}= \beta_{AB(CD)}= \beta_{CDAB}
\ee
we also require that $\beta$ is trace free on its first (or last) indices, i.e. $\beta^A{}_{ABC}=0$. In other words, $\beta$ represents a symmetric bilinear form on the space of symmetric trace-free tensors at each point of  $\mathbb S^2$. We can consider boundary conditions
\ben{new bcs}
r^3 \left( \beta_{AB}{}^{CD} \hat{E}_{CD} +\epsilon_{(A}{}^{C} \hat{H}_{B) C}\right) \to 0, \qquad \textrm{ as } r \to \infty.
\een
Provided that  $\beta_{ABCD}$ is uniformly close to the canonical inner product on symmetric trace-free tensors, our methods will apply. More concretely, there is a $\delta>0$, which could be explicitly calculated, such that if
\be
(1+\delta) \abs{{\Xi}}^2 \geq \beta_{ABCD}(\omega) \Xi^{AB} \Xi^{CD} \geq (1-\delta) \abs{{\Xi}}^2 
\ee
holds for any symmetric traceless $\Xi$ and $\omega \in \mathbb S^2$, then our results hold for the boundary conditions \eq{new bcs}.

Of course, this does not imply that boundary conditions which don't satisfy this inequality lead to growth, simply that our approach breaks down when the boundary conditions are too far from the ``optimally dissipative'' ones (\ref{optd}).

\subsection{The Dirichlet problem for (B)}  \label{sec:Dirichlet}
Let us briefly discuss the Dirichlet problem for the Weyl tensor and its relation to the Dirichlet problem from the point of view of metric perturbations. In the latter, one wishes to fix, to linear order, the conformal class of the metric at infinity. Fixing the conformal class to be that of the unperturbed anti-de Sitter spacetime is equivalent to requiring that the Cotton-York tensor of the perturbed boundary metric vanishes. One may verify that the Cotton-York tensor of the boundary metric vanishes if and only if
\be
\abs{r^3\hat{H}_{AB}} + \abs{r^3E_{A\r}}+  \abs{r^3H_{\r\r}} \to 0\qquad \textrm{as }r\to \infty \, ,
\ee
as can be established by considering the metric in Fefferman-Graham coordinates. Let us illustrate in what way fixing the conformal class of the metric on the boundary is a more restrictive condition than fixing Dirichlet-conditions on the Weyl tensor, $\abs{r^3\hat{H}_{AB}} \rightarrow 0$.

It is possible to extract from the Bianchi equations a symmetric hyperbolic system on $\scri$ involving only ${r^3E_{A\r}},  {r^3H_{\r\r}}$, where $r^3\hat{H}_{AB}$ appears as a source term. Using this system it is possible to show that if ${r^3E_{A\r}}$ and ${r^3H_{\r\r}}$ vanish at infinity for the initial data then this condition propagates. Moreover, it is easy to see how to construct a large class initial data satisfying this vanishing condition at infinity, as well as a large class \emph{not} satisfying it illustrating that fixing the conformal class of the metric on the boundary is a more restrictive condition than purely fixing Dirichlet-conditions on the Weyl tensor. 


In conclusion, for initial data satisfying the vanishing condition, we may return to the estimate \eq{naive energy} and establish directly that solutions of the Bianchi system representing a linearised gravitational perturbation fixing the conformal class of the boundary metric are bounded. This is in accordance with the results of \cite{WaldIshi}, in which it is shown that metric perturbations obeying the linearised Einstein equations can be decomposed into components which separately obey wave equations admitting a conserved energy.

\subsection{The relation to the Teukolsky equations} \label{Teukolsky section}
We finally contrast our result with an alternative approach to study the spin 2 equations on AdS, which has a large tradition in the asymptotically flat context and relies on certain curvature components satisfying decoupled wave equations. As we will see below, however, in the AdS context this approach merely obscures the geometric nature of the problem and does not provide any obvious simplification as the resulting decoupled equations couple via the boundary conditions (and moreover do not admit a conserved energy). 

To decouple the spin $2$ equations, we take the standard $\theta, \phi$ coordinates for the spherical directions\footnote{The $(r, t,\theta,\phi)$ do not quite cover AdS, and so some care must be taken at the axis. For the purposes of this section, we shall ignore this difficulty.}, and choose as basis $e_1 = r^{-1} \p_\theta$, $e_2 = (r \sin \theta)^{-1} \p_\phi$. We then write:
\be
\Psi^{\pm} = E_{11} - E_{22} \mp2 H_{12} \mp i \left( H_{11} - H_{22} \pm2 E_{12}\right) \, .
\ee
These quantities obey the Teukolsky equations:
\begin{align}  \nonumber
0&= -\frac{r^2}{1+r^2} \p_t^2 \Psi^\pm  \pm 4 \frac{r}{1+r^2} \p_t \Psi^\pm + \frac{1+r^2}{r^3} \p_r \left(\frac{r^4} {1+r^2} \p_r \left[r (1+r^2)\Psi^\pm \right] \right) \\& \qquad + \frac{1}{\sin \theta} \p_\theta \left( \sin \theta \p_\theta \Psi^\pm \right) +\frac{1}{\sin^2\theta} \p_\phi^2 \Psi^\pm - 4 i \frac{\cos \theta}{\sin^2 \theta} \p_\phi \Psi^\pm \label{teukolsky}\\
&\qquad  -\left(\frac{4 }{\sin^2 \theta}-2\right)  \Psi^\pm \nonumber \,.
\end{align}
Once $\Psi^\pm$ have been found, the other components of $W$ can be recovered by solving an elliptic system coupled to the symmetric hyperbolic system in the boundary discussed in Section \ref{sec:Dirichlet}. It might appear that one can simply study the decoupled equations for $\Psi^\pm$ separately. Unfortunately, in general, the correct boundary conditions couple the equations. 

Let us see what the appropriate boundary conditions to impose on $\Psi^\pm$ are in order to fix the conformal class of the boundary metric. Clearly the vanishing of $\hat{H}_{AB}$ at $\scri$ is equivalent to the condition
\ben{Dirichlet bc}
\abs{r^3(\Psi^+ - \Psi^-  )} \to 0\qquad \textrm{as }r\to \infty.
\een
This however only gives us one condition for two equations. For the other condition we use the fact that in the context of the Dirichlet problem (fixing the conformal class discussed in Section \ref{sec:Dirichlet}) we know that ${r^3E_{A\r}},  {r^3H_{\r\r}}$ vanish on the boundary.\footnote{In the general case, their trace on the boundary can de determined by solving transport equations as outlined in Section \ref{sec:Dirichlet}} Inserting this into the Bianchi equations we can derive that
\be
\abs{r^2 \frac{\p}{\p r} \left( r^3\hat{E}_{AB}\right) } \to 0\qquad \textrm{as }r\to \infty.
\ee
This gives us a Neumann condition for the Teukolsky equations:
\ben{Neumann bc}
\abs{r^2 \frac{\p}{\p r} \left[ r^3(\Psi^+ + \Psi^-  )\right] } \to 0\qquad \textrm{as }r\to \infty.
\een
We should of course not be surprised that the two Teukolsky equations couple at the boundary. The two scalar functions $\Psi^\pm$ represent the outgoing and ingoing radiative degrees of freedom. Since the Dirichlet boundary conditions are reflecting, one should of course expect that the two components couple at the boundary. The pair of equations \eq{teukolsky} with the boundary conditions \eq{Dirichlet bc}, \eq{Neumann bc} forms a well posed system of wave equations, as may be seen for instance with the methods of \cite{Warnick:2012fi}. In any case, there seems to be no advantage in studying this coupled system of wave equations over the first order techniques of this paper.

We remark that the Teukolsky approach was recently used in \cite{Dias:2013sdc} to consider perturbations of the Kerr-AdS family of metrics (which includes anti-de Sitter for $a=m=0$). In this paper, the Teukolsky equation is separated and a boundary condition (preserving the conformal class of the metric at infinity) is proposed for the radial part of each mode of $\Psi^{\pm}$ \emph{separately}. This appears to contradict our discussion above. When one examines equations (3.9-12) of \cite{Dias:2013sdc} one sees spectral parameters appearing up to fourth order\footnote{In the presence of rotation things are even worse, as the square root in the boundary conditions implies that the operator on the boundary is \emph{non-local}.}. Returning to a physical space picture, these will appear as fourth order operators on the boundary. Accordingly, it is far from clear whether these boundary conditions can be meaningfully interpreted as giving boundary conditions for a dynamical evolution problem. Indeed, our heuristic argument for the coupling strongly suggests that the conditions of \cite{Dias:2013sdc} understood as boundary conditions for a dynamical problem \emph{cannot} give rise to a well posed evolution. That is not to say that these boundary conditions are not suitable for calculating (quasi)normal modes: any such mode will certainly obey these conditions, providing a useful trick to simplify such computations.

\section{Appendix}

\subsection{Proof of Lemma \ref{elliptic wave}} \label{elliptic wave proof}
\begin{proof}
Let us first consider the first two terms of $K$. We set
\be
K_1 :=\p_r \left(r^2 \p_r \left(u\sqrt{1+r^2} \right) \right) (\slashed{\nabla}^Au)\,  e_A -  (r \slashed{\nabla}_A u) \tilde{\p}_r \left(r \slashed{\nabla}^Au\right)\, e_\r \, .
\ee
Calculating with the expression for the divergence of a vector field \eq{vector divergence}, we find
\begin{align*}
\textrm{Div }K_1 &= \slashed{\nabla}_A \left[ \p_r \left(r^2 \p_r \left(u\sqrt{1+r^2} \right) \right) (\slashed{\nabla}^Au) \right] - \frac{1}{r^2}\frac{\p}{\p r} \left[ r^2 \sqrt{1+r^2} (r\slashed{\nabla}_A u) \tilde{\p}_r \left(r\slashed{\nabla}^Au\right)\right] \\
&= \p_r \left(r^2 \p_r \left(u\sqrt{1+r^2} \right) \right)\left[\slashed{\nabla}_A \slashed{\nabla}^A u\right] - \sqrt{1+r^2} \abs{\tilde{\p}_r\left( r\slashed{\nabla} u\right)}^2 \\
&\quad  +  \slashed{\nabla}_A \left[ \p_r \left(r^2 \p_r \left(u\sqrt{1+r^2} \right) \right)\right] (\slashed{\nabla}^Au) - \frac{\sqrt{1+r^2}}{r^2} (r\slashed{\nabla}_A u) \frac{\p}{\p r} \left[r^2 \tilde{\p}_r \left(r\slashed{\nabla}^Au\right)\right] \\
&= \p_r \left(r^2 \p_r \left(u\sqrt{1+r^2} \right) \right)\left[\slashed{\nabla}_A \slashed{\nabla}^A u\right] - \sqrt{1+r^2} \abs{\tilde{\p}_r\left( r\slashed{\nabla} u\right)}^2 \\
&\quad + \left[ \p_r \left(r^2 \p_r \left( r\slashed{\nabla}_A u\sqrt{1+r^2} \right) \right)\right] \frac{(\slashed{\nabla}^Au)}{r} - \frac{\sqrt{1+r^2}}{r^2} (r\slashed{\nabla}_A u) \frac{\p}{\p r} \left[r^2 \tilde{\p}_r \left(r\slashed{\nabla}^Au\right)\right]\\
&= \p_r \left(r^2 \p_r \left(u\sqrt{1+r^2} \right) \right)\left[\slashed{\nabla}_A \slashed{\nabla}^A u\right] - \sqrt{1+r^2} \abs{\tilde{\p}_r\left( r\slashed{\nabla} u\right)}^2 \\
&\quad +\frac{r}{\sqrt{1+r^2}} \left( r\slashed{\nabla}_A u\right)\tilde{\p}_r\left( r\slashed{\nabla}^A u\right) \, .
\end{align*}
Finally, taking
\be
K_2:= - \frac{r }{2(1+r^2)} \abs{r \slashed{\nabla} u}^2 e_\r \, ,
\ee
we calculate
\begin{align*}
\textrm{Div }K_2 &= - \frac{1}{r^2}\frac{\p}{\p r} \left[ r^2 \sqrt{1+r^2} \frac{r }{2(1+r^2)} \abs{r \slashed{\nabla} u}^2\right] \\
&= - \frac{r}{\sqrt{1+r^2}} \left( r\slashed{\nabla}_A u\right)\tilde{\p}_r\left( r\slashed{\nabla}^A u\right) - \frac{1+r^2}{2 r^2}  \abs{r \slashed{\nabla} u}^2 \frac{\p}{\p r}\left( \frac{r^3}{(1+r^2)^{\frac{3}{2}}}\right)
\\
&= - \frac{r}{\sqrt{1+r^2}} \left( r\slashed{\nabla}_A u\right)\tilde{\p}_r\left( r\slashed{\nabla}^A u\right) - \frac{3 \abs{ r\slashed{\nabla} u}^2}{2 (1+r^2)^{\frac{3}{2}}} \, .
\end{align*}
Adding these two contributions, we arrive at the result.
\end{proof}

\subsection{Proof of Lemma \ref{maxwell cross}}\label{maxwell cross proof}
Let us introduce
\be
\alpha := -\frac{2}{\sqrt{1+r^2}}  \partial_r \left( r \sqrt{1+r^2} H_A \right) r\slashed{\nabla}^A H_\r.
\ee
From the expression \eq{vector divergence} for the divergence of a vector field, we can quickly establish that, owing to cancellation between the mixed partial derivatives, we have
\begin{align*}
\textrm{Div } K &= 2\sqrt{1+r^2}  \frac{\partial}{\partial r} \left(\frac{r^2 H_\r}{\sqrt{1+r^2}} \right) \slashed{\nabla}^B H_B -  2 \frac{\partial}{\partial r} \left( r \sqrt{1+r^2} H_B \right) \frac{r \slashed{\nabla}^B H_\r}{\sqrt{1+r^2}} \\
&- \frac{1}{r^2}  \frac{\partial}{\partial r} \left(\frac{r^5}{ \sqrt{1+r^2}} H_\r^2 \right)
\end{align*}
so that
\begin{align*}
\alpha &= \textrm{Div } K -   2\sqrt{1+r^2}  \frac{\partial}{\partial r} \left(\frac{r^2 H_\r}{\sqrt{1+r^2}} \right) \slashed{\nabla}^B H_B + \frac{1}{r^2}  \frac{\partial}{\partial r} \left(\frac{r^5}{ \sqrt{1+r^2}} H_\r^2 \right) \\
&= \textrm{Div } K + 2 \frac{1+r^2}{r^2} \frac{\partial}{\partial r} \left(\frac{r^2 H_\r}{\sqrt{1+r^2}} \right) \frac{\partial}{\partial r} \left(r^2 H_\r \right) + \frac{1}{r^2}  \frac{\partial}{\partial r} \left(\frac{r^5}{ \sqrt{1+r^2}} H_\r^2 \right) \\
&= \textrm{Div } K + 2 \frac{\sqrt{1+r^2}}{r^2}\abs{ \frac{\partial}{\partial r} \left(r^2 H_\r \right) }^2 - \frac{2 r}{\sqrt{1+r^2}} H_\r  \frac{\partial}{\partial r} \left(r^2 H_\r \right) \\
&\quad + \frac{2 r}{\sqrt{1+r^2}} H_\r  \frac{\partial}{\partial r} \left(r^2 H_\r \right) + \frac{r^2}{(1+r^2)^{\frac{3}{2}}} \abs{H_r}^2 \\
&= 2\frac{\sqrt{1+r^2}}{r^2} \abs{\frac{\partial }{\partial r}\left ( r^2H_\r \right)}^2 + \frac{r^2}{\left(1+r^2\right)^{\frac{3}{2}}} \abs{H_\r}^2 + \textrm{Div } K \, .
\end{align*}
Here we have used the constraint equation to pass from the first to the second line. This completes the proof. 

\subsection{Proof of Lemma \ref{bianchi cross}}\label{bianchi cross proof}
This is a straightforward calculation with the formula for the divergence of a vector field \eq{vector divergence}. We find
\begin{align*}
\textrm{Div }K &= \frac{1}{r^2}\frac{\partial }{\partial r} \left(2 r^5 \sqrt{1+r^2} \slashed{\nabla}^C H_{BC} H^B{}_\r - \abs{r^3 H_{B\r}}^2 \right) \\
& \quad - \slashed{\nabla}^C\left( \frac{2 r^2 H_{B\r}}{\sqrt{1+r^2}} \frac{\partial}{\partial r}\left[ r(1+r^2) H^{B}{}_{C}\right]  \right) \\
&= \frac{2 r H_{B\r}}{\sqrt{1+r^2}} \frac{\partial}{\partial r}\left[ r^2(1+r^2) \slashed{\nabla}^C H^{B}{}_{C}\right]  + 2(1+r^2) \slashed{\nabla}^C H_{BC} \frac{\partial}{\partial r}\left[ \frac{r^3}{\sqrt{1+r^2}}  H^B{}_\r\right] \\
&\quad - \frac{2 r H_{B\r}}{\sqrt{1+r^2}} \frac{\partial}{\partial r}\left[ r^2(1+r^2) \slashed{\nabla}^C H^{B}{}_{C}\right] - \frac{2 r^2}{\sqrt{1+r^2}} \slashed{\nabla}^C H_{B\r}  \frac{\partial}{\partial r}\left[ r(1+r^2) H^B{}_C\right] \\
&\quad - 2 r H_{B\r} \frac{\partial}{\partial r} \left(r^3 H^B{}_\r \right) \\
&= -\frac{2(1+r^2)^{\frac{3}{2}} }{r^3}  \frac{\partial}{\partial r}\left( r^3 H_{B\r}\right) \left[  \frac{1}{\sqrt{1+r^2}} \frac{\partial}{\partial r}\left( r^3 H^B{}_\r\right) - \frac{r^4}{(1+r^2)^{\frac{3}{2}}} H^B{}_\r \right] \\
&\quad- \frac{2 r^2}{\sqrt{1+r^2}} \slashed{\nabla}^C H_{B\r}  \frac{\partial}{\partial r}\left[ r(1+r^2) H^B{}_C\right]  - 2 r H_{B\r} \frac{\partial}{\partial r} \left(r^3 H^B{}_\r \right) \\
&= -\frac{2 r^2}{\sqrt{1+r^2}} \slashed{\nabla}^C H_{B\r}  \frac{\partial}{\partial r}\left[ r(1+r^2) H^B{}_C\right] - 2\frac{1+r^2}{r^3} \abs{\partial_r \left(r^3 H_{B\r}\right)}^2 \, .
\end{align*}
Here, we have used the constraint equation in passing from the second equality to the third by replacing $\slashed{\nabla}^C H_{BC}$ with a term involving $\partial_r(r^3 H_{B\r})$.

\providecommand{\href}[2]{#2}\begingroup\raggedright\endgroup


\end{document}